\documentclass[12pt]{article}
\usepackage{formatting}
\usepackage{shortcuts}

\begin{document}

\begin{titlepage}

\begin{flushright}BRX-TH-6730\end{flushright}
\begin{center}
\phantom{ }
\vspace{-.2cm}

 {\bf \Large{Twirled Perfect Tensor Networks:}}\\[.15cm]{\bf \Large{Computationally covariant holographic tensor networks}}

\vskip 1cm
\large{Gurbir Arora${}^{1}$, Matthew Headrick${}^{2,3}$, Albion Lawrence${}^{2}$,\\[.15cm] Martin Sasieta${}^{4}$, Brian Swingle${}^{2}$, and Connor Wolfe${}^{2}$}
\vskip 1cm

\small{${}^{1}$ \textit{College of Engineering, Cornell University}}
\vskip -.05cm
{\textit{Ithaca, New York 14850, USA}}
\vskip 0.3cm
\small{${}^{2}$ \textit{Martin Fisher School of Physics, Brandeis University}}
\vskip -.05cm
{\textit{Waltham, Massachusetts 02453, USA}}
\vskip 0.3cm
\small{${}^{3}$ \textit{Institut des Hautes Etudes Scientifiques}}
\vskip -.05cm
{\textit{91893 Bures-sur-Yvette France}}
\vskip 0.3cm
\small{${}^{4}$ \textit{Leinweber Institute for Theoretical Physics and Department of Physics,}
\vskip -.05cm
{\textit{ University of California, Berkeley, California 94720, USA}}}\\[0.5cm]

\small{\texttt{E-mail}:
\href{mailto:g.arora@cornell.edu}{g.arora@cornell.edu},
\href{mailto:mph@brandeis.edu}{mph@brandeis.edu},
\href{mailto:albion@brandeis.edu}{albion@brandeis.edu}\\
\href{mailto:msasieta@berkeley.edu}{msasieta@berkeley.edu},
\href{mailto:cwolfe@brandeis.edu}{cwolfe@brandeis.edu},
\href{mailto:bswingle@brandeis.edu}{bswingle@brandeis.edu}}\\[.8cm]

\end{center}

\begin{abstract}
We define a novel class of tensor networks motivated by the Python’s Lunch Conjecture (PLC) in local tensor network models of the black hole interior. We start from the observation that, for extended black brane states with short-range correlations, the PLC predicts a complexity that is smaller than the upper bound for generic short-range correlated states.  We argue that the PLC makes implicit assumptions about the fine structure of the relevant tensor networks modeling gravity that render them non-generic.  We demonstrate this explicitly in random tensor network models of the python’s lunch, where the exponential complexity is not generally controlled by the PLC exponent. We trace the difference with the PLC to a lack of ``computational covariance'' in random tensor networks: while the PLC is motivated by an ability to arbitrarily decompose space into low-complexity units provided certain basic rules are followed, we show that random tensor networks do not generically have this property. We propose another class of tensor networks built from what we call ``twirled perfect tensors'' that do satisfy the computational covariance property and have a complexity bounded by the PLC value. We still find a discrete limitation from local postselection that appears to be absent in gravity. Moreover, we show that this class of tensor networks combines desirable holographic features of perfect tensor networks and random tensor networks, for example, it obeys a lattice Ryu-Takayanagi formula for arbitrary boundary subregions. Though motivated by holography, these tensor networks provide a flexible framework with potential applications beyond quantum gravity.
\end{abstract}
\vspace{-.1cm}
{\small \hspace{.9cm} A video abstract is available at \url{https://youtu.be/AX-cSk-iI6o}.}

\end{titlepage}

\setcounter{tocdepth}{2}
\tableofcontents
\noindent\makebox[\linewidth]{\rule{\textwidth}{0.4 pt}}

\section{Introduction}

AdS/CFT provides our most controlled framework for non-perturbative quantum gravity: a concrete holographic correspondence in which gravitational physics is described in terms of a nongravitational quantum theory. In principle, this allows questions about the gravitational system to be precisely formulated and answered within the dual CFT. In practice, however, the precise map between bulk quantities and boundary quantities remains only partially understood.

A central element of the holographic dictionary concerns how the bulk classical geometry and quantum information are encoded in the CFT. The cornerstone of this connection is the Ryu–Takayanagi (RT) formula \cite{Ryu:2006bv}, together with its covariant and quantum generalizations \cite{Hubeny:2007xt,Faulkner:2013ana,Engelhardt:2014gca},
\be\label{eq:RT}
S(\rho_{\A}) =  \frac{\text{Area}(X_{\A})}{4G_{\rm N}} + S_{\rm out}(X_{\A})\,.
\ee
Here, $\A$ denotes a spatial subregion of the CFT, $\rho_{\A}$ is the corresponding reduced density matrix, and $S(\rho_{\A}) = -\mathrm{Tr}(\rho_{\A}\log \rho_{\A})$ is its von Neumann entropy. The right-hand side defines the bulk {\it generalized entropy} $S_{\rm gen}(X_{\A})$ associated with a codimension-two spatial surface $X_{\A}$. The first term is the classical area contribution, while $S_{\rm out}(X_{\A})$ is the entanglement entropy of the bulk quantum fields across $X_{\A}$. In the classical, time-symmetric settings relevant for this paper, $X_{\A}$ is defined as the minimal area surface homologous to $\A$, commonly referred to as the RT surface. More generally, $X_{\A}$ is the minimal quantum extremal surface (QES) homologous to $\A$.

A direct consequence of Eq. \eqref{eq:RT} is that it determines the bulk ``dual'' to the boundary subregion $\A$, namely, the bulk subregion whose quantum information is encoded in $\A$. This region is called the {\it entanglement wedge} $EW_\A$, defined as the region lying between $X_\A$ and $\A$. In this context, the implication of \eqref{eq:RT} is the existence of a recovery channel acting solely on $\A$, which is able to ``reconstruct'' the bulk degrees of freedom in the entanglement wedge $EW_\A$. This process of accessing the bulk quantum information is commonly called {\it entanglement wedge reconstruction} \cite{Czech:2012bh,Jafferis:2015del,Almheiri:2014lwa, Dong:2016eik, Cotler:2017erl,Harlow:2018fse,Chen:2019gbt,Akers:2020pmf}.

In the early developments of the RT proposal, it was observed \cite{Swingle:2009bg} that tensor networks (TNs) used to describe quantum critical many-body states \cite{Vidal:2008zz} exhibit a lattice version of the RT formula. In TNs, the entanglement entropy of a region is upper bounded by the minimal number of edges that must be cut to separate it from the rest of the system, with a weight given by the logarithm of the bond dimension of each edge. This upper bound is expected to be approximately saturated in generic TNs, as explicitly realized in random tensor network (RTN) models \cite{Hayden:2016cfa}.  The graph structure of the TN can then be interpreted as a discrete bulk geometry, which, in the case of critical ground states, is hyperbolic. This provides a concrete and simple illustration of how the spatial structure of AdS emerges from the entanglement and correlation structure of the CFT wavefunctions.

Other TN constructions, such as the HaPPY code \cite{Pastawski:2015qua}, provide explicit models of “holographic codes” designed to capture subregion complementary recovery and mimic entanglement wedge reconstruction in AdS. In fact, strong connections between the RT formula and quantum error correction have been found \cite{Harlow:2016vwg,Akers:2021fut}. In recent years, TNs have been developed in several directions to incorporate additional features of AdS/CFT. These efforts include combining critical behavior with quantum error correction \cite{Cao:2021wrb}; extending TN constructions to non-rigid or background-independent settings \cite{Donnelly:2016qqt,Qi:2022lbd,Dong:2023kyr,Cao:2020ksw,Colafranceschi:2022dig,Akers:2024ixq,Cao:2024nrx,Cao:2026uoq}; and deriving realistic TNs either from the Euclidean path integral preparation of CFT states \cite{Caputa:2017urj,Caputa:2017yrh,Chen:2022wvy,Chandra:2023dgq,Chen:2024unp,Hung:2024gma,Geng:2025efs,Bao:2025plr} or from entanglement distillation protocols \cite{Bao:2018pvs,Bao:2019fpq}. Additional work has aimed at incorporating aspects of bulk dynamics \cite{Kohler:2018kqk,Apel:2021tnn}, and more broadly at capturing further gravitational features \cite{Faist:2019ahr,Dolev:2021ofc,Akers:2024wab,Balasubramanian:2025rcr}. TNs have also proved useful for modeling properties of the black hole interior, including its dynamical growth \cite{Hartman:2013qma,Susskind:2014rva}, the non-isometric character of interior holographic maps \cite{Akers:2022qdl}, and phenomena such as effective field theories in time-dependent backgrounds \cite{Cotler:2022weg,Cao:2023gkw}.

Taken together, these developments illustrate why TNs have become a useful testing ground for ideas about holography and the emergence of geometry from entanglement. Although they are simplified models that do not include full bulk dynamics, they often make the underlying geometric, quantum information-theoretic, and quantum computational structures particularly transparent.

In this spirit, we focus in this paper on a particular proposal motivated by TN models. We begin by reviewing the proposal and its key elements, and then turn to our main results.

\subsection{The PLC}

In many situations, the entanglement wedge $EW_\A$ contains, in addition to the minimal QES $X_\A$, a second, non-minimal QES $X^c_\A$ in the same homology class, which lies closer to the boundary and is the outermost surface relative to $\A$. In such cases, as illustrated in Fig.\ \ref{fig:geompython}, the entanglement wedge develops a {\it python’s lunch}: the region between the minimal surface $X_\A$ and the outermost surface $X^{c}_\A$, called the {\it constriction}.

The python’s lunch conjecture (PLC) \cite{Brown:2019rox} posits that whenever a python's lunch is present, reconstructing the quantum information beyond the constriction is quantum computationally hard. To quantify this, the proposal focuses on an intermediate QES, the {\it bulge} $X^{b}_\A$, which is defined through a maximinimax prescription between $X_\A$ and $X^{c}_\A$ \cite{Brown:2019rox}. In time-symmetric classical settings, the bulge admits a rigorous characterization using Almgren–Pitts min–max theory \cite{Arora:2024edk}. In these cases, the area functional has a second order variation differential operator with a real spectrum, and $X^{b}_\A$ possesses a single negative spatial deformation mode, or equivalently, it has a Morse index $1$.

\begin{figure}[h]
    \centering
    \includegraphics[width=0.55\linewidth]{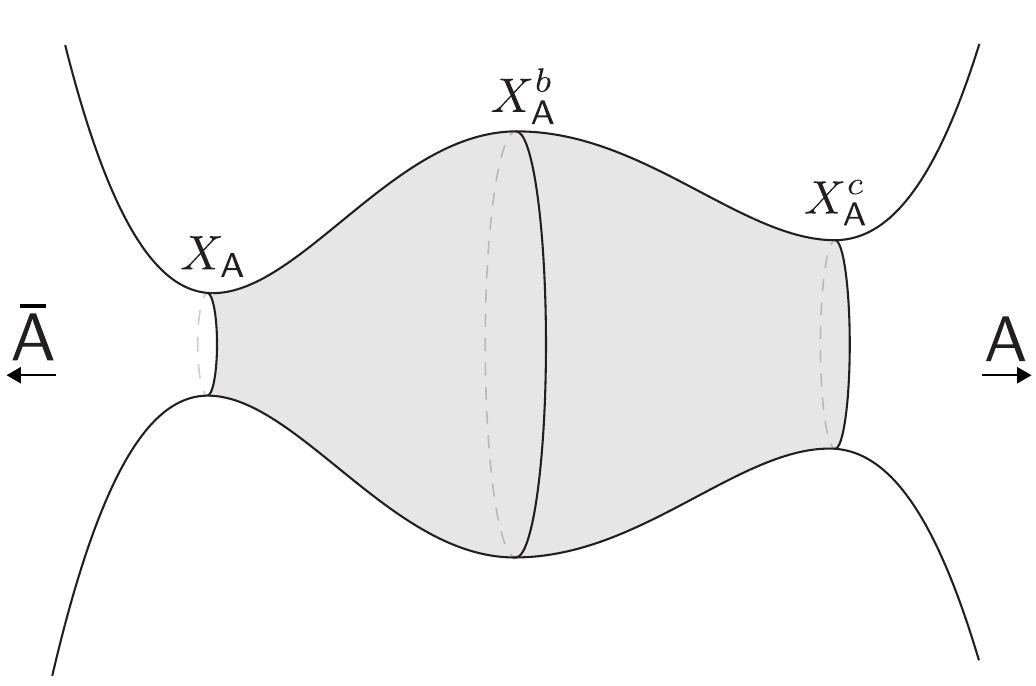}
    \caption{A semiclassical python's lunch in $EW_{\A}$ (gray).}
    \label{fig:geompython}
\end{figure}

Concretely, the proposal is that $X^c_{\A}$ divides $EW_{\A}$ into two regions: the ``simple wedge'' on its exterior and the python's lunch in its interior. On the one hand, decoding the region outside $X^{c}_\A$ is considered computationally ``simple'', and explicit reconstruction protocols were given in \cite{Engelhardt:2021mue} that work in a similar spirit to HKLL \cite{Hamilton:2006az}. The simple wedge includes the {\it causal wedge} $CW_{\A}$, defined as the region of causal accessibility from the boundary $\A$ in some given fixed geometry.

On the other hand, decoding the python's lunch is conjectured to be vastly more complex. This region is always behind the (causal) horizon, and in situations with black holes, it lies in the black hole interior. The computational obstruction is quantified by the exponentially large number of ``simple'' unitary operations on $\A$, or unitary complexity $\Co_{\text{PLC}}$, required to reconstruct this region, which, in the $G_{\rm N}\to 0$ limit, scales as:
\be\label{eq:PLC}
\Co_{\text{PLC}} \;\sim\;  \exp\left(\dfrac{S_{\text{gen}}(X_\A^b)-S_{\text{gen}}(X^c_\A)}{2}\right)\,.
\ee
We will neglect an extensive multiplicative volume term in the complexity, which is subleading in most situations. Moreover, in this paper, we will be mostly interested in classical situations in which the relevant piece of the generalized entropy is the area term.\footnote{Nevertheless, an important aspect of this proposal is that python’s lunches arise naturally from bulk entanglement if the corresponding code subspace is sufficiently large. This occurs even if the “reference” state dual to some semiclassical geometry does not contain a geometric python’s lunch. To reconstruct such code subspaces, one must consider the worst-case scenario of \eqref{eq:PLC} over all states in the chosen subspace. Generic states on sufficiently large code subspaces can develop quantum python’s lunches \cite{Engelhardt:2021qjs}.} 

{\bf Motivation from tensor networks.} The PLC exponent \eqref{eq:PLC} was originally motivated from TN considerations in \cite{Brown:2019rox}.\footnote{For a different line of evidence, based on cryptographic bounds, see \cite{May:2024epy}.} We now briefly review the original argument.

Consider the simplest TN model of the python’s lunch, consisting of two tensors and three homologous cuts, as illustrated in Fig.~\ref{fig:TNpython}. These three cuts correspond to the three QESs in the geometry: $X_\A, X_\A^b, X_\A^c$. To each cut $X_i$, we associate a Hilbert space $\mathcal{H}(X_i)$ consisting of $S_i$ qubits, so that $|\mathcal{H}(X_i)| = 2^{S_i}$. The number of qubits is identified with the generalized entropy of the corresponding surface via $S_i \log 2 = S_{\text{gen}}(X_i)$. To review the PLC argument, we suppress factors of $\log 2$ and use generalized entropies; these factors can be reinstated trivially. We also omit the explicit “bulk dangling legs” that are often included in holographic TNs. Instead, we regard them implicitly as part of the Hilbert space $\mathcal{H}(X_{\A})$.

The TN defines the linear map $V:\mathcal{H}(X_\A)\to\mathcal{H}(X^c_\A)$ with additional structure, namely 
\be 
V = V_2^\dagger V_1\,,
\ee
for isometries $V_1:\mathcal{H}(X_\A)\to \mathcal{H}(X^b_\A) $ and $V_2:\mathcal{H}(X^c_\A)\to \mathcal{H}(X^b_\A)$. Up to an overall rescaling, $V$ is assumed to be approximately isometric. Entanglement wedge reconstruction in this model amounts to inverting this isometry on $\mathcal{H}(X^c_\A)$. The PLC then quantifies the quantum computational cost of doing this. Recall that we focus solely on the python’s lunch region because the simple wedge is already expected to admit low-complexity TN descriptions, for instance, of MERA type. In AdS/CFT, reconstruction methods such as HKLL \cite{Hamilton:2006az} or generalizations \cite{Engelhardt:2021mue} suggest that operators in the simple wedge should be simple to reconstruct.

\begin{figure}[h]
    \centering
    \includegraphics[width=\linewidth]{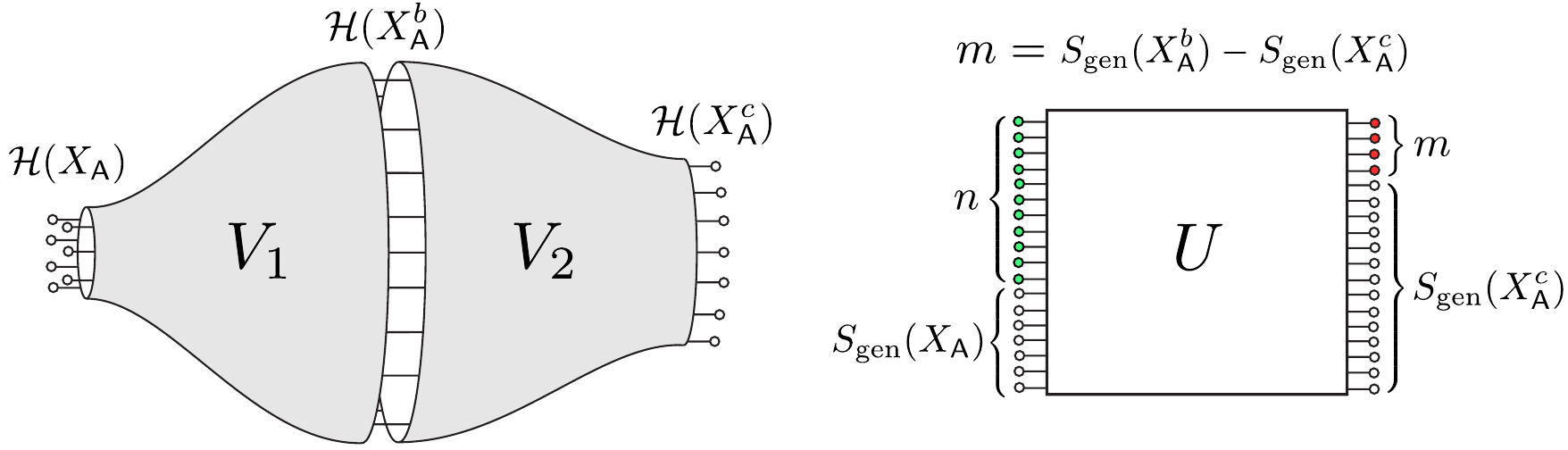}
    \caption{ Left: TN model of a python's lunch, with three cuts associated with the minimal QES $X_\A$, bulge $X^b_\A$ and constriction $X^c_\A$, whose Hilbert spaces are mapped by the two isometries $V_1,V_2$. Right: The map $V = V_2^\dagger V_1$ viewed as a unitary $U$ together with postselection onto the state $\ket{0}^{\otimes m}$ of the $m$ qubits (in red). The amount of postselected and ancilla qubits is $m = S_{\rm gen}(X^b_\A) - S_{\rm gen}(X^c_\A)$ and $n = S_{\rm gen}(X^b_\A) - S_{\rm gen}(X_\A)$, respectively. Reconstructing the lunch amounts to finding an approximate unitary $U_0$ that reproduces the measurement, which is famously hard.}
    \label{fig:TNpython}
\end{figure}

To estimate the complexity, it is convenient to rewrite the map $V$ as a postselected unitary. To do this, one first extends $V_1$ to a unitary by attaching $n = S_{\text{gen}}(X_\A^b)-S_{\text{gen}}(X_\A)$ input ancillas initialized in the fixed state $\ket{0}^{\otimes n}$. Then, one considers $m = S_{\text{gen}}(X_\A^b)-S_{\text{gen}}(X^c_\A)$ output ancillas and postselects (``measures'') them to the $\ket{0}^{\otimes m}$ state as a consequence of the orthogonal projection $V_2^\dagger$. In this way, the action of the isometry $V$ can be written (up to normalization) as $\ket{0}^{\otimes m}\bra{0}^{\otimes m}U\ket{\psi}\ket{0}^{\otimes n}$, for $\ket{\psi}\in\mathcal{H}(X_\A)$ and some unitary $U$. The task is then to construct a unitary $U_0$ that reproduces the measurement,
\be\label{eq:actionTN}
U_0 \ket{\psi} \ket{0}^{\otimes n}   \,\propto\,  \ket{0}^{\otimes m} \bra{0}^{\otimes m}U\ket{\psi}\ket{0}^{\otimes n} \,.
\ee 
This can be achieved using conventional amplitude amplification, but both the algorithm and the resulting unitary depend on the input state $\ket{\psi}$. The proposal of \cite{Brown:2019rox} is to implement a state-independent version of Grover's search, the so-called {\it robust oblivious amplitude amplification} algorithm \cite{Berry:2014ivo}, an algorithm introduced in \cite{Berry:2013tiy} for Hamiltonian simulation. In our context, such an algorithm finds an approximate unitary $U_0$ that approximately satisfies \eqref{eq:actionTN} for {\it any} $\ket{\psi} \in \mathcal{H}(X_\A)$, with exponentially small error. In this case, the algorithm works provided $n \gg m \gg 1$ and that $U$ is sufficiently scrambling.\footnote{These conditions are expected to be met in semiclassical gravity whenever $S_{\text{gen}}(X_\A^c)-S_{\text{gen}}(X_\A) = O(G_{\rm N}^{-1})$ and the wormhole is parametrically longer than $O(\log G_{\rm N}^{-1})$ in AdS units.} The number of ``simple'' unitary operations in this algorithm is
\be\label{eq:unitaryops}
\Co(V) \sim \exp\left(\dfrac{m}{2}\right)\,,
\ee 
which, upon identification of $m$ with the difference in entropies between the bulge and constriction cuts, leads directly to the proposal \eqref{eq:PLC}. The PLC algorithm is expected to be optimal under genericity conditions. To reconstruct the python's lunch from the exterior, one initializes the measurement apparatus in $\ket{0}^{\otimes m}$ and applies $U_0^\dagger$ to the constriction and apparatus.

\subsection{Locality of the tensor network}

The TN representing $V$ above is structureless beyond the decomposition $V = V_2^\dagger V_1$. The bulge cut $X^b_{\A}$ of the network is defined as the auxiliary Hilbert space in this decomposition. However, to obtain \eqref{eq:unitaryops} as the leading unitary complexity, there is a central but implicit assumption: $V_1$ and $V_2$ must themselves be ``low-complexity'' isometries. Thus, despite the general structure of the network, the tensors cannot be generic. Otherwise, $V$ would be a maximally complex isometry, with unitary complexity\footnote{The complexity of any isometry is upper bounded by $\Co_{\max}\sim D_{\text{in}}D_{\text{out}}$, where $D_{\text{in}}$ is the dimension of the input and $D_{\text{out}}$ is the dimension of the output Hilbert spaces. A generic isometry saturates this upper bound.}  
\be\label{eq:cmax}
\Co_{\max} \sim \exp\left(S_{\text{gen}}(X_\A)+S_{\text{gen}}(X^c_\A)\right)\,,
\ee 
in which case the presence of the bulge would be entirely irrelevant. In this paper, we will assume that $V_1$ and $V_2$ have additional structure and explore whether the PLC prediction of the complexity is too low. However, it is also possible for the PLC prediction \eqref{eq:PLC} to be larger than the maximal complexity bound \eqref{eq:cmax}, when $S_{\text{gen}}(X_\A^b) > 3S_{\text{gen}} (X^c_\A)+ 2 S_{\text{gen}}(X_\A)$. This failure mode of the PLC is not the subject of this paper but is briefly discussed in Appendix \ref{app:overestimate}.  

The non-genericity of the isometries $V_1$ and $V_2$ may be justified by requiring the TN to include some notion of {\it locality}. As mentioned above, in TN models of gravity, we interpret the network as a discretization of the bulk spatial geometry, with some degree of locality built into the network.\footnote{This need not be strict spatial locality (few-body nearest-neighbor connectivity), but could be a more general notion such as $k$-locality (few-body all-to-all connectivity).} For example, one might associate a tensor with each AdS-scale volume, as in MERA \cite{Swingle:2009bg}, in the HaPPY code \cite{Pastawski:2015qua}, or in holographic RTNs \cite{Hayden:2016cfa}. For the latter, we could in principle also consider TN models with a finer, sub-AdS discretization built into the network. Formally, even continuum versions of these models can be defined in the $G_{\rm N}\rightarrow 0$ limit, in terms of the local algebras \cite{Faulkner:2022ada}.

\begin{figure}[h]
    \centering
    \includegraphics[width=.55\linewidth]{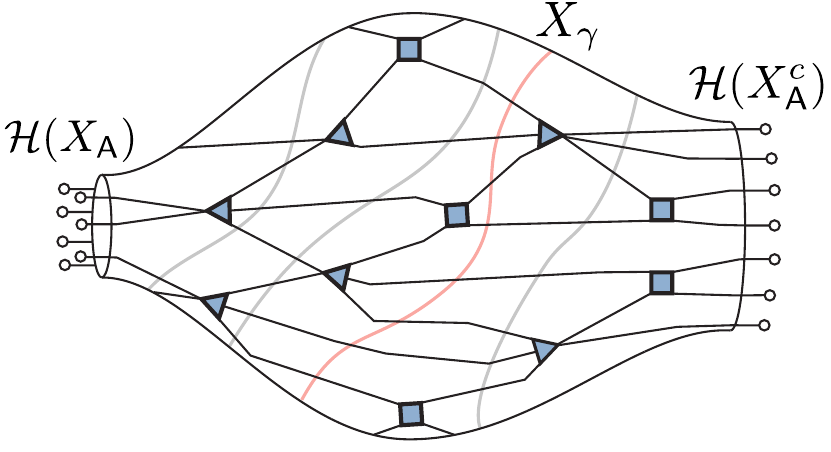}
    \caption{A local TN model of the python's lunch. In gray, a foliation $\gamma$ of the TN with the associated max cut $X_{\gamma}$ in red. }
    \label{fig:localTNpython}
\end{figure}

An example of a local TN model $V_{\rm TN}:\mathcal{H}(X_\A)\to\mathcal{H}(X^c_\A)$ is shown in Fig.\ \ref{fig:localTNpython}. In this local setting, there are multiple max cuts of the TN and associated decompositions of the map $V_{\rm TN}$. The proposal of \cite{Brown:2019rox} is that the bulge is determined by a minimax prescription. Consider a foliation $\gamma$ consisting of a discrete sequence of homologous cuts $\gamma=\{X_t\}_{t=1}^T$ connecting $X_1=X_{\A}$ to $X_T=X^c_{\A}$. The minimax prescription consists of first taking, within the foliation, the maximal cut and then minimizing this maximal cut over all foliations. This yields
\be\label{eq:minimaxTN}
|X^b_{\A}| = \min_{\gamma} \max_{1\le t \le T} |X_t| \,.
\ee
The heuristic motivation for this prescription is to minimize the amount of postselection by choosing a foliation of the TN whose maximal cut is as small as possible. Interpreting cut sizes as generalized entropies leads naturally to a minimax prescription in gravity, which selects the bulge as a well-defined, non-minimal QES. The PLC then assigns a complexity
\be\label{eq:PLCprediction} 
\log_D \,\Co_{\text{PLC}} \;{\sim}\;\dfrac{|X^b_{\A}|-|X^c_{\A}|}{2}\,.
\ee 
We will refer to this formula as the PLC prediction for TNs.

\subsection{This paper}

Our goal in this paper is to revisit the PLC prediction in local TN models of the python's lunch $V_{\rm TN}$ and to quantitatively test whether
\be\label{eq:PLCTN}
\hspace{2.5cm}
\Co(V_{\text{TN}}) \;\stackrel{?}{\sim}\; \Co_{\text{PLC}}\,.
\ee
Here $\Co(V_{\rm TN})$ denotes the circuit complexity of implementing the TN map $V_{\rm TN}$ to fixed accuracy and in a chosen gate model. This is the TN proxy for the PLC reconstruction complexity. Throughout the paper, we focus on the leading large-$D$ exponent and ignore polynomial prefactors and other subleading costs.

In doing so, we identify a central but implicit assumption of the PLC, which we term ``computational covariance''. One imagines that the TN in Fig. \ref{fig:localTNpython} can be contracted along any foliation, with the local tensors being applicable with low complexity operations. In particular, the important source of exponential complexity is supposed to be postselection, which occurs when the inputs are larger than the outputs in moving from one cut of the foliation $X_t$ to the next $X_{t+1}$. This covariance property will play a key role in our discussion.

We propose a new class of TNs that satisfy the computational covariance property while being generic. We call them {\it twirled perfect tensor networks} (TPTNs). Although introduced in the context of the PLC, these TNs offer structural insight into it and appear to have applications beyond holography. We show that TPTNs combine the desirable properties of RTNs and holographic perfect (stabilizer) TNs. In particular, TPTNs saturate lattice versions of the RT (QES) formula for arbitrary subregions in the large bond dimension limit.

Our main motivation to test \eqref{eq:PLCTN} is that, as we explain in Sec.\ \ref{sec:markov}, the PLC predicts a surprisingly small unitary complexity for spatially extended black brane interiors. These states come with an IR length scale $L_{\rm IR}$ such that the total entropy is proportional to $L_{\rm IR}$, yet the PLC predicts that $\log \Co_{\rm PLC}$ does not scale with $L_{\rm IR}$. This follows non-trivially from the fact that the bulge spontaneously breaks the spatial symmetries of the holographic system, as originally observed in \cite{Arora:2024edk}. As we discuss, a certain naive TN approach does predict $\log \Co \sim L_{\rm IR}$; however, we show that it is always possible to remove the scaling with $L_{\rm IR}$ in any system with suitably short-ranged correlations. In this regard, the PLC prediction is vindicated. However, the PLC goes further, predicting a quantitative $O(L_{\rm IR}^0)$ exponent that is also smaller than in our general TN construction, thus indicating something special about the complexity of holographic states if the PLC is to be true.
  
We diagnose the needed special ingredients in a TN language in Sec.\ \ref{sec:requirements},  pointing out that the PLC makes implicit assumptions about the bulk or the tensor network model that includes the aforementioned computational covariance property. In Sec.\ \ref{sec:RTN} we will show that this property requires fine-tuning, insofar as it is not a property of random tensor networks (RTNs). Then in Sec.\ \ref{sec:TPTN} we introduce and analyze the twirled perfect tensor networks, which provide an explicit realization of a computationally covariant TN. In Sec. \ref{sec:localshrink} we discuss the complexity of postselection in gravity, drawing on intuition from TPTNs. Finally, in Sec. \ref{sec:discussion} we give an outlook and discuss open questions. The appendices collect further technical details of the analysis.

\section{Case study: extended python's lunches}
\label{sec:markov}

In \cite{Arora:2024edk} it was observed that classical bulge surfaces $X^b_\A$ develop rather surprising geometric properties in spatially extended python's lunches. These extremal surfaces break the spatial symmetries of the system, and this causes the PLC exponent \eqref{eq:PLC} to be sub-extensive in the size of the system. We will now briefly review this phenomenon.

For concreteness, we consider $\A$ to be a holographic two-dimensional CFT in flat space $\mathbb{R}^{1,1}$. We consider a homogeneous time reflection symmetric state $\ket{\Psi} \in \mathcal{H}_{\bar{\A}} \otimes \mathcal{H}_{\A} $ in two copies of the CFT. The semiclassical dual state is defined by the initial data at the bulk moment of time symmetry $\Sigma$,
\be 
\text{d}s_\Sigma^2 = \text{d}\rho^2 + r^2(\rho) \,\text{d}x^2\,,
\label{eq:planar_metric}
\ee 
for $x \in \mathbb{R}$. We take the slice $\Sigma$ to have the topology of $[0,1]\times \mathbb{R}$. Asymptotically, we demand that $r(\rho) \sim e^{|\rho|}$ for $\rho \rightarrow \pm\infty$, so that the spacetime is asymptotically AdS (with $\ell_{\rm AdS} = 1$). Additionally, we will assume that the positive function $r(\rho)$ has a global minimum at $\rho=0$ (the RT surface $X_\A$), a local maximum at $\rho_b>0$, and a local minimum at $\rho_c>\rho_b$ (the constriction $X^c_\A$); this defines a python's lunch. To consider finite quantities, we will regularize the IR of the transverse geometry by adding a box with boundaries at $x = \pm L_{\text{IR}}/2$. We require that all extremal surfaces of interest meet the IR cutoff surface orthogonally.

The naive bulge candidate $X^{0}_\A$ sits at $\rho = \rho_b$. This surface is extremal since $r'(\rho_b) =0$. In the large volume limit, however, this surface will have more than one negative mode, and its Morse index will be greater than $1$, precluding it from being the bulge. More precisely, this occurs when \cite{Arora:2024edk} 
 \be\label{eq:LIR} 
 L_{\text{IR}} >  \dfrac{2\pi}{\sqrt{- r_0 r_0''}}\,,
 \ee 
 where we have defined $r_0 = r(\rho_b)$ and $ r''_0 =r''(\rho_b)< 0$.

  \begin{figure}[h]
 		\centering
 		\includegraphics[width = .95\textwidth]{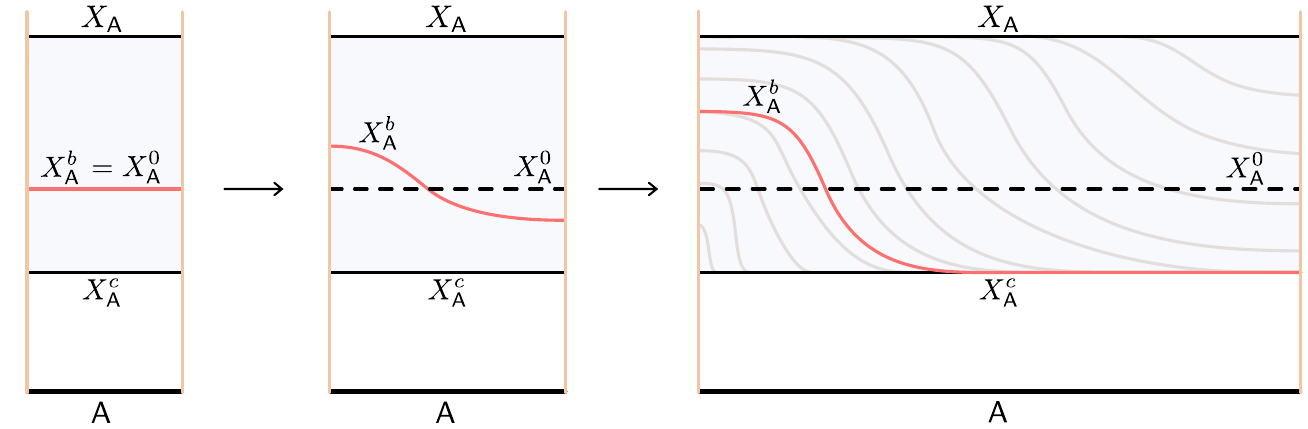}
 		\caption{A planar python's lunch (in blue) delimited in the transverse direction by the spatial IR regulator. As the IR cutoff $L_{\text{IR}}$ is increased, the bulge undergoes a transition from the naive symmetric bulge to an extremal surface which spontaneously breaks the $\mathbb{Z}_2$ spatial reflection symmetry of the IR regulator. In the thermodynamic limit $L_{\text{IR}}\rightarrow \infty $, the bulge hugs the constriction and they only differ at a finite distance from the cutoff. This produces a PLC exponent which is non-extensive in the system's spatial volume. (Figure adapted from  \cite{Arora:2024edk}.)}
 		\label{fig:planarlunch}
 	\end{figure}

When \eqref{eq:LIR} holds, the true bulge $X^{b}_\A$ will spontaneously break the reflection symmetries of the IR regulator, and thus two degenerate bulges will exist, related by a reflection.\footnote{If we were to compactify the system on a spatial circle, by imposing periodic boundary conditions instead of a hard cutoff, we would obtain a similar phenomenon in the large volume limit, where the bulge would spontaneously break $U(1)$, and a continuum set of bulges would exist, related by a zero mode deformation.} Moreover, the bulge $X^{b}_\A$ will mostly hug the constriction, except in a finite region in the thermodynamic limit $L_{\text{IR}}\rightarrow \infty$, as represented in Fig. \ref{fig:planarlunch}. Accordingly, the resulting log complexity predicted by the PLC will satisfy non-extensive behavior
\be\label{eq:areadifference} 
\log \Co_{\text{PLC}} = \dfrac{\text{Area}(X^b_\A) - \text{Area}(X^c_\A)}{8 G_{\rm N}}  \sim   O(L_{\text{IR}}^{0})\,,
\ee
as $L_{\text{IR}} \rightarrow \infty$. 

The first observation is that the thermodynamic entropy of $\A$ scales with the volume of space $L_{\text{IR}}$ and thus the complexity to reconstruct extended interiors is comparatively simpler than for systems at small spatial volume with the same total entropy. As we will now argue, this non-extensive scaling is, in fact, a positive prediction of the PLC. However, we will want to test the proposal at a quantitative level, looking at the precise coefficient in the exponent, which is
\be\label{eq:logcomplexityblackbrane} 
\log \Co_{\text{PLC}} \sim  \dfrac{1}{8G_{\rm N}}\int_{\rho^-_m}^{\rho^+_m} \text{d}\rho \,\sqrt{1-\dfrac{r_c^2}{r^2(\rho)}}\,.
\ee 
where $r_c = r(\rho_c)$ is the value of the $r$ coordinate at the constriction and $\rho^{\pm}_m$ are the two solutions to $r(\rho) = r_c$ closest to $\rho_b$. We provide details on how to compute this exponent in Appendix \ref{app:detailsplanar}.

\subsection{An example}

For concreteness, we shall consider a specific choice of state, where the profile of the bulk geometry is given by
\be 
r(\rho) = \begin{cases}
    r_h \cosh(\rho)\quad \quad \rho< \rho_b\,,\\[.2cm]
    r_c \cosh(\rho-\rho_c)\quad \quad \rho\geq \rho_b\,,
\end{cases}\label{eq:geomex}
\ee 
where $\rho$ is measured in units of $\ell_{\text{AdS}}$. The three independent parameters of the solution can be taken to be $r_h,r_c$ and $r_b = r_h \cosh (\rho_b) = r_c \cosh (\rho_b-\rho_c)$. In this case, the initial data on $\Sigma$ is vacuum, except at $\rho=\rho_b$, where a thin planar stress-energy is included. Gluing constructions of this kind can be found in e.g. \cite{Balasubramanian:2022gmo,Balasubramanian:2022lnw} using thin matter interfaces with a dust equation of state.

Using \eqref{eq:logcomplexityblackbrane} we find that in this case the PLC predicts 
\be\label{eq:logcomplexityPLCmera} 
\log \Co_{\text{PLC}} \sim  \dfrac{\ell}{8G_{\rm N}}\,,
\ee 
for the length scale
\be 
\dfrac{\ell}{\ell_{\text{AdS}}} =\log\left(\dfrac{\sqrt{r_b^2-r_h^2} + \sqrt{r_b^2-r_c^2}}{\sqrt{r_c^2-r_h^2}}\right)
-\frac{r_c}{r_h}\,\text{arctanh}\!\left(\frac{r_h}{r_c}\sqrt{\frac{r_b^2-r_c^2}{r_b^2-r_h^2}}\right)
+\log \frac{r_b}{r_c} \,.
\ee
Note that the length $\ell$ is independent of the IR regulator $L_{\rm IR}$.\footnote{ When $r_c \approx r_h$ and $r_b/r_c\gg1$ the length simplifies at leading logarithmic order to $\ell \approx 2\ell_{\rm AdS}\log(r_b/r_c) \approx \ell_0$, where $\ell_0=  \ell_{\rm AdS}(\text{arccosh} \frac{r_b}{r_h} + \text{arccosh} \frac{r_b}{r_c})$ is the length of the wormhole.}

\subsection{Comparison with MERA model}
\label{sec:mera_interior}

\begin{figure}[h]
    \centering
    \includegraphics[width = .9\textwidth]{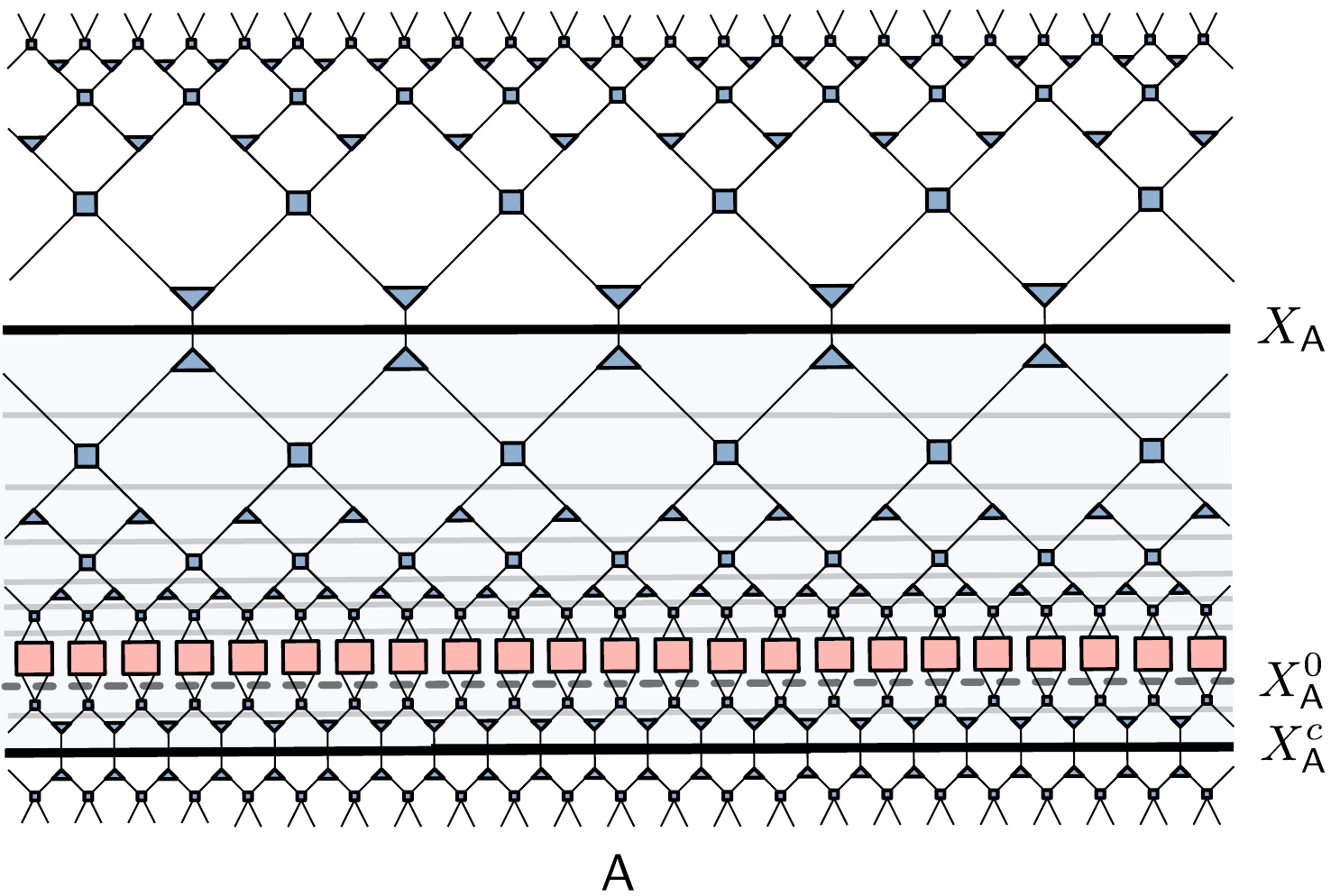}
    \caption{TN model of a planar python's lunch. Two MERA TFDs are contracted together at some radial cutoff by a chain of local unitaries (red boxes) associated with the thin matter domain wall. The map $V$ defined by this TN admits a simple factorization $V = V_2^\dagger V_1$ along the dashed cut. The grey lines denote a standard planar-symmetric foliation and associated contraction scheme. Because the number of postselections required is proportional to $L_{\text{IR}}$, this way of contracting the network yields a complexity $ \log \Co_{\rm TN} \propto L_{\rm IR}$. This is parametrically larger than the PLC prediction.}
    \label{fig:MERApython}
\end{figure}

The non-extensive nature of $\log \Co_{\text{PLC}}$ in this spatially extended setting provides an interesting testing ground to explore the PLC. The PLC prediction suggests that there ought to exist a TN with an analogous sub-extensive $\log \Co$ as well. To set the stage, we first construct a TN that mimics the chosen bulk geometry and discuss a naive estimate of its complexity. Since the geometry is constructed from a cut-and-gluing procedure of vacuum solutions, we will assume that the corresponding local TN model can be constructed likewise. In the most naive model, the vacuum parts are given by two multiscale entanglement renormalization ansatze (MERA) TNs representing thermofield double (TFD) states  \cite{Vidal:2008zz,Swingle:2009bg}. We glue them together at the position of the planar shell, consisting of a layer of local unitaries, as represented in Fig.\ \ref{fig:MERApython}.

By construction, such a TN admits a planar-symmetric foliation and an associated decomposition $V_{\rm MERA} = V_2^\dagger V_1$ into simple isometries $V_1,V_2$, obtained by cutting the network along the $r=r_b$ planar slice, i.e., the naive bulge $X^0_{\A}$. If these isometries, or the operators forming the matter shell, are generic enough, then in the large bond dimension limit the map $V$ becomes an approximate isometry.\footnote{This occurs if $S(X_\A^c)>2S(X_\A)$ and the individual isometries $V_1,V_2$ form approximate unitary $2$-designs.} For this pattern of contraction, the TN satisfies a PLC-like formula with an extensive exponent of the form
\be\label{eq:CMERA}
\log \Co(V_{\rm MERA}) \lesssim \dfrac{(r_b - r_c) L_{\rm IR}}{8 G_{\rm N}}\,.
\ee
As $L_{\text{IR}} \to \infty$, the PLC prediction \eqref{eq:logcomplexityPLCmera} is parametrically smaller than \eqref{eq:CMERA}. This prediction corresponds to choosing the ``naive'' symmetric bulge in Fig.  \ref{fig:planarlunch}. We have already shown that such a foliation is not the one found by the minimax procedure, so it is not surprising that the complexity of this contraction of the MERA network is larger than that predicted by the PLC.

{\bf Non-extensivity.} The non-extensivity of the PLC exponent \eqref{eq:areadifference} is a positive prediction of the PLC. Physically, far apart regions of the planar black brane state above are uncorrelated in a way that we will make precise below. Therefore, one can treat these subsystems as effectively independent of each other. The complexity to decode uncorrelated subsystems is presumably additive in the number of subsystems, and thus the PLC prediction is non-trivially predicting this for this setting.\footnote{A similar phenomenon was also observed for multi-boundary wormholes in \cite{Arora:2024edk}.}

\subsection{Markov length and comparison with MPS model}

Readers with a background in quantum complexity will also know that there are more efficient ways to prepare such a short-range correlated state than the naive contraction of the MERA network above. Here we discuss two such methods that leverage the short-range correlated nature of the state. The first method makes use of a quantum Markov chain property inherent in the state to build the state with a subextensive complexity exponent. The second method invokes a matrix product state (MPS) model of part of the bulk geometry and uses a known method for contracting such MPS states with a subextensive complexity exponent. We also quantitatively compare both models to the PLC prediction.

\subsubsection*{Markov chain approach}

This method leverages a conditional independence property that enables a state to be assembled from its constituents. This property is characterized in terms of the mutual information, 
\be
    I(\A : \C) := S(\A) + S(\C) - S(\A\C)
\ee
and conditional mutual information,
\be
    I(\A:\C|\B) = S(\A\B)+S(\B\C)-S(\B)-S(\A\B\C).
\ee
In particular, if $I(\A:\C)=0$ then $\rho_{\A \C} = \rho_{\A} \otimes \rho_{\C}$, and if $I(\A:\C|\B)=0$, then $\A \B \C$ form a quantum Markov chain, and the full state can be reconstructed from the marginals $\rho_{\A \B}$ and $\rho_{\B \C}$. For classical distributions $p$, the formula is 
\begin{equation}
    p_{\A \B \C} = \frac{p_{\A \B} p_{\B \C}}{p_{\B}},
\end{equation}
and the corresponding quantum formula is a generalization of this that can be phrased as a quantum channel which builds $\rho_{\A \B \C}$ from $\rho_{\A \B}$ and $\rho_{\B \C}$.

This structure can be used to build our planar symmetric state through an iterated construction in which the full system is slowly grown, starting from the left boundary at $x = - L_{\text{IR}}/2$ (we are assuming open boundary conditions for simplicity). Divide the IR cutoff into segments of length $\Delta L$ and label these segments starting from the left with $i=1,\cdots, \frac{L_{\rm IR}}{\Delta L}$.

At a given stage of the procedure, the regions will be defined by
\begin{align}
    &\A : [- \tfrac{1}{2}L_{\rm IR} , -\tfrac{1}{2}L_{\rm IR} + L], \\
    &\B : [- \tfrac{1}{2}L_{\rm IR} +L , -\tfrac{1}{2} L_{\rm IR} + L + \ell_{\rm M}], \\ 
    &\C : [- \tfrac{1}{2}L_{\rm IR} + L + \ell_{\rm M}, -\tfrac{1}{2}L_{\rm IR} + L + \ell_M +\Delta L ],
\end{align}
where $\ell_{\rm M}$ is chosen such that $I(\A:\C|\B)=0$ (at the classical level) for all $ L \geq 0$. We refer to $\ell_{\rm M}$ as the Markov length and note that $\A$ is growing, $\B$ is always one Markov length, and $\C$ is always one segment length.

For later use, we also denote the coarse-grained entropy density as
\begin{equation}
    s = \frac{r_c}{4 G_{\rm N}},
\end{equation}
so that the coarse-grained entropy of the entire system is $s L_{\rm IR}$. For a boundary segment of $x$-coordinate length $L$, we take the effective Hilbert space size to be $e^{s L}$. This ansatz views the MERA components of the TN as embedding the IR Hilbert spaces defined at the constriction and the RT surface into the UV Hilbert space. We neglect this part of the TN in the remainder of this subsection. 

The procedure starts with the state of $\rho_{\A \B}(L=0)$, for which $\A$ is empty. Viewed as a generic mixed state on a $e^{s \ell_{\rm M}}$ dimensional Hilbert space, this state can be built with complexity at most $e^{2 s \ell_{\rm M}}$. The first new segment $\mathsf{C}$ of length $\Delta L$ is now glued to the state of $\A \B$ to form $\A\B \C$. This is accomplished with a quantum channel $\mathcal{R}_{\B \to \B \C}$ that maps $\rho_{\A \B}$ to $\rho_{\A \B \C}$.

At a general point in the growth process, $\A$ has size $L$, and we add a new $\Delta L$ chunk in the form of $\C$. Crucially, the quantum channel which adds $\C$ is independent of $\A$ and, hence, has a complexity that does not depend on the size of $\A$. We can use the Petz recovery channel for this step. It can be implemented by an isometry that adds $\C$ and an auxiliary system of dimension at most $d_{\C}$. We then discard the auxiliary system. With $d_{\B}=e^{s\ell_{\rm M}}$ and $d_{\C}=e^{s\Delta L}$, the isometry has input dimension $d_{\B}$ and output dimension at most $d_{\B}d_{\C}^2$. Its complexity is therefore at most $O(d_{\B}^2d_{\C}^2)=O(e^{2s(\ell_{\rm M}+\Delta L)})$. We keep the looser upper bound $e^{3s(\ell_{\rm M}+\Delta L)}$ in the estimates below.

Following this procedure and tracking only the exponential parts of the complexity and the total number of steps, the complexity to reach $L = L_{\rm IR}$ is 
\begin{equation}
    \Co \sim e^{2 s \ell_{\rm M}} + \frac{L_{\rm IR}}{\Delta L} e^{3s (\ell_{\rm M} + \Delta L)}.
\end{equation}
Viewing $\ell_{\rm M}$ as a fixed property of the state, we determine $\Delta L$ by minimizing the complexity. The minimum is obtained when $\Delta L = \frac{1}{3 s}$, so the final minimal complexity of this approach is
\begin{equation}
    \Co_{\rm Markov} \sim 3 s L_{\rm IR} e^{ 3 s \ell_{\rm M} + 1}. 
\end{equation}
As desired, we have a subextensive result for the complexity exponent.

What remains is to determine the Markov length. It turns out to correspond to the length at which a phase transition occurs in the RT surface of a subregion. Since $\Delta L$ turned out to be a small length, formally scaling like the Planck length, the key criterion is that $\B$ must be large enough so that any region larger than $\B$ has entropy controlled by an RT surface that reaches the whole boundary RT surface. The situation is sketched in Fig. \ref{fig:estimateMarkov}, which shows two candidate RT surfaces for the $\B$ subregion. Again neglecting the part of the RT length outside the constriction, the difference in length between the two RT candidates is roughly
\begin{equation}
    |X_{\B}^1| - |X_{\B}^2|=  r_c \ell_{\B}  - ( 2\ell_0 + r_h \ell_{\B}).
\end{equation}
If $\B$ is large enough, 
\begin{equation}
   \ell_{\B} > \frac{2 \ell_0}{r_c - r_h},
\end{equation}
then the deeper RT surface is the dominant one, and any region bigger than $\B$ will have the same class of RT surface. In this situation, the entropies entering the conditional mutual information will be
\begin{align}
    & S(\A \B) = \frac{\ell_0 + r_h (\ell_{\A} + \ell_{\B})}{4G_{\rm N}} \\
    & S(\B \C) = \frac{2 \ell_0 + r_h (\ell_{\B} + \ell_{\C} )}{4G_{\rm N}} \\
    & S(\B) = \frac{2 \ell_0 + r_h \ell_{\B}}{4G_{\rm N}} \\
    & S(\A \B \C) = \frac{\ell_0 + r_h (\ell_{\A} + \ell_{\B} + \ell_{\C})}{4G_{\rm N}},
\end{align}
from which it can be seen that $I(\A:\C| \B)=0$ at the classical level. Note that the choice of $\ell_0$ versus $2\ell_0$ is determined by whether the region reaches the left boundary or not. Hence, we may set
\begin{equation}
    \ell_{\rm M} = \frac{2 \ell_0}{r_c - r_h}.\label{eq:markovlengthestim}
\end{equation}

We have already noted that $\log \Co_{\rm Markov}$ is subextensive in the Markov approach, but how does the exponent compare quantitatively with the PLC prediction? For the regime in which $r_h \approx r_c$ and $r_b/r_c\gg1$, the wormhole length $\ell_0$ obeys at leading order $\ell_0 \approx 2 \ln \frac{r_b}{r_c}$, and the PLC prediction \eqref{eq:logcomplexityPLCmera} reduces to
\begin{equation}
    \log \Co_{\rm PLC} \sim \frac{\ell_0}{8 G_N}.
\end{equation}
In contrast, the Markov length is diverging $\ell_{\rm M}\to \infty$ in this limit, and so the predicted complexity exponent in this approach $\log \Co_{\rm Markov}$ is also diverging. This divergence is reasonable since reconstructing the interior requires access to it, but the $\B \C$ subregion does not go inside the constriction unless $\B$ is large enough.

\begin{figure}[h]
 		\centering
 		\includegraphics[width = .6\textwidth]{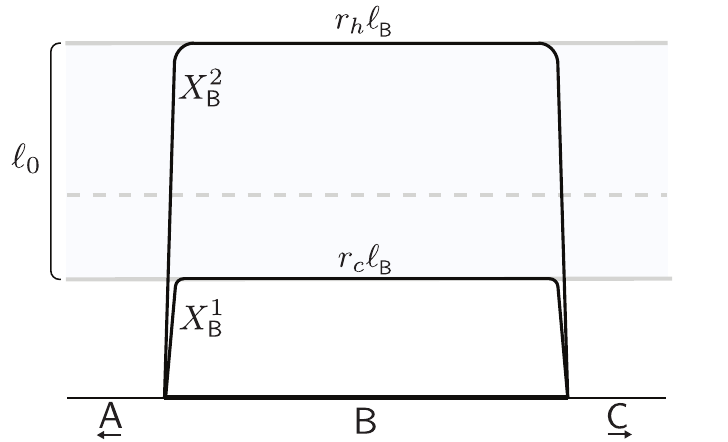}
 		\caption{The two RT candidates $X_{\B}^1$ and $X_{\B}^2$. For $\ell_{\B}>\ell_{\rm M}$, the RT surface is $X_{\B}^2$, and the boundary satisfies the Markov chain condition $I(\A:\C| \B)=0$ at the classical level.}
 		\label{fig:estimateMarkov}
 	\end{figure}

\subsubsection*{MPS approach}

We also consider a second approach in which we model the state using a matrix product state (MPS) and then leverage known contraction methods. The idea is to replace each Markov length chunk with a single MPS tensor with appropriate physical and bond dimensions, as shown in Fig. \ref{fig:planarlunchMPS}.

\begin{figure}[h]
    \centering
    \includegraphics[width = .75\textwidth]{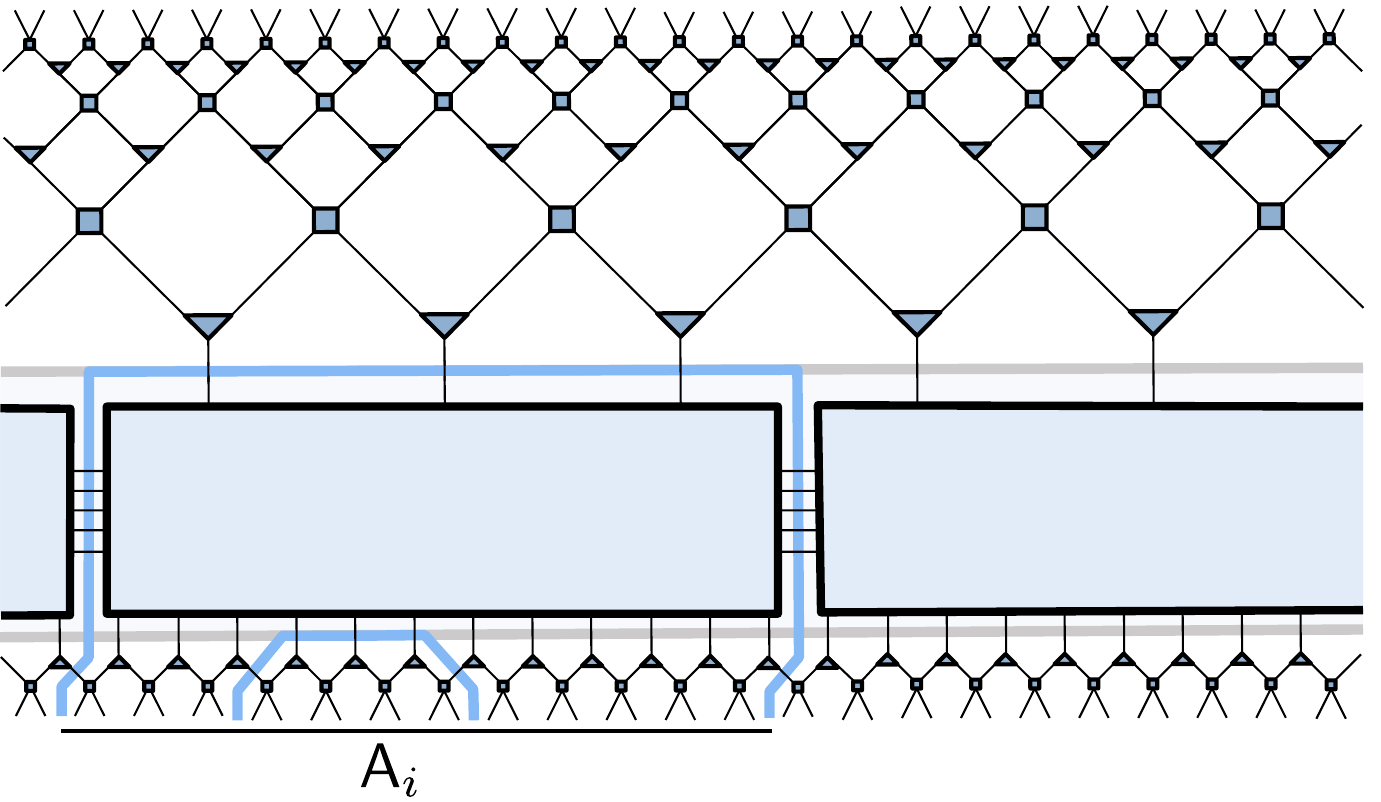}
    \caption{An MPS TN model of a planar python's lunch. We have divided the boundary into uncorrelated intervals $\A_i$ of Markov length $\ell_{\rm M}$.}
    \label{fig:planarlunchMPS}
\end{figure}

\begin{figure}[h]
    \centering
    \includegraphics[width = .7\textwidth]{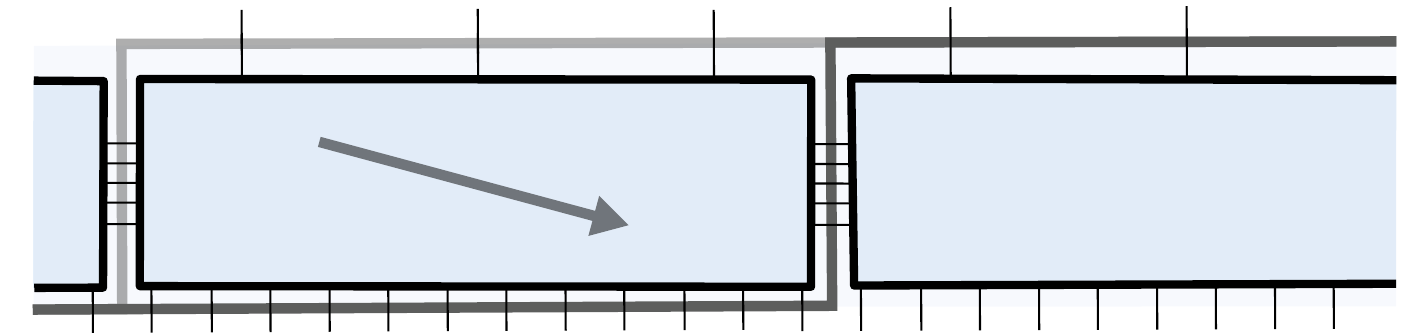}
    \caption{A step in the sideways contraction of the MPS representing the map $V$. The input legs are contained in the light gray slice and the output legs are contained in the dark gray slice.}
    \label{fig:slicing}
\end{figure}

We view the tensor as being defined by the geometry of the black hole as follows. For a region of length $\ell_{\rm M}$ at the Markov length, the two candidate RT surfaces in Fig. \ref{fig:estimateMarkov} are equal in area, 
\begin{equation}
    r_c \ell_{\rm M} = 2 \ell_0 + r_h \ell_{\rm M}.
\end{equation}
A single tensor of the MPS corresponds to the region of space between these two candidate RT surfaces, as shown in Fig. \ref{fig:planarlunchMPS}. The physical dimension of this tensor is the size of the Hilbert space on the constriction and the true RT surface, 
\begin{equation}
    \log D_{\text{phys}} = \frac{(r_c + r_h)\ell_{\rm M}}{4 G_{\rm N}}.
\end{equation}
The bond dimension is set by the orthogonal dimension of the region, the length of the wormhole,
\begin{equation}
    \log D_{\text{bond}} = \frac{\ell_0}{4 G_{\rm N}}.
\end{equation}

We can use the structure of the MPS to contract it in other ways rather than in the longitudinal direction. In this contraction, each MPS tensor is viewed as a map whose input consists of the  RT legs together with one wormhole bond, and whose output consists of the constriction legs together with the other wormhole bond. The input and output dimensions of each step in the contraction are set by
\be 
\log D_{\text{in}} = \dfrac{\ell_{0} + r_h \ell_{\rm M}}{4G_{\rm N}}\,,\qquad \log D_{\text{out}} = \dfrac{\ell_{0} + r_c \ell_{\rm M}}{4G_{\rm N}}\,.
\ee 
For $D_{\text{out}}/D_{\text{in}}\gg 1$, a generic random tensor will be an approximately isometric sideways map. Using this, and that the generic isometry which approximates the tensor is maximally complex, of complexity $D_{\text{in}}D_{\text{out}}$, we arrive at the estimate for the exponent of contracting the MPS sideways:
\be\label{eq:CMPS} 
\log \mathcal{C_{\text{MPS}}} \sim \log(D_{\text{in}}D_{\text{out}}) \sim \frac{2\ell_0+(r_h+r_c)\ell_{\rm M}}{4G_{\rm N}} \approx  \dfrac{\ell_0}{G_{\rm N}}  \,,
\ee
where in the last step we are using \eqref{eq:markovlengthestim} in the regime $r_c \gg r_h$.

Comparing \eqref{eq:CMPS} with the PLC prediction \eqref{eq:logcomplexityPLCmera}, one finds that the PLC exponent is significantly smaller than the generic MPS exponent.
\be\label{eq:discrepancy}
\log \Co_{\text{MPS}} - \log \Co_{\text{PLC}} \gtrsim \dfrac{7\ell_0}{8G_{\rm N}}  \,.
\ee
The inequality follows from the fact that the length of the bulge is upper bounded by the length of the wormhole $\ell < \ell_0$.

The quantitative discrepancy in the complexity exponents \eqref{eq:discrepancy} can be traced to an implicit assumption in the PLC, namely that the elementary building blocks of the MPS carry additional structure and are therefore not maximally complex. As mentioned in the introduction, one might expect that locality of the TN is enough to guarantee this. However, this is far from clear. In particular, as we will see in the next sections, even if one begins with the hyperbolic network of Fig.\ \ref{fig:MERApython}, and chooses its local tensors at random, there is no general reason to expect that the complexity exponent is the one predicted by the PLC.

\subsection{Summary}

In this section we first reviewed the symmetry breaking phenomenon described in \cite{Arora:2024edk} and showed that the resulting PLC prediction is much smaller than an estimate obtained from a model of MERA TNs contracted in a naive planar-symmetric way. We then observed that, in fact, the PLC prediction is qualitatively reasonable because of the short-range correlated nature of the quantum states in question. A quantitative version of this is provided by a simple Markov chain construction and associated MPS construction, which indeed gave a $\log \Co_{\rm TN}$ that did not scale linearly with $L_{\text{IR}}$. However, we still found that the remaining $O(L_{\text{IR}}^0)$ exponent was smaller in the PLC prediction than in the general MPS construction.

This observation suggests that there is still a puzzle with PLC, even though the most egregious issue, the lack of scaling with $L_{\text{IR}}$, is actually expected. Specifically, the tensors entering the MPS construction of the state must somehow be special, so that they may be contracted with a complexity that is less than the maximum. This is because we expect that generic tensors in the MPS would require something close to the maximum complexity to contract. There is of course no sharp paradox---we are free to imagine that the tensors relevant for holographic systems are indeed quite special. On the other hand, there are ample examples in holography where the states and unitaries involved are indeed generic, such as in discussions of scrambling and chaos. We will focus on this remaining puzzle for the rest of the paper.

\section{Computational covariance in tensor networks}
\label{sec:requirements}

In this section we formulate conditions under which a TN model of the python's lunch can satisfy a PLC-like complexity bound. The basic issue is whether the TN map can be decomposed along intermediate cuts into low-complexity building blocks so that any exponential cost is isolated in the postselection step. This naturally leads to a notion of computational covariance.

We consider a general TN model of the python's lunch defined on a planar graph $\mathcal{G} = (\mathcal{V},\mathcal{E})$, where $\mathcal{V}$ is the set of vertices and $\mathcal{E}$ the set of edges. We assume that the TN defines an approximate isometry
\be
V_{\rm TN}:\mathcal{H}(X_{\A})\rightarrow \mathcal{H}(X^c_{\A})\,,
\ee
between two homologous minimal cuts $X_{\A}$ and $X^c_{\A}$, with
$|\mathcal{H}(X_{\A})|=D^{|X_{\A}|}$ and $|\mathcal{H}(X^c_{\A})|=D^{|X^c_{\A}|}$ for
$|X_{\A}|<|X^c_{\A}|$. Here $D$ denotes the bond dimension of the network, which for simplicity we take to be uniform throughout the graph. For holographic TNs, the local entropy $\log D$ should scale with the CFT central charge, so the natural regime is the large-$D$ limit. In that limit, we identify cut sizes with generalized entropies through
\be
S_{\rm gen}(X_{\A})=|X_{\A}|\log D\,,
\qquad
S_{\rm gen}(X^c_{\A})=|X^c_{\A}|\log D\,.
\ee

We will moreover assume that the two minimal cuts $X_{\A}$ and $X^c_{\A}$ are adjacent in the sense that no other homologous minimal cut lies strictly between them inside their common homology region. This restricts attention to the ``single lunch'' configuration shown in Fig.~\ref{fig:geompython}, for which the original PLC proposal \eqref{eq:PLC} is intended to apply.\footnote{For ``multi-lunch'' geometries, in which several homologous minimal cuts occur, one must instead use the more general prescriptions of \cite{Engelhardt:2021qjs}, or the further extension of \cite{Arora:2024edk}.} An example of such a TN is shown in Fig.~\ref{fig:localTNpython}.

The main questions are then the following: when is the complexity exponent of $V_{\rm TN}$ governed by the PLC prediction \eqref{eq:PLCprediction}, and what structural assumptions are needed for this to hold?

A first necessary ingredient is that the TN map admit a factorization through an intermediate cut into two computationally simple pieces.

\begin{definition}[Simple factorization]
Let $V_{\rm TN}:\mathcal{H}(X_{\A})\to \mathcal{H}(X^c_{\A})$ be an approximately isometric map defined by a TN of bond dimension $D$, where $X_{\A}$ and $X^c_{\A}$ are adjacent homologous minimal cuts. We say that $V_{\rm TN}$ admits a simple factorization if there exists a homologous cut $X_{\max}$ such that cutting the network along $X_{\max}$ induces a decomposition
\be
V_{\rm TN}=V_2^\dagger V_1,
\ee
where
\be
V_1:\mathcal{H}(X_{\A})\to \mathcal{H}(X_{\max}),
\qquad
V_2:\mathcal{H}(X^c_{\A})\to \mathcal{H}(X_{\max}),
\ee
are $\varepsilon(D)$-approximate isometries and have quantum complexity subexponential in $\log D$.
\end{definition}

The point of this requirement is that the complexities of $V_1$ and $V_2$ should be negligible compared with the postselection cost associated with the intermediate cut $X_{\max}$.

One could further require $V_1$ and $V_2$ to be sufficiently generic, or scrambling, so that the state-independent amplitude-amplification argument underlying the PLC applies directly. Even when this holds, however, the only effect is to reduce the postselection exponent by a factor of $2$ relative to naive measurement-and-repetition. Since this does not alter the qualitative scaling, we will not emphasize this refinement here.

Whenever a simple factorization exists, one obtains a PLC-like upper bound on the TN complexity in the large-$D$ limit. If $|\mathcal{H}(X_{\max})|=D^{|X_{\max}|}$, the associated exponent is
\be\label{eq:PLClike}
\log_D \Co(V_{\rm TN})
\lesssim
\frac{|X_{\max}|-|X^c_{\A}|}{\kappa}\,,
\ee
where $\kappa=2$ if state-independent amplitude amplification can be used, and $\kappa=1$ if one instead relies on direct postselection.

The inequality in \eqref{eq:PLClike} allows for the possibility that a more efficient algorithm exists for approximating $V_{\rm TN}$. A simple example is provided by the MERA model of the planar python's lunch shown in Fig.~\ref{fig:MERApython}. The naive planar-symmetric bulge surface $X_{\max}$ yields a simple foliation of the lunch and therefore suggests an extensive PLC exponent through \eqref{eq:PLClike}. However, after rewriting the MERA as an MPS, the transverse contraction discussed in the previous section gives a smaller exponent, in particular one that does not grow extensively with system size.

If the TN map admits a simple factorization, its complexity is controlled by a PLC-like formula of the form \eqref{eq:PLClike}. The next question is whether the intermediate cut $X_{\max}$ can be identified with the minimax cut \eqref{eq:minimaxTN}, namely the TN analog of the bulge QES. The underlying intuition is that one should be able to choose an optimal slicing of the python's lunch that minimizes the amount of postselection. From the TN perspective, however, this already assumes a kind of covariance: along any admissible slicing of the network, the TN map should continue to factor into simple pieces. We now make this more precise.

\begin{definition}[Everywhere non-contracting sequence]
Let $\mathcal{G}=(\mathcal{V},\mathcal{E})$ be a planar graph, and let
\be
\gamma:\quad X_0\to X_1\to \cdots \to X_T
\ee
be a discrete sequence of homologous cuts. We say that $\gamma$ is everywhere non-contracting if, for each step $t\to t+1$, the cut $X_{t+1}$ is obtained from $X_t$ by pushing the cut across a collection of vertices $\mathcal{V}_t\subset \mathcal{V}$ adjacent to $X_t$, and for every $v\in \mathcal{V}_t$ the number of incident cut edges does not decrease under that elementary move:
\be
n_{X_{t+1}}(v)\ge n_{X_t}(v)\,.
\ee
Here $n_X(v)$ denotes the number of edges in the cut $X$ incident on the vertex $v$.
\end{definition}

This condition is local: it requires that every elementary move in the foliation be non-shrinking at every vertex across which the cut is pushed.

For any two homologous cuts $X$ and $Y$, whenever the portion of the TN between them defines a linear map, we denote this induced map by
\be
V_{\rm TN}[X\to Y]:\mathcal{H}(X)\to \mathcal{H}(Y).
\ee
In particular, for a sequence $\{X_t\}$ we write
\be
V^t_{\rm TN}:=V_{\rm TN}[X_t\to X_{t+1}]\,.
\ee

\begin{definition}[Computational covariance]
A TN on a planar graph $\mathcal{G}=(\mathcal{V},\mathcal{E})$ with constant bond dimension $D$ is called computationally covariant if, for every everywhere non-contracting sequence of homologous cuts $\{X_t\}_{t=0}^T$, each elementary induced map
\be
V^t_{\rm TN}:\mathcal{H}(X_t)\to\mathcal{H}(X_{t+1})
\ee
is:
\begin{enumerate}
    \item an $\varepsilon(D)$-approximate isometry (or, when $|X_t|=|X_{t+1}|$, an approximate unitary), and
    \item of quantum complexity subexponential in $\log D$.
\end{enumerate}
\end{definition}

This definition captures the idea that any foliation built entirely from locally non-shrinking moves should yield a decomposition of the TN map into computationally simple pieces.

\begin{lemma}
Let $V_{\rm TN}:\mathcal{H}(X_{\A})\to\mathcal{H}(X^c_{\A})$ be a computationally covariant TN, where $X_{\A}$ and $X^c_{\A}$ are adjacent homologous minimal cuts. Suppose there exists a homologous cut $X_{\max}$ such that:
\begin{enumerate}
    \item $X_{\max}$ can be reached from $X_{\A}$ by an everywhere non-contracting sequence of homologous cuts, and
    \item $X_{\max}$ can also be reached from $X^c_{\A}$ by an everywhere non-contracting sequence of homologous cuts.
\end{enumerate}
Then $V_{\rm TN}$ admits a simple factorization through $X_{\max}$.
\end{lemma}

\begin{proof}
By computational covariance, the non-contracting sequence from $X_A$
to $X_{\max}$ composes to an approximate isometry
$V_1:\mathcal H(X_A)\to \mathcal H(X_{\max})$ of subexponential
complexity. Similarly, the non-contracting sequence from $X_A^c$ to
$X_{\max}$ gives an approximate isometry
$V_2:\mathcal H(X_A^c)\to \mathcal H(X_{\max})$. Cutting the network
along $X_{\max}$ gives $V_{\rm TN}=V_2^\dagger V_1$, so
$V_{\rm TN}$ admits a simple factorization through $X_{\max}$.
\end{proof}

\begin{figure}[h]
    \centering
    \includegraphics[width=.55\textwidth]{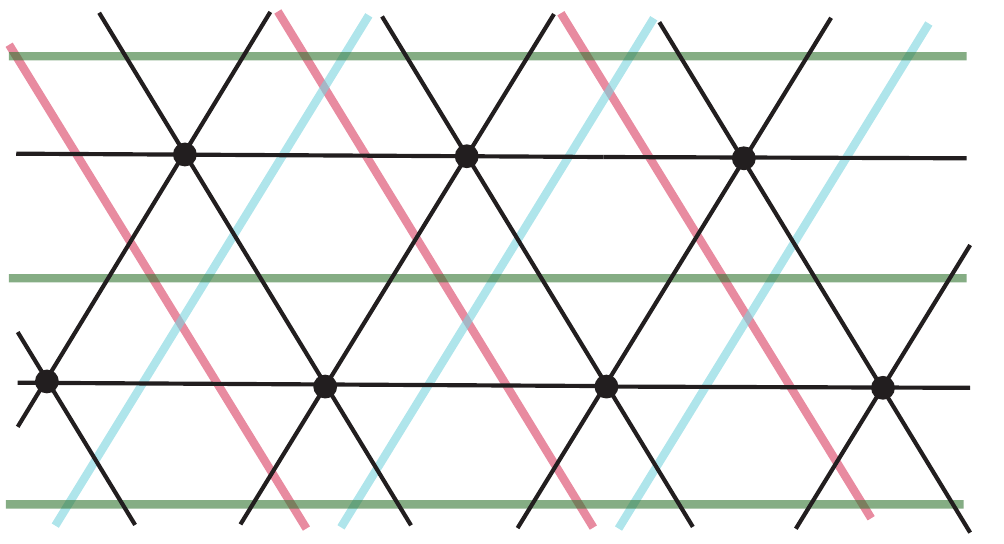}
    \caption{A computationally covariant TN in a flat geometry admits a sequence of unitary maps in any of the three slicings (blue, red, green). Moreover, the respective unitaries are quantum computationally ``simple''.}
    \label{fig:compcov}
\end{figure}

If $V_{\rm TN}$ is computationally covariant, then any everywhere
non-contracting foliation gives a simple factorization of the network.
One can therefore optimize over such foliations and minimize the largest
cut that appears along the way. This defines the \emph{non-contracting
bulge}
\be\label{eq:minimaxTN2}
|X_{\A}^{\rm nc}|
=
\min_{\gamma^*}\;\max_{0\le t\le T}|X_t|\,,
\ee
where the minimization is over sequences of homologous cuts
\be
\gamma^*:\quad X_0\to X_1\to\cdots\to X_T,
\qquad
X_0=X_{\A},\quad X_T=X^c_{\A},
\ee
such that, from each minimal cut, the sequence is everywhere
non-contracting up to a maximal cut. The PLC argument then gives the
covariant TN bound
\be\label{eq:PLClike2}
\log_D \Co(V_{\rm TN})
\lesssim
\frac{|X_{\A}^{\rm nc}|-|X^c_{\A}|}{\kappa}\,.
\ee

The cut $X_{\A}^{\rm nc}$ should not be confused with the minimax
bulge $X_{\A}^b$ of the discrete network. The minimax bulge is defined
by optimizing over all foliations, while $X_{\A}^{\rm nc}$ only allows
foliations that are locally non-contracting on the way to the largest
cut. These two notions can differ in a discrete graph, and
$X_{\A}^{\rm nc}$ need not exist at all. Thus
\eqref{eq:PLClike2} gives a PLC-like bound for computationally covariant
TNs, but it does not automatically reproduce the geometric PLC exponent in generic situations.

As we explain in Sec.~\ref{sec:localshrink}, this obstruction is inherently tied to the discreteness of the TN. Along a generic TN foliation, some local moves can decrease the number of cut legs at a vertex, even if the total cut size is increasing. Such moves require local postselection and are therefore excluded from the definition of $X_{\A}^{\rm nc}$. In a smooth geometry, this local obstruction is absent: one can choose a foliation that expands locally from a minimal surface up to the bulge. We return to this point in Sec.~\ref{sec:localshrink}.

\section{Testing random tensor networks}
\label{sec:RTN}

In this section, we test the computational covariance property in random tensor network (RTN) models of holography \cite{Hayden:2016cfa}.\footnote{RTNs have been shown to reproduce different holographic properties \cite{Hayden:2016cfa,Akers:2021pvd,Akers:2022zxr,Akers:2024pgq}.} At first sight, RTNs seem like promising candidates. Since the tensors are drawn from rotationally invariant ensembles, the network does not distinguish any preferred direction, and one might therefore expect a covariant behavior, at least on average, under different slicings of the graph. Moreover, the linear maps defined by RTNs often preserve inner products approximately on small sets of states, suggesting that they may be close to isometries or unitaries in an appropriate sense. However, this intuition is misleading. Whether a random map is ``close to unitary'' depends strongly on the notion of approximation, and several inequivalent notions that look similar at first lead to rather different conclusions.

A first goal of this section is to disentangle these notions. For instance, a random map may reproduce the pairwise overlaps of a given collection of states and hence admit a unitary approximation on that restricted set, while still failing to be close to any single unitary uniformly on the full Hilbert space. Since computational covariance requires a factorization that works simultaneously for all input states, this distinction is essential.

Even when some approximate unitary description exists, there is a further issue related to complexity. A random linear map can sometimes be completed to a unitary on a sufficiently restricted input, for example, after tensoring with an ancilla or when evaluated on a particular entangled state. But the resulting unitary need not respect the tensor-product structure naturally associated with the network, in particular, if the corresponding input subspace is highly entangled. This makes it difficult to extract a useful complexity bound from the locality of the underlying tensors alone. Thus, the question is not only whether RTNs admit approximate unitary descriptions but also whether they do so in a way compatible with the geometric decomposition of the network and with the notion of complexity relevant for the PLC.

With this in mind, we proceed in stages. We first show that individual random tensors are generally not perfect: depending on the notion of approximation, they may preserve some overlaps while still failing to be well approximated by a single isometry on the full Hilbert space. We then consider RTN maps built from local tensors and show that locality itself tends to obstruct isometricity, even when the graph is everywhere expanding. This leads us to restrict the map to smaller $\alpha$-bit code subspaces, for which approximate isometries can exist. We next ask whether the corresponding closest isometries are computationally simple and argue that they are not, in general. Finally, we discuss more explicit ways of implementing RTN maps, including schemes based on local measurements or local unitary replacement. The overall picture is that RTNs do not naturally realize a simple strong form of computational covariance, although in certain restricted settings they can come closer to it.

\subsection{Definition}

A RTN is defined by a graph $\mathcal{G} = \lbrace \mathcal{V}, \mathcal{E} \rbrace $, where $\mathcal{V}$ is the collection of vertices and $\mathcal{E}$ is the collection of edges. For each vertex $x \in \mathcal{V}$, we define a Hilbert space $\mathcal{H}^{x}_{1}\otimes ... \otimes \mathcal{H}^{x}_{i_{n_{x}}}$ and the state
\be\label{eq:tensor}
\ket{R_{x}} = \dfrac{1}{
\sqrt{D^{n_{x}}}
}\sum_{i_1,...,i_{n_{x}}=1}^D\,(R_{x})_{ i_1...i_{n_{x}}} \ket{i_1}...\ket{i_{n_{x}}}\,.
\ee
Here $n_{x}$ is the degree of vertex $x$ (the number of edges that contain it), and $D = |\mathcal{H}^{x}_{i}|$ is the {\it bond dimension}, which for simplicity will be kept constant throughout the network. The states $\ket{R_{x}}$ are chosen to be independent random states distributed according to some probability distribution, which could be the uniform random state distribution \cite{Hayden:2016cfa}. For convenience, however, we will take $(R_{x})_{ i_1...i_{n_{x}}}$ to be independent complex Gaussian random variables with zero mean and unit variance\footnote{Uniform random states are effectively Gaussian random for most purposes in the large bond dimension limit.}
\be 
\overline{(R_{x})_{I}} =0\,,\quad\quad \overline{(R_{x})_{I}(R_{y})^*_{J}} =\delta_{x,y}\delta_{I,J}\,.
\ee 
where we are using the shorthand notation $I = i_1...i_{n_{x}}$, $J = j_1...j_{n_{y}}$ for all the indices. With this choice, the state \eqref{eq:tensor} is only normalized on average.

We now need to connect the vertices to construct the network. Given an edge $e = \lbrace x,y \rbrace  \in \mathcal{E}$, the rule is to contract the corresponding factors $\mathcal{H}^{x}_{a}$ and $\mathcal{H}^{y}_{b}$ with the maximally entangled state on the corresponding factors
\be\label{eq:epr} 
\ket{\phi_e}  = \dfrac{1}{\sqrt{D}}\sum_{i=1}^D\ket{i_{x,a}^*}\ket{i_{y,b}}\,.
\ee

Denote by $\partial \mathcal{G}$ the set of open tensor legs, i.e., those tensor indices that are not paired by an internal edge in $\mathcal E$. The contraction of the internal edges defines a (generally unnormalized) pure state on the boundary,
\be
\ket{V_{\rm RTN}}
=\left(\prod_{e\in\mathcal E}\bra{\phi_e}\right)
\left(\prod_{x\in\mathcal V}\ket{R_{x}}\right)
\in \mathcal H(\partial \mathcal{G})\,.
\ee
For the models relevant to this paper, we identify
$\mathcal H(\partial \mathcal{G})\simeq \mathcal H(X_\A)\otimes \mathcal H(X_\A^c)$.\footnote{For simplicity, we do not include bulk dangling legs and instead absorb them into $\mathcal H(X_\A)$.}
The python's lunch model is the linear map $V_{\rm RTN}:\mathcal H(X_\A)\to \mathcal H(X_\A^c)$ associated with this boundary state by the state-map isomorphism. Equivalently, for any $\ket{\psi}\in\mathcal H(X_\A)$,
\be\label{eq:VstateRTN_action}
V_{\rm RTN}\ket{\psi}
\,\propto\,
\left(\bra{\psi^*}_{X_\A}\otimes \mathbf 1_{X_\A^c}\right)
\ket{V_{\rm RTN}}\,,
\ee
where the contraction is performed on the $X_\A$ legs only. The global constant of proportionality will be fixed by the condition that $V_{\rm RTN}$ is an isometry on average.

\subsection{Random tensors are not perfect}
\label{sec:gaussianRT}

We now consider an individual random state of the form \eqref{eq:tensor}. We split the legs of the tensor into $n$ input legs and $m$ output legs, and contract the input legs with suitable maximally entangled states to view the tensor as a Gaussian random map
\begin{center} 
\includegraphics[width=0.21\linewidth]{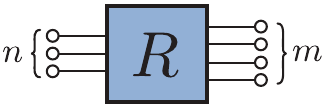} 
\end{center}
representing
\be\label{eq:gaussrandmap} 
R = \dfrac{1}{\sqrt{D^{m}}}\,\sum_{I,J}\,R_{IJ} \ket{I}\bra{J^*} 
\ee 
where $I = i_1,...,i_m$, $J = j_1,...,j_n$ is shorthand notation for the indices and $\ket{I} = \ket{i_1}...\ket{i_m}$, $\ket{J} = \ket{j_1}...\ket{j_n}$. For convenience, we have chosen a different normalization from the one in \eqref{eq:tensor}. This is possible since the overall scaling of the tensor is unimportant for our purposes.

We want to quantify how close to an isometry the map $R$ is for $n\leq  m$, or equivalently, how close to a co-isometry it is (i.e., how close $R^\dagger$ is to an isometry) for $n\geq  m$. The isometry condition would be $R^\dagger R=\mathbf{1}_{D^n}$, or pictorially:
\begin{center} 
\includegraphics[width=0.42\linewidth]{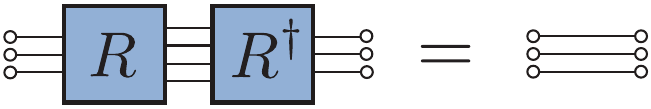} 
\end{center}
We quantify the deviation from isometricity by $\|R^\dagger R-\mathbf{1}_{D^n}\|_{1}$, with the standard definition of the $1$-norm ${\big\|X\big\|_{1}} = \text{Tr}\sqrt{X^\dagger X}$. The following theorem measures this:

\begin{theorem}\label{thrm:gaussianrandomtensor}
Let $R: \mathbb{C}^{D^n} \to \mathbb{C}^{D^m}$ be a Gaussian random linear map, with $n \leq m$, normalized such that 
\be\label{eq:normcondition}
\overline{R^\dagger R} = \mathbf{1}_{D^n} \,.
\ee
Then, in the limit $D \to \infty$, the map $R$ is an isometry with unit probability if and only if
\be\label{eq:condisomg} 
    m > 3n \,.
\ee
\end{theorem}
\begin{proof}
For Gaussian $R$, the operator $R^\dagger R$ is distributed as a Wishart random matrix.  
A key parameter controlling its spectrum is the ratio of the input and output Hilbert space dimensions, given by
\be 
    r :=
    D^{n-m}\,.
\ee
Without loss of generality, assume $n \leq m$, so that $0 < r \leq 1$.  
In the large bond–dimension limit $D \to \infty$, the eigenvalue distribution of the Wishart random matrix $R^\dagger R$ converges to the Marchenko–Pastur (MP) law \cite{marchenko1967distribution}, given by
\be\label{eq:MPdistribution}
    p(\lambda) = \frac{1}{2\pi r \lambda} \sqrt{(\lambda_+ - \lambda)(\lambda - \lambda_-)} \,,
\ee
supported on the interval $\lambda \in [\lambda_-, \lambda_+]$, where
\be
    \lambda_{\pm} := (1 \pm \sqrt{r})^2 \,.
\ee
The distribution has mean $\overline{\lambda} = 1$ and variance $\overline{\lambda^2} - 1 = r$ (see Fig.~\ref{fig:MP}).

\begin{figure}[h]
    \centering
    \includegraphics[width =\textwidth]{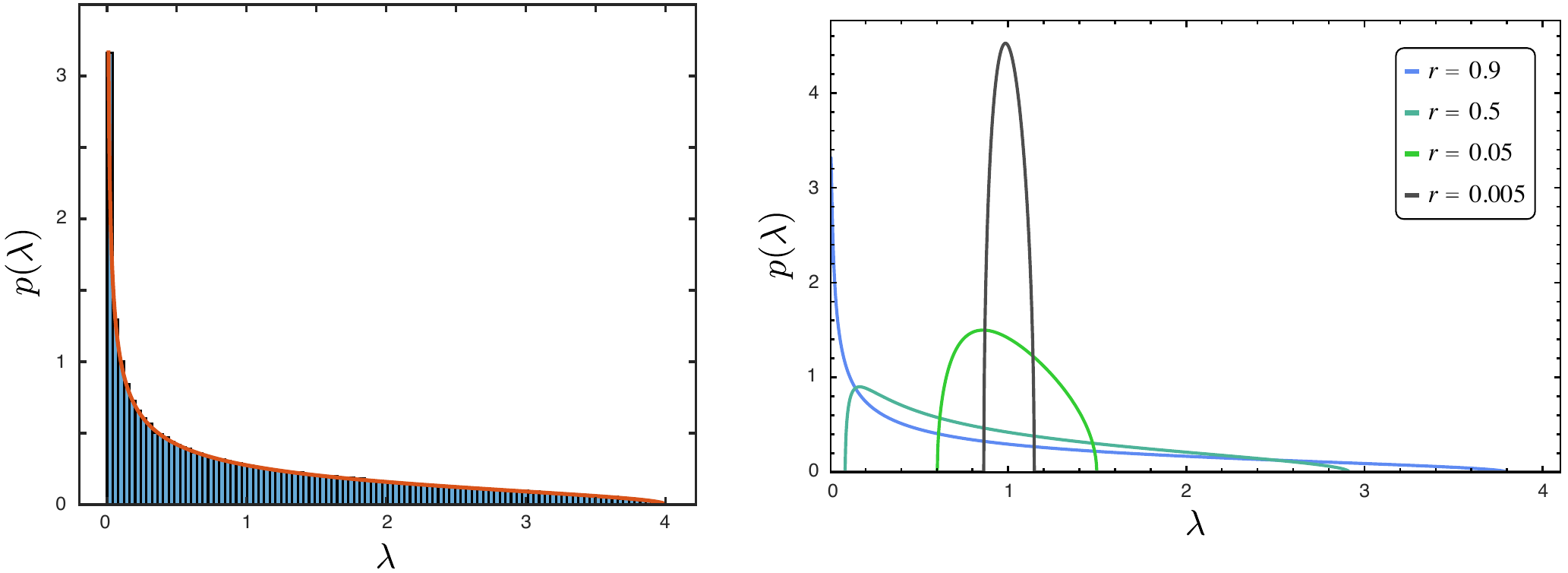}
    \caption{Left: histogram of the eigenvalues of $R^\dagger R$ for a square Gaussian random tensor $R$ with bond dimension $D=400$, compared to the MP distribution at $r=1$. The $O(1)$ spread of the singular values indicates that the map is far from an isometry. 
Right: MP distribution for different values of $r$. As $r$ decreases, the distribution becomes increasingly peaked around $1$, converging to a delta function in the limit $r \to 0$.}
    \label{fig:MP}
\end{figure}

On average, the distance to isometry is determined by the mean absolute deviation of the eigenvalues of $R^\dagger R$ from one,
\be 
    \overline{\big\|R^\dagger R- \mathbf{1}_{D^n}\big\|}_{1} = D^n \,\overline{|\lambda-1|}\,.
\ee
Using the MP distribution \eqref{eq:MPdistribution}, define $\lambda=1+r+\sqrt r\,x$. For $m>n\geq1$, as $D\to\infty$,
\be
    \frac{\overline{|\lambda-1|}}{\sqrt r}
    \longrightarrow \int_{-2}^{2}\frac{|x|\sqrt{4-x^2}}{2\pi}\,\mathrm dx
    =\frac{8}{3\pi}\,.
\ee
Thus
\be\label{eq:avgonenormR}
    \overline{\big\|R^\dagger R-\mathbf{1}_{D^n}\big\|}_{1}
    =\left(\frac{8}{3\pi}+o(1)\right)D^n\sqrt r
    \sim\frac{8}{3\pi}D^{(3n-m)/2}\,.
\ee
\end{proof}

Therefore, we conclude that the average distance to isometry vanishes as $D\rightarrow \infty$ only if \eqref{eq:condisomg} is satisfied, that is, if the number of output legs is more than three times the number of input legs. In this sense, individual random tensors are not ``perfect''. This result does not really depend on the fact that the tensor $R$ was chosen to be Gaussian distributed. In particular, it is also expected to hold if the state $\ket{R}$ is defined as a uniform random state on the $n+m$ legs. In Appendix~\ref{app:measureconcentration}, we outline stronger measure concentration results regarding the asymptotic probability that an individual draw $R$ from the Gaussian ensemble is approximately isometric as $D\to \infty$. 

\subsubsection{Notions of approximate isometry}

A less restrictive measure of approximate isometry is to consider the operator norm, or max singular value, $\big\|X\big\|_{\infty } = \text{sup}_{\ket{\psi}} \|X\ket{\psi}\|$, where the supremum is taken over normalized states. For the random tensor $R$, in the large bond dimension limit, we have
\be\label{eq:avgerrorRT} 
\overline{\big\|R^\dagger R- \mathbf{1}_{D^n}\big\|}_{\infty } \sim 2\sqrt{r} = 2D^{(n-m)/2}\qquad(m>n)\,.
\ee 
This is already exponentially small in $\log D$ if $m>n$ for which case random tensors would behave as approximately ``perfect'' according to this definition. 

For an approximation in the unnormalized $1$-norm, the errors on typical input states must be small enough that their accumulated squared error vanishes. For an isometry $T$, define the state error $\epsilon_T(\psi)=\|(R-T)\ket{\psi}\|$. Averaging over normalized Haar-random inputs gives
\be
 D^n\int\mathrm d\psi\,\epsilon_T(\psi)^2
 =\|R-T\|_2^2
 \leq\|R-T\|_1^2\,.
\ee
Thus $1$-norm closeness necessarily requires a root-mean-square state error $o(D^{-n/2})$. This condition controls the $2$-norm and is necessary, but not sufficient, for $1$-norm closeness. The scale $D^{-n/2}$ also characterizes the overlap between independent random input states.\footnote{More physically, one may consider drawing $\ket{\Psi}$ and $\ket{\Phi}$ from an approximate quantum state $1$-design, which reproduces the same scaling of $|\braket{\Psi}{\Phi}|^2$ while being efficiently prepared on a quantum computer.} More precisely, the expectation value of the overlap between two independent random states is \cite{bengtsson2017geometry}
\be\label{eq:avmagnitudeoverlap}
\int \text{d}\Psi \text{d}\Phi \,|\braket{\Psi}{\Phi}| = \dfrac{\sqrt{\pi}}{2} \dfrac{\Gamma(D^n)}{\Gamma(D^n+\tfrac{1}{2})} \sim \dfrac{\sqrt{\pi}}{2}D^{-\tfrac{n}{2}}\,.
\ee
This follows from the fact that the probability $x =|\langle\Psi|\Phi\rangle|^2$ is distributed as $\mathrm{Beta}(1,D^n-1)$, which is the distribution for the square of a single coordinate of a random unit vector in $\mathbb{C}^{D^n}$. The probability distribution is $f(x) = (D^n-1)(1-x)^{D^n-2}$ on $x \in [0,1]$. The above integral \eqref{eq:avmagnitudeoverlap} is simply $\int_0^1 \text{d}x \,f(x)\,x^s$ for $s=1/2$ .\footnote{Note that in the $D\rightarrow \infty $ limit, $f(x) \sim D^n \exp(-x D^n)$ asymptotes to the Porter-Thomas distribution.}

For Gaussian $R$ with $m>n$, the minimum squared $2$-norm distance scales as
\be
\overline{\min_{T^\dagger T=\mathbf1}\|R-T\|_2^2}
\sim\frac14D^{2n-m}\,.
\ee
The typical-state error criterion above therefore gives the necessary condition $m>2n$. This is not sufficient for the $1$-norm approximation required here, whose condition remains $m>3n$.

\subsection{Locality vs isometricity}

\begin{figure}[h]
    \centering
    \includegraphics[width=0.77\linewidth]{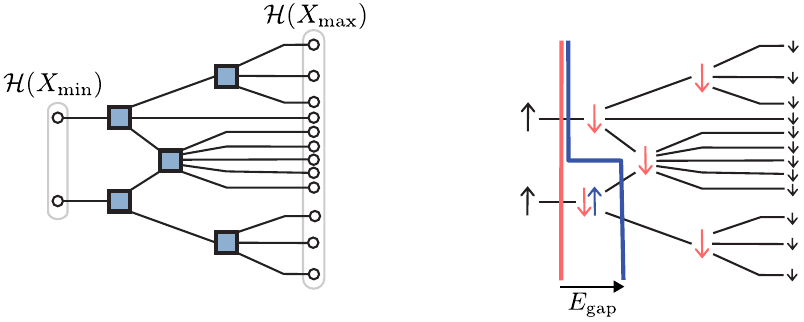}
    \caption{Left: the RTN map $W_{\rm RTN}$, acting from left to right. Right: the ground state and a first excited state of the classical $\mathbb{Z}_2$ Ising model on the network. The black arrows represent the pinning magnetic field, while the red and blue arrows correspond to Ising spins. The ground state corresponds to the red domain wall configuration at the minimal cut. The first excited state is obtained by flipping a single spin in the neighborhood of this domain wall, with the resulting domain wall shown in blue.}
    \label{fig:isometricityRTN}
\end{figure}

We now consider an RTN defined on a planar graph $\mathcal{G} = \lbrace \mathcal{V}, \mathcal{E} \rbrace $, such as the one shown in Fig.\ \ref{fig:isometricityRTN}. This RTN defines a map
\be 
W_{\rm RTN}: \mathcal{H}(X_{\min}) \rightarrow \mathcal{H}(X_{\max})\,,
\ee
between the minimal cut $X_{\min}$ and the ``maximal'' cut $X_{\max}$, homologous to each other. Under what conditions is $W_{\rm RTN}$ an approximate isometry? 

The naive expectation from the previous subsection is that $W_{\rm RTN}$ should be approximately isometric if the graph expands locally at a sufficiently fast rate; for example, when each random tensor has more than  three times as many output legs as input legs. As we now show, however, this expectation does not hold.

\begin{theorem}\label{thrm:RTN}  
Let $W_{\rm RTN}: \mathcal{H}(X_{\min}) \to \mathcal{H}(X_{\max})$ be a Gaussian RTN map of constant bond dimension $D$, where $|\mathcal{H}(X_{\min})| = D^{|X_{\min}|}$ and $|\mathcal{H}(X_{\max})| = D^{|X_{\max}|}$, defined on the graph $\mathcal{G} = \{\mathcal{V}, \mathcal{E}\}$ and normalized such that  
\be\label{eq:condRTN} 
\overline{W_{\rm RTN}^\dagger W_{\rm RTN}} = \mathbf{1}_{D^{|X_{\min}|}}\,.
\ee 
Consider the classical $\mathbb{Z}_2$ Ising Hamiltonian on the network,  
\be \label{eq:isingH}
E[\{s_x\}] = -\frac{1}{2}\left[ \sum_{\langle xy\rangle \in \mathcal{E}} (s_x s_y - 1) + \sum_{e\in\partial\mathcal G}(h_e s_{v(e)}-1) \right]\,,
\ee
where $s_x \in \{\pm 1\}$ is the spin variable on vertex $x \in \mathcal{V}$, $\langle xy\rangle$ denotes nearest neighbors, and $v(e)$ denotes the vertex attached to the open leg $e\in\partial\mathcal G$. The boundary sum counts each open leg separately. The pinning magnetic field $h_e$ is fixed by  
\be\label{eq:pinningm} 
h_e = \begin{cases}
    +1, & e \in X_{\min}, \\[0.2cm]
    -1, & e \in X_{\max}.
\end{cases}
\ee
Then, in the limit $D \to \infty$, the map $W_{\rm RTN}$ is asymptotically isometric in trace norm, with probability tending to one, provided that the energy $E_{\rm GS}$ of the unique ground state and the energy gap $E_{\rm gap}$ of \eqref{eq:isingH} satisfy  
\be \label{eq:RTNconditionisom}
E_{\rm gap} > 2E_{\rm GS}\,.
\ee
Conversely, in the limit $D \to \infty$, the map $W_{\rm RTN}$ is non-isometric with unit probability if
\be\label{eq:RTNconditionisom2}
E_{\rm gap} < E_{\rm GS}\,.
\ee
\end{theorem}

\begin{proof}
 We can write $W_{\rm RTN}^\dagger W_{\rm RTN} = D^{|X_{\min}|} \rho_{X_{\min}}$ where $\rho_{X_{\min}}$ is the (unnormalized) reduced density matrix to subsystem $X_{\min}$ in the Choi state $\ket{W_{\rm RTN}} \in \mathcal{H}(X_{\min}) \otimes   \mathcal{H}(X_{\max})$. On average, $\rho_{X_{\min}}$ corresponds to the maximally mixed state by condition \eqref{eq:condRTN}.

The average square $2$-norm distance to isometry is determined by the average purity of this state,
\be\label{eq:RTNtwonormdistance} 
\overline{\big\|W_{\rm RTN}^\dagger W_{\rm RTN}- \mathbf{1}_{D^{|X_{\min}|}}\big\|^2_{2}} 
= D^{|X_{\min}|}\!\left( D^{|X_{\min}|}\,\overline{\Tr(\rho_{X_{\min}}^2)} - 1 \right)\!.
\ee 

As shown in~\cite{Hayden:2016cfa}, the calculation of the average purity reduces to the thermal partition function of a classical $\mathbb{Z}_2$ Ising model defined on the network,  
\be\label{eq:RTNising}
\overline{\Tr(\rho_{X_{\min}}^2)} \;=\; \sum_{\{s_x=\pm 1\}} e^{-\beta E[\{s_x\}]}\,,
\ee 
with effective inverse temperature
\be 
\beta = \log D\,,
\ee 
and Hamiltonian \eqref{eq:isingH} with boundary conditions set by the pinning field \eqref{eq:pinningm}.  
The key simplification for our purposes, given \eqref{eq:RTNtwonormdistance}, is that the state $\rho_{X_{\min}}$ is normalized only on average, so \eqref{eq:RTNising} holds exactly. If instead we had normalized the state explicitly, as in~\cite{Hayden:2016cfa}, the classical partition function would only approximate the normalized average.

The large bond dimension limit $D \to \infty$ corresponds to the low-temperature limit $\beta \to \infty$ of the Ising model. In this regime, the partition function \eqref{eq:RTNising} is dominated by the ground state of the Ising Hamiltonian. As explained in~\cite{Hayden:2016cfa}, the ground state configuration consists of a single domain wall located at the minimal cut $X_{\min}$, with energy 
\be\label{eq:GSE}
E_{\rm GS} = |X_{\min}|\,.
\ee 
In Fig.\ \ref{fig:isometricityRTN}, this configuration appears as the red domain wall.

The leading correction at low temperature comes from the first excited states of the Ising Hamiltonian. Accordingly, the purity admits the expansion  
\be 
\overline{\Tr(\rho_{X_{\min}}^2)} = e^{-\beta E_{\rm GS}}\!\left(1 + N_* e^{-\beta E_{\text{gap}}} + \dots \right),
\ee 
where $N_*$ denotes the number of first excited states and $E_{\text{gap}}$ their energy above the ground state. The dots represent subleading contributions from higher-energy configurations. Substituting into \eqref{eq:RTNtwonormdistance} gives  
\be\label{eq:twonormspectralgap} 
\overline{\big\|W_{\rm RTN}^\dagger W_{\rm RTN}- \mathbf{1}_{D^{|X_{\min}|}}\big\|^2_{2}} \;\sim\; N_*\, e^{\beta(E_{\rm GS}-E_{\text{gap}})}.
\ee 

Using Jensen's inequality, together with standard relations between Schatten norms, one has
\be\label{eq:twonormspectralgap2} 
\overline{\big\|W_{\rm RTN}^\dagger W_{\rm RTN}- \mathbf{1}_{D^{|X_{\min}|}}\big\|_{1}} \;\leq\; \sqrt{D^{|X_{\min}|}\, \overline{\big\|W_{\rm RTN}^\dagger W_{\rm RTN}- \mathbf{1}_{D^{|X_{\min}|}}\big\|^2_{2}}}\;\sim\; \sqrt{N_*}\, e^{\beta(E_{\rm GS}-E_{\text{gap}}/2)}\,.
\ee 
Therefore, $W_{\rm RTN}$ is an isometry on average in the large bond dimension limit, $\beta \to \infty$, provided the Ising model has a spectral gap larger than twice the ground state energy \eqref{eq:RTNconditionisom}. The opposite statement \eqref{eq:RTNconditionisom2} follows from the divergence of \eqref{eq:twonormspectralgap} if $E_{\rm GS}> E_{\rm gap}$. By concentration of measure as $D \to \infty$, the above statements hold for a single $W_{\rm RTN}$ with unit probability.
\end{proof}

The theorem concerns the unnormalized RTN map, so its deviation from isometry can also include fluctuations in the overall normalization.\footnote{We thank Chris Akers for pointing out this distinction. Writing $G=W_{\rm RTN}^\dagger W_{\rm RTN}$, $d=D^{|X_{\min}|}$ and $q=\Tr G/d$, we have $\overline{\|G-\mathbf{1}_d\|_2^2}=\overline{\|G-q\mathbf{1}_d\|_2^2}+\operatorname{Var}(\Tr G)/d$. For the physical applications we have in mind, we assume that the normalization term is subleading to the contribution from the eigenvalue spread.}

Let us emphasize that Theorem \ref{thrm:RTN} does not decide the
intermediate regime $E_{\rm GS}\leq E_{\rm gap}\leq 2E_{\rm GS}$. Our trace-norm isometry bound requires $E_{\rm gap}>2E_{\rm GS}$,
whereas the second-moment argument only proves non-isometricity when
$E_{\rm gap}<E_{\rm GS}$. The gap between these conditions comes from
the factor $D^{|X_{\min}|/2}$ needed to convert the 2-norm estimate
into a 1-norm estimate.

For a local RTN, the condition \eqref{eq:RTNconditionisom} will generally not be satisfied; instead, \eqref{eq:RTNconditionisom2} will be true. The ground state energy $E_{\rm GS}$ is extensive in the number of edges intersected by the minimal cut \eqref{eq:GSE}. By contrast, the gap $E_{\rm gap}$ is non-extensive since the Ising Hamiltonian is local, and its first excited states are obtained by flipping single spins in the ground state configuration. As a result, $E_{\rm gap}$ is controlled by the smallest local expansion of the bulk vertices of the network.\footnote{ For a single random tensor with $n$ input and $m$ output legs, the corresponding Ising model has a single spin with $E_{\rm GS} = n$, $E_{\text{gap}} = m-n$, and $N_*=1$. The condition \eqref{eq:RTNconditionisom} becomes $m>3n$, the sharp trace-norm threshold derived in Section~\ref{sec:gaussianRT}; \eqref{eq:RTNconditionisom2} detects divergence of the mean squared $2$-norm deviation for $m<2n$.}

\begin{corollary}
Let $W_{\rm RTN} : \mathcal{H}(X_{\min}) \to \mathcal{H}(X_{\max})$ denote the Gaussian RTN map defined above. 
For each bulk vertex $v$ adjacent to the minimal-cut domain wall, let $\chi(v)$ be the number of incident edges of $v$ that are not intersected by $X_{\min}$, and let $\chi_{\min}$ be the minimum of $\chi(v)$ over such vertices. If $\chi_{\min}<|X_{\min}|$, then in the $D\to\infty$ limit, the Gaussian RTN map $W_{\rm RTN}$ is non-isometric.
\end{corollary}

\begin{proof}
The ground state domain wall lies on $X_{\min}$, so $E_{\rm GS}=|X_{\min}|$. Flipping a vertex adjacent to this wall changes the domain wall energy by $\Delta E(v)=\chi(v)-a(v)\leq \chi(v)$, where $a(v)\geq 0$ is the number of incident edges of $v$ intersected by $X_{\min}$. Hence $E_{\rm gap}\leq \chi_{\min}$. If $\chi_{\min}<|X_{\min}|=E_{\rm GS}$, then $E_{\rm gap}<E_{\rm GS}$, and Theorem~\ref{thrm:RTN} implies non-isometricity.
\end{proof}

Thus, even when the graph admits an everywhere-expanding foliation, local expansion alone does not ensure trace-norm isometricity of the associated RTN map $W_{\rm RTN}$.\footnote{This notion of non-isometricity without having a non-trivial kernel was dubbed ``weak non-isometricity'' in \cite{Antonini:2024yif}.} The obstruction is that the Ising gap is controlled by local excitations, while the ground state energy is extensive in the size of the minimal cut.

For instance, in the RTN shown in Fig.\ \ref{fig:isometricityRTN}, we find $E_{\rm GS} = 2$ and $E_{\rm gap} = 1$, so the condition \eqref{eq:RTNconditionisom2} is satisfied, and the RTN map is asymptotically non-isometric. The corresponding second-moment formulas also hold for uniform random states and quantum state $2$-designs. 

\subsubsection{A simple example}

In order to illustrate this explicitly, we consider taking the single Gaussian random tensor $R$ from Sec. \ref{sec:gaussianRT}. We can factor this tensor with the identity on some other factor of $N$ qudits and consider the map $W_R = R\otimes \mathbf{1}_{D^N}$ instead. Pictorially, this is the map
\begin{center} 
\includegraphics[width=0.23\linewidth]{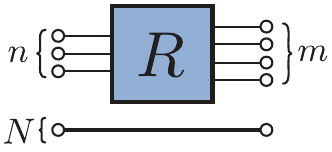} 
\end{center}
where we are using a single thicker line to represent the auxiliary $N$ factors.

Using \eqref{eq:avgonenormR} in this case, we can explicitly compute the $1$-norm distance to isometry 
\be\label{eq:distex}
\overline{\big\|W_R^\dagger W_R- \mathbf{1}\big\|}_{1} \sim \frac{8}{3\pi}D^{N+(3n-m)/2} = \frac{8}{3\pi}D^{E_{\rm GS}-E_{\rm gap}/2}\,.
\ee
In the second step, we have translated the result to the $\mathbb{Z}_2$ Ising model quantities, which in this network only includes a single bulk spin. The ground state energy is $E_{\rm GS} = N+n$, while the single excited state ($N_*=1$) has an energy gap $E_{\rm gap} = m-n$. Alternatively, the average squared $2$-norm from \eqref{eq:twonormspectralgap} scales as $D^{E_{\rm GS}-E_{\rm gap}}$, a different quantity.

Therefore, $W_R$ is approximately isometric on average only if $E_{\rm gap}>2E_{\rm GS}$, or equivalently if
\be 
m > 2N+3n\,.
\ee
For fixed $m,n$ this condition will not be met if $N$ is large enough, even if $m>3n$ so that $R$ is itself approximately isometric in trace norm. Tensoring with the identity multiplies this unnormalized trace norm by $D^N$, while leaving the operator-norm error unchanged.

\subsection{$\alpha$-bit restriction}

One way around the non-isometricity of RTNs is to recognize that demanding reconstruction of code subspaces with $O(|X_\A| \log D)$ entropy may simply be too much. Indeed, this kind of limitation might be expected within the framework of approximate (state-independent) quantum error correction \cite{Hayden:2017xed,Hayden:2018khn}.
Instead, we can restrict ourselves to a smaller subspace $\mathcal{H}_\alpha \subset \mathcal{H}(X_{\min})$ with
\be\label{eq:defalpha} 
\alpha :=\dfrac{\log_D |\mathcal{H}_\alpha|}{\log_D |\mathcal{H}(X_{\min})|}\,.
\ee
We define the {\it $\alpha$-bit restriction} of the RTN map
\be
W_\alpha = W_{\rm RTN} \circ V_\alpha \,,
\ee
where $V_\alpha : \mathcal{H}_\alpha \hookrightarrow \mathcal{H}(X_{\min})$ is an isometric embedding. This map is illustrated in Fig. \ref{fig:RTNalpha} where the green tensor is the isometry $V_\alpha$. We will moreover assume that the $\alpha$-bit subspace $\mathcal{H}_\alpha$ is generic by modeling $V_\alpha$ with a random isometry (in fact, a $2$-design suffices for our purposes).

\begin{figure}[h]
    \centering
    \includegraphics[width=0.4\linewidth]{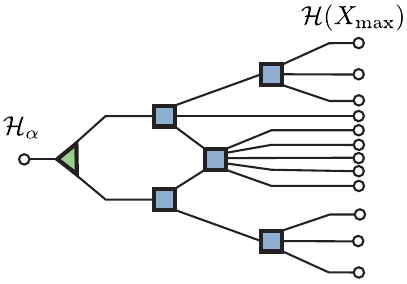}
    \caption{The $\alpha$-bit restriction $W_\alpha:\mathcal{H}_\alpha\rightarrow \mathcal{H}(X_{\max})$ is an approximate isometry, provided $\alpha$ is sufficiently small.  If the RTN has a finer degree of locality close to the minimal cut, we expect $\alpha \ll 1$ (generally $\alpha \lesssim O(1/|X_{\rm min}|)$).}
    \label{fig:RTNalpha}
\end{figure}

Given this, we have the following result:

\begin{corollary}
\label{thrm:alphabit}  
Let $V_\alpha:\mathcal{H}_\alpha\hookrightarrow \mathcal{H}(X_{\min})$ be an isometry drawn from a $2$-design. Consider the $\alpha$-bit restriction $W_\alpha:= W_{\rm RTN} \circ V_\alpha$ of the Gaussian RTN map $W_{\rm RTN}: \mathcal{H}(X_{\min}) \to \mathcal{H}(X_{\max})$, normalized such that  
\be\label{eq:condRTNalpha} 
\overline{W_{\alpha}^\dagger W_{\alpha}} = \mathbf{1}\,.
\ee 
Consider the classical $\mathbb{Z}_2$ Ising Hamiltonian \eqref{eq:isingH} on the extended network, with pinning magnetic field
\be
h_x = \begin{cases}
+1\,, & x\in X_\alpha\,,\\
-1\,, & x\in X_{\max}\,,
\end{cases}
\ee
where $X_\alpha$ is the cut associated with the subspace $\mathcal H_\alpha$, so that $|X_\alpha|=\alpha |X_{\min}|$. The vertex associated with the isometry $V_\alpha$ carries an additional Ising spin. However, by the exact isometricity of $V_\alpha$, flipping this spin while keeping all other spins in the ground state configuration leaves the contribution unchanged. Let $E_{\rm GS}=|X_\alpha|$ be the ground state energy, and let $E_{\rm gap}$ be the gap to the lowest nontrivial excited configuration, namely the lowest excitation not related to the ground state in this trivial way. Then the $\alpha$-bit restriction $W_\alpha$ is asymptotically isometric in trace norm, with probability tending to one as $D\rightarrow \infty$, provided that
\be\label{eq:conditionalpha}
\alpha < \dfrac{E_{\rm gap}}{2|X_{\min}|}\,.
\ee 
\end{corollary}

For a fixed isometry $V_\alpha$, one can prove similar results provided $V_\alpha$ is sufficiently entangling. In particular, the average squared $2$-norm distance to isometry is given by the partition function of a $\mathbb{Z}_2$ Ising model on the original network. The modification with \eqref{eq:isingH} is that now the Ising Hamiltonian contains a non-local boundary term replacing the pinning magnetic field for the vertices at $X_{\min}$. This boundary condition penalizes configurations by an energy given by the second R\'{e}nyi entropies of the state $|V_\alpha\rangle$ across different bipartitions. See Sec.~\ref{sec:RTTPTN} for similar stat mech models.

For the Ising model associated with the TN of Fig.~\ref{fig:RTNalpha}, the lowest nontrivial excitation is obtained by flipping the spins at $V_\alpha$ and at the bottom tensor in the first layer, which gives $E_{\rm gap}=3-2\alpha$. Since in this example $|X_{\min}|=2$, condition \eqref{eq:conditionalpha} implies that the map is isometric in trace distance as $D\to\infty$ whenever $\alpha<1/2$.

\subsection{Closest isometry and complexity}

Given the approximate isometry $W_\alpha$, we now want to explicitly find the isometry that best approximates it. This isometry is obtained from the polar decomposition of $W_\alpha$.

\begin{lemma}
Let $W: \mathcal{H}_A \to \mathcal{H}_B$ be a full-rank linear map, where $|\mathcal{H}_A| = D^n$ and $|\mathcal{H}_B| = D^m$ with $m \geq n$. The isometry that minimizes the distance to $W$,  
\be\label{eq:isomTR}
\min_{T\,{\rm isom}} \|W - T\|_{2}\,,
\ee
is the isometry $T_W: \mathcal{H}_A \hookrightarrow \mathcal{H}_B$ appearing in the polar decomposition
\be
T_W = W(W^\dagger W)^{-1/2}\,,
\ee
Moreover, defining  $\Sigma=(W^\dagger W)^{1/2}-\mathbf{1}_{D^n}$ using the positive square root, the minimum distance to the isometry is determined by the spread of the singular values,
\be
\big\|W - T_W\big\|_{2} = \sqrt{{\rm Tr}\!\left(\Sigma^2\right)}\,.
\ee

\end{lemma}

\begin{proof}
The restricted minimization \eqref{eq:isomTR} can be reformulated as the unconstrained problem defined by the Lagrangian for the square distance
\begin{align}
\mathcal{L}(T,\Sigma) &= \|W - T\|^2_{2} + \text{Tr}\!\left(\Sigma(T^\dagger T - \mathbf{1}_{D^n})\right) \nonumber \\[.2cm]
&= \text{Tr}\!\left(W^\dagger W + T^\dagger T - W^\dagger T - T^\dagger W + \Sigma T^\dagger T - \Sigma \right), \label{eq:lagrangiandistance}
\end{align}
where $\Sigma$ is a Hermitian matrix of Lagrange multipliers enforcing the isometry constraint on $T$.

Stationarity requires
\be\label{eq:minimizationEL}
\dfrac{\delta \mathcal{L}}{\delta T^\dagger } = 0 \quad \Rightarrow \quad T(\mathbf{1}_{D^n}+\Sigma) = W\,,
\ee
which is the polar decomposition of $W$. A way to see it is $(\mathbf{1}_{D^n}+\Sigma)^2 = W^\dagger W$ since $T$ is an isometry and $\Sigma$ is Hermitian. Because the matrix $W^\dagger W$ is full rank\footnote{For a generic tensor $W$ drawn from a smooth probability distribution, this is the relevant case since non–full rank tensors form measure-zero subsets.}, $T_W = W(W^\dagger W)^{-1/2}$. Replacing this in \eqref{eq:lagrangiandistance} leads to $\big\|W - T_W\big\|_{2} = \sqrt{{\rm Tr}\!\left(\Sigma^2\right)}$.
\end{proof}

An alternative way to obtain the polar decomposition is to start from the singular value decomposition (SVD), which allows us to write $W = X \Lambda U^\dagger$ with $U^\dagger:\mathcal{H}_A \to \mathcal{H}_A$, a diagonal square matrix $\Lambda=\mathrm{diag}(\lambda_1,\dots,\lambda_{D^n})\succ0$ containing the singular values of $W$, and $X:\mathcal{H}_A \hookrightarrow \mathcal{H}_B$ an isometry. Comparing with \eqref{eq:minimizationEL}, we obtain
\be
T_W = X U^\dagger\,, \qquad \mathbf{1}_{D^n} + \Sigma = U \Lambda U^\dagger\,.
\ee

The following lemma extends the Frobenius case to all Schatten $p$-norms (indeed, to any unitarily invariant norm).

\begin{lemma}
Let $W: \mathcal{H}_A \to \mathcal{H}_B$ be a full-rank linear map, where $|\mathcal{H}_A| = D^n$ and $|\mathcal{H}_B| = D^m$ with $m \geq n$. For any Schatten $p$-norm,
\be
\min_{T\,\mathrm{isom.}}\|W-T\|_p
=\big\|(W^\dagger W)^{1/2}-\mathbf 1_{D^n}\big\|_p
=\|\Sigma\|_p,
\ee
and the minimizer is again $T_W=W(W^\dagger W)^{-1/2}$.
\end{lemma}

\begin{proof}
This is the rectangular Fan-Hoffman theorem: for any unitarily invariant norm, the closest isometry to a full-rank rectangular matrix is its polar factor. Thus, if $W=X\Lambda U^\dagger$ is the SVD, the minimizer is $T_W=XU^\dagger=W(W^\dagger W)^{-1/2}$. Therefore
\be
\min_{T\,\mathrm{isom.}}\|W-T\|_p
=
\|W-T_W\|_p
=
\|\Lambda-\mathbf 1_{D^n}\|_p
=
\|(W^\dagger W)^{1/2}-\mathbf 1_{D^n}\|_p
=
\|\Sigma\|_p .
\ee
\end{proof}

In general, the image of the $\alpha$-bit encoding $V_\alpha$ is a highly entangled subspace of $\mathcal{H}(X_{\min})$. Consequently, even if the map $W$ has an everywhere-expanding RTN representation, the corresponding isometry $T_{W_\alpha}$ that best approximates $W_\alpha$ cannot be obtained merely by replacing each tensor with a local isometry. Thus, in the circuit model of quantum computation where local gates are regarded as “simple”, there is no reason for the isometry $T_{W_\alpha}$ to be simple. On the contrary, generally, we expect this isometry to be maximally complex, of complexity (see Appendix \ref{app:overestimate})
\be 
 \Co(T_{W_\alpha}) \sim  D^{\alpha |X_{\min}| + |X_{\max}|} \,,
\ee 
Thus, under the natural assumption that $T_{W_\alpha}$ is typical, the RTN does not satisfy computational covariance: its optimal isometric approximation is not expected to be simple.\footnote{The python's lunch model $V_{\rm RTN}$ might still admit a simple factorization if the tensors are chosen to be exact random isometries in some given foliation, like the case of a randomly chosen planar MERA TN in Fig. \ref{fig:MERApython}. However, any other foliation will generally fail to define a simple factorization, at least within the best approximation of the RTN map described here.}

\subsubsection{A simple example}
\label{sec:simpleexample2}

Consider a simple example: a Gaussian random map $R: \mathcal{H}_n \rightarrow \mathcal{H}_m$ of the form \eqref{eq:gaussrandmap} with $m>3n$. In this case, $R$ will be isometric in trace norm in the large bond dimension limit. Now, as above, include an additional $N$ qudits and consider the map $R \otimes \mathbf{1}_{D^N}$, which will cease to be isometric in trace norm for $N\geq(m-3n)/2$. Define the $\alpha$-bit restricted map $W_R: \mathcal{H}_\alpha \rightarrow \mathcal{H}_m \otimes \mathcal{H}_N$ of the form
\be
W_R = (R \otimes \mathbf{1}_{D^N}) V_\alpha \,.
\ee
Here $V_\alpha:  \mathcal{H}_\alpha \hookrightarrow \mathcal{H}_n\otimes  \mathcal{H}_N$ is an isometry and $\mathcal{H}_\alpha$ is some subspace of dimension $D^{\alpha(n+N)}$ according to the definition of $\alpha$ in \eqref{eq:defalpha}. The squared $2$-norm estimates below give the sufficient trace-norm condition
\be
\alpha < \min\!\left\{\frac{n+m}{2(n+N)},\frac{m+N}{3(n+N)}\right\}\quad\Rightarrow\quad \|W_R^\dagger W_R-\mathbf1\|_1\to0\quad(D\to\infty)\,.
\ee
Pictorially, this map is represented on the left:
\begin{center}
\includegraphics[width=0.7\linewidth]{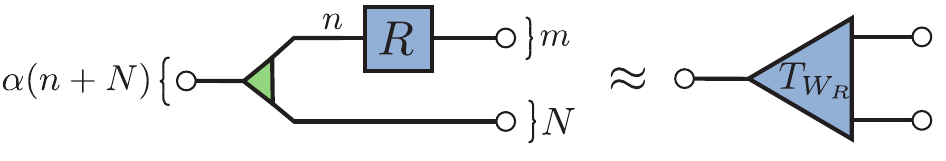}
\end{center}

This map is best approximated by the isometry $T_{W_R}$ represented on the right. Recall that
\be
T_{W_R}=(R\otimes \mathbf 1_{D^N})V_\alpha\,K_\alpha^{-1/2},
\ee
where
\be
K_\alpha:=V_\alpha^\dagger (R^\dagger R\otimes \mathbf 1_{D^N})V_\alpha \,.
\ee

Write $R=T_R S$, where $S=(R^\dagger R)^{1/2}$, and define $B=(S\otimes\mathbf1_{D^N})V_\alpha$. Since $T_R\otimes\mathbf1_{D^N}$ is an isometry, the polar factor has the exact factorization
\be
T_{W_R}=(T_R\otimes\mathbf1_{D^N})T_B\,.
\ee
Its complexity therefore has the additive upper bound
\be
\mathcal C(T_{W_R})\lesssim D^{m+n}+D^{(1+\alpha)(n+N)}\,.
\ee
The equality $T_{W_R}=(T_R\otimes\mathbf1_{D^N})V_\alpha$ holds precisely when the image of $V_\alpha$ is invariant under $S\otimes\mathbf1_{D^N}$; the condition $K_\alpha\propto\mathbf1$ is neither necessary nor sufficient.

\subsection{Simpler ways to implement RTNs}

There are, however, simpler ways to implement RTNs using isometries. In this subsection we describe two explicit methods.

{\bf Local postselection.} The first method is very general. For concreteness, one may consider the approximately isometric RTN map $W_\alpha$ discussed above, although the construction does not rely on approximate isometricity. The goal is not to construct the polar isometry $T_{W_\alpha}$, but rather to implement the map $W_\alpha$ itself by enlarging the Hilbert space, using isometries, and postselecting on simple measurement outcomes.\footnote{See \cite{Malz:2024val,Harley:2025gjs} for related methods in the context of PEPS and fault-tolerant quantum computation.} 

The basic idea can be understood with a simple example (where, in fact, the RTN is not close to an isometry in trace distance). Consider a single Gaussian random tensor $R:\mathbb C^{D^n}\to \mathbb C^{D^n}$, normalized so that
\be
\langle \psi|R^\dagger R|\psi\rangle\le 1
\ee
for every $|\psi\rangle\in\mathbb C^{D^n}$. Now define the quantum channel
\be
\Phi_R(\rho)
=
R\rho R^\dagger
+
\left(\sqrt{1-R^\dagger R}\right)\rho\left(\sqrt{1-R^\dagger R}\right).
\ee
A Stinespring dilation is then provided by a unitary
\be
V_R:\mathbb C^{D^n}\otimes\mathbb C^2\to \mathbb C^{D^n}\otimes\mathbb C^2
\ee
satisfying
\be
V_R|\psi\rangle|0\rangle
=
R|\psi\rangle|0\rangle
+
\sqrt{1-R^\dagger R}\,|\psi\rangle|1\rangle.
\ee
Postselecting the ancilla on the outcome $|0\rangle$ therefore implements $R$ on $|\psi\rangle$. The corresponding success probability is
\be
p(0)=\langle\psi|R^\dagger R|\psi\rangle.
\ee
Conditioned on obtaining the ancilla outcome $\ket{0}$, this implements $R$ on $\ket{\psi}$. If one has access to fresh preparations of the input state, the expected number of trials is $p(0)^{-1}$. For an unknown input state in a coherent algorithm, however, a failed measurement disturbs the input, so one cannot literally repeat the procedure unless the computation is restarted or the postselection is implemented coherently, for example, by
amplitude amplification.

For a Gaussian random map $R$ with normalization \eqref{eq:normcondition}, the operator norm $\|R\|_\infty^2=\lambda_{\max}(R^\dagger R)$ concentrates near $4$ at large $D$. We therefore rescale $R\to R/2$, so that $R^\dagger R\le 1$, up to finite-$D$ tail effects that we ignore. With this normalization, the ensemble-averaged success probability is
\be
\overline{p(0)}=\frac14.
\ee
Thus, one may implement $R$ by adding an ancilla initialized in $|0\rangle$, applying $V_R$, measuring the ancilla, and repeating until the outcome $|0\rangle$ is obtained. On average, this takes about
\be
\left(\overline{p(0)}\right)^{-1}=4
\ee
attempts. One may also use amplitude amplification to reduce the number of applications.

The extension to a rectangular map $R:\mathbb C^{D^n}\to \mathbb C^{D^m}$ is straightforward. When $m\ge n$, the conclusion is essentially unchanged: $R$ can again be implemented by postselecting a single ancilla qubit, with an expected $O(1)$ number of repetitions for a Gaussian random map normalized as above. When $m<n$, by contrast, one must postselect $n-m$ qudits. The success probability is then $O(D^{\,m-n})$, so the expected number of repetitions grows as $O(D^{\,n-m})$, with a possible square-root improvement from amplitude amplification.

We now apply this idea to the full RTN map $W_\alpha:\mathcal H_\alpha\to\mathcal H(X_{\max})$ by following a chosen foliation
\be
X_0\to X_1\to \cdots \to X_T,
\qquad
X_0=X_\alpha,\quad X_T=X_{\max},
\ee
and implementing the network one tensor at a time. The discussion above shows that any tensor whose number of input legs is less than or equal to its number of output legs can be implemented by a Stinespring dilation followed by repeated measurement of a single ancilla qubit until the desired outcome is obtained. This way of implementing the RTN map is illustrated in Fig. \ref{fig:RTNalphaunit}. The precise number of repetitions depends on the tensor, but will typically be $O(1)$.

\begin{figure}[h]
    \centering
    \includegraphics[width=0.45\linewidth]{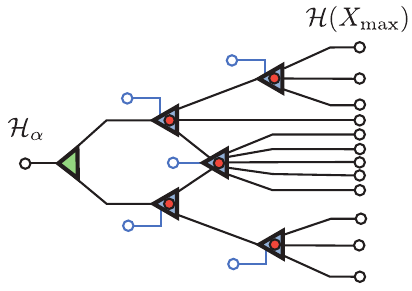}
    \caption{Implementation of the map of Fig. \ref{fig:RTNalpha} with isometries and local qubit postselection. The red circles represent output ancillas postselected into the $\ket{0}$ state and the triangular tensors correspond to the Stinespring dilations of the quantum channels associated with each tensor.}
    \label{fig:RTNalphaunit}
\end{figure}

Assume now that the chosen foliation $\gamma$ is everywhere non-contracting, so that every tensor encountered along the way has at least as many output legs as input legs. Let $i=1,\ldots,K$ label the tensors, and let $r_i=O(1)$ denote the inverse success probability associated with the postselection implementing the corresponding dilation $V_{R_i}$.

If a failure at any step forces one to restart the entire postselected implementation, the expected cost is bounded by
\be\label{eq:complexityexpander}
\Co(W_{\alpha})
\lesssim
\left(\prod_{i=1}^K r_i\right)\sum_{i=1}^K \Co(V_i).
\ee
This is the conservative estimate appropriate for a fully postselected implementation. More efficient coherent implementations may be possible if the postselections can be amplitude-amplified.

Note that the task here is different from the reconstruction task considered throughout the paper. Before, we wanted to find an exact isometry $T_{W_\alpha}$ that best approximates the RTN map $W_\alpha$. Here we instead implement the generally non-isometric map $W_\alpha$ itself, using ancillas and postselection. The former task is expected to be exponentially more complex, with $\mathcal{C}(T_{W_\alpha})\gg \Co(W_{\alpha})$ as $D\to \infty$. In particular, $W_\alpha$ need not be close to an isometry in trace norm, or even in operator norm, for the postselected implementation to have small postselection cost.

{\bf Local isometric replacement.} The second method applies only in restricted situations. Consider the approximately isometric TN map $W_\alpha:\mathcal H_\alpha \to \mathcal H(X_{\max})$ defined in this section, and suppose there exists an everywhere-expanding foliation $\gamma$ such that the slices expand locally fast enough, with each tensor having significantly more outputs than inputs, making all individual tensors approximate isometries. The idea is to replace each local tensor with its isometric counterpart. Although this will not generally be the isometry that best approximates the full map $W_\alpha$, it still defines an isometry that becomes asymptotically close to $W_\alpha$ in the $D\to\infty$ limit. We now make this precise.

\begin{lemma}[Telescoping bound in trace norm]
Let $W_\alpha$ and $\widetilde W_\alpha$ be two tensor networks that differ only by replacing, at each vertex $i=1,\dots,K$, the local tensor $R_i$ by another tensor $T_i$. Then
\be
\|W_\alpha-\widetilde W_\alpha\|_1
\;\le\;
\sum_{i=1}^K
\|\mathcal E_i\|_{1\leftarrow 1}\,
\|R_i-T_i\|_1,
\ee
where $\mathcal E_i$ is the linear map obtained by contracting the network with all tensors held fixed except at vertex $i$, and
\be
\|\mathcal E_i\|_{1\leftarrow 1}
:=
\sup_{y\neq 0}\frac{\|\mathcal E_i(y)\|_1}{\|y\|_1}
\ee
is the induced trace norm.

In particular, if
\be
\|R_i-T_i\|_1\le \varepsilon
\qquad\text{and}\qquad
\|\mathcal E_i\|_{1\leftarrow 1}\le L
\ee
for all $i$, then
\be
\|W_\alpha-\widetilde W_\alpha\|_1
\le K\,L\,\varepsilon.
\ee
\end{lemma}

\begin{proof}
Define the intermediate networks
\be
W^{(i)}
=
N(T_1,\dots,T_i,R_{i+1},\dots,R_K),
\ee
so that $W^{(0)}=W_\alpha$ and $W^{(K)}=\widetilde W_\alpha$. By linearity of the network in each tensor component, for each $i$ there exists a linear map $\mathcal E_i$ such that
\be
W^{(i-1)}-W^{(i)}=\mathcal E_i(R_i-T_i).
\ee
Summing over $i$ gives
\be
W_\alpha-\widetilde W_\alpha
=
\sum_{i=1}^K \mathcal E_i(R_i-T_i).
\ee
The bound then follows from the triangle inequality and the definition of the induced norm.
\end{proof}

Applying the estimates from this section, the local replacement error at vertex $i$ scales as
\be
\|R_i-T_i\|_1 \lesssim D^{-a_i},
\ee
where $a_i=(m_i-3n_i)/2>0$ for a Gaussian tensor with $n_i$ inputs and $m_i$ outputs is the local trace-norm expansion exponent associated with the foliation. To control the corresponding environment map, we further assume that the partial contraction obtained by fixing all tensors except at vertex $i$ does not amplify the operator norm by more than an $O(1)$ factor,
\be
\|\mathcal E_i(Y)\|_\infty \lesssim \|Y\|_\infty .
\ee
Under this assumption, one has
\be
\|\mathcal E_i(Y)\|_1
\le
\rank(\mathcal E_i(Y))\,\|\mathcal E_i(Y)\|_\infty
\lesssim
D^{\alpha |X_{\min}|}\,\|Y\|_1,
\ee
since $\mathcal E_i(Y)$ is a map out of the code subspace $\mathcal H_\alpha$, whose dimension is $D^{\alpha |X_{\min}|}$. Therefore, it is natural to expect
\be
\|\mathcal E_i\|_{1\leftarrow 1}\lesssim D^{\alpha |X_{\min}|}.
\ee
This estimate is certainly crude, but it captures the idea that the only extensive enhancement in the induced trace norm should come from the size of the code subspace. The telescoping lemma then gives
\be
\|W_\alpha-\widetilde W_\alpha\|_1
\lesssim
K\,D^{\alpha |X_{\min}|-a_{\min}},
\ee
where
\be
a_{\min}:=\min_i a_i.
\ee
It follows that, whenever
\be
\alpha |X_{\min}|<a_{\min},
\ee
the distance between $W_\alpha$ and its locally isometric replacement $\widetilde W_\alpha$ vanishes in the large-$D$ limit.

The upshot for the PLC is that local replacement of the random tensors by their isometric parts only works when the foliation expands every tensor sufficiently fast. This is clearly a very restrictive condition. It is stronger than the requirement of an everywhere-expanding foliation since one needs not just expansion everywhere, but expansion at a large enough rate at each tensor. We do not expect this to be a general property, even in gravity.

\subsubsection{A simple example}

Consider the simple example of Sec.~\ref{sec:simpleexample2} for the map $W_R=(R\otimes \mathbf{1}_{D^N})V_\alpha$, where $R:\mathcal{H}_n \to \mathcal{H}_m$ is Gaussian random. Let $T_R$ and $T_{W_R}$ denote the polar parts of $R$ and $W_R$, respectively.

To compare the different notions of approximation, it is useful to consider four related quantities:
\begin{gather}
D_1:=\|W_R^\dagger W_R-\mathbf{1}_{D^{\alpha(n+N)}}\|_2^2,\notag\\
D_2:=\|W_R-T_{W_R}\|_2^2,\notag\\
D_3:=\|W_R-(T_R\otimes \mathbf{1}_{D^N})V_\alpha\|_2^2,\notag\\
D_4:=\|T_{W_R}-(T_R\otimes \mathbf{1}_{D^N})V_\alpha\|_2^2.
\end{gather}
Here $D_1$ measures the failure of $W_R$ to be isometric, $D_2$ measures its distance to the closest isometry, $D_3$ measures the error of the naive local replacement $R\to T_R$, and $D_4$ measures the distance between this naive isometry and the true polar isometry of the restricted map.

The important point is that these quantities scale differently with $D$. For Haar-random $V_\alpha$, one finds that as $D\to\infty$,
\be
\mathbb E_{V_\alpha,R}[D_1]
\sim
D^{-m-n+\alpha(n+N)}+D^{-m-N+2\alpha(n+N)},
\ee
and similarly
\be
\mathbb E_{V_\alpha,R}[D_2]
\sim
\frac14\,\mathbb E_{V_\alpha,R}[D_1].
\ee
Thus both $D_1$ and $D_2$ are small whenever the restricted map $W_R$ is close to an isometry.

By contrast, the naive local replacement is controlled by
\be\label{eq:distancereplacementisomexample}
\mathbb E_R[D_3]
\sim
\frac14\,D^{\,n-m+\alpha(n+N)},
\ee
and one likewise finds
\be
\mathbb E_{V_\alpha,R}[D_4]
\sim
\mathbb E_R[D_3].
\ee
Therefore $D_3$ and $D_4$ are exponentially larger than $D_1$ and $D_2$. Nevertheless, if $m-n$ is sufficiently large, their exponents can still be negative, so these distances also vanish in the large-$D$ limit. In this regime, the naive replacement $(T_R\otimes \mathbf{1}_{D^N})V_\alpha$ approximates $W_R$ less well than the optimal polar isometry, but it remains asymptotically close to it. Its main advantage is that its circuit complexity is additive in the individual TN building blocks. This replacement, however, only works when the local expansion is sufficiently large.

To test the telescoping bound in the present example, note that after the replacement of $R$ by $T_R$ one has $W_R-\widetilde W_R=\mathcal E(R-T_R)$, for the map $\mathcal E(Y)=(Y\otimes \mathbf{1}_{D^N})V_\alpha$. In this case, it follows that the environment map does not exponentially amplify the operator norm
\be
\|\mathcal E(Y)\|_\infty
=
\|(Y\otimes \mathbf{1}_{D^N})V_\alpha\|_\infty
\le
\|Y\|_\infty\,.
\ee
This, in turn, implies the bound
\be
\|\mathcal E\|_{1\leftarrow 1}\le D^{\alpha(n+N)}\,.
\ee
Together with the computed trace distance $\|R-T_R\|_1\sim\frac{4}{3\pi}D^{(3n-m)/2}$ we have that the telescoping bound yields
\be
\|W_R-\widetilde W_R\|_1
\lesssim
D^{\,(3n-m)/2+\alpha(n+N)}\,,
\ee
Thus, the telescoping argument gives the sufficient trace-norm condition $2\alpha(n+N)<m-3n$ for the local replacement to be possible.

\section{Twirled perfect tensor networks}\label{sec:TPTN}

In this section we propose a new class of TNs that satisfy the computational covariance property while being generic: {\it twirled perfect tensor networks} (TPTNs). 

The very definition of computational covariance in Sec. \ref{sec:requirements} requires that, locally, each tensor is perfect. Recall that a perfect tensor is defined in terms of an absolutely maximally entangled (AME) state of the form
\be\label{eq:AME}
\ket{P} = \sum_{i_1,...,i_n=1}^D\,P_{ i_1...i_{n}} \ket{i_1}...\ket{i_n}\,,
\ee
where the coefficients are chosen so that $\ket P$ is normalized. AME states are defined to be maximally entangled across any bipartition in which the chosen set of input legs is no larger than its complement. More precisely, for any subset of $k$ legs with $k \le n/2$, the reduced density matrix obtained by tracing out the remaining $n-k$ legs is maximally mixed
\be\label{eq:AMEcond} 
\rho_{k} := \text{Tr}_{n-k}(\ket{P}\bra{P}) = \frac{1}{D^{k}}\,\mathbf{1}_{D^{k}}\qquad \forall \,k\leq \frac{n}{2}\,.
\ee 
Given any bipartition of the legs into a set of $k$ input legs and a complementary set of output $n-k$ legs, we may reshape the coefficients of $\ket{P}$ into a linear map $P_k: \mathbb{C}^{D^k}\to \mathbb{C}^{D^{n-k}}$, with a suitable global rescaling. By definition, \eqref{eq:AMEcond} implies that such a linear map is exactly isometric (unitary) when the number of input legs is smaller (equal) to the number of output legs, 
\be 
P_k^\dagger P_k = \mathbf{1}_{D^k}\,\quad \text{for } k\leq \dfrac{n}{2}\,.
\ee
Therefore, perfect tensors exhibit isometric covariance and are natural candidates to implement TNs with computational covariance properties.

However, the explicit perfect tensors commonly used in the literature are very special.\footnote{The classification of perfect tensors at large enough bond dimension remains an open problem. For a recent survey on the classification of AME states, see \cite{Rajchel-Mieldzioc:2025mqh}.} Examples of low bond dimension perfect tensors are known, and these are the ones commonly used in holographic TNs \cite{Almheiri:2014lwa,Pastawski:2015qua}. These correspond to the four-qutrit ($n=4$, $D=3$) or six-qubit ($n=6$, $D=2$) AME states. Importantly, in these examples, the state $\ket{P}$ is typically chosen to be a stabilizer state. This implies that such tensors can only generate Clifford circuits or stabilizer quantum error correcting codes. As a result, TNs built entirely from these perfect tensors are highly non-universal: they cannot produce states with genuinely large (exponential) circuit complexity. In fact, the states they generate remain efficiently classically simulable by the Gottesman–Knill theorem. 

For this reason, the stabilizer nature of perfect TNs obstructs any notion of exponential quantum complexity from applying if we want to use these TNs to model gravity. If the entire network is Clifford, then the naive quantum circuit application, brute-force post-selection, or amplitude amplification cannot possibly be the optimal algorithm, as there exists an efficient classical simulation instead. As a consequence, for models of the python's lunch built from these stabilizer tensors, the PLC estimate would not reflect the true computational complexity of the TN.

To overcome this limitation, we will ``twirl'' perfect tensors with local random unitaries. Pictorially:
\begin{center}
    \includegraphics[width=0.3\linewidth]{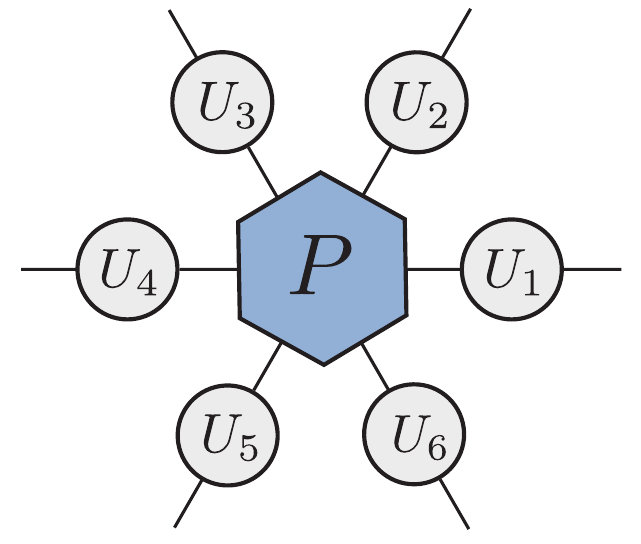}
\end{center}
Twirling simply generates another representative in the same local unitary orbit, so the tensor remains perfect because the AME property is invariant under local unitaries. The difference is that the tensor is no longer special from a computational point of view. The perfect tensor, together with the local unitaries, generate a universal gate set, by the standard result that arbitrary single-qudit unitaries together with an entangling gate are universal. For the small bond dimension examples mentioned above, this removes the Clifford nature and allows such TNs to approximate arbitrary quantum states.

To achieve computational covariance in the sense of Sec.~\ref{sec:requirements}, we assume that the local twirls are simple in the chosen gate model. This is not automatic. In a fixed qudit gate model, a Haar-random unitary typically has complexity polynomial in $D$, so it can contribute to the leading large-$D$ complexity exponent. One way to avoid this extra cost is to count these local unitaries as elementary operations. Another possibility is to draw them from an approximate unitary design generated by polynomial-size circuits. For what follows, an approximate unitary $2$-design is enough.

\subsection{Definition} \label{subsec:TPTN}

Consider a planar graph $\mathcal{G} = \lbrace \mathcal{V}, \mathcal{E} \rbrace $, where $\mathcal{V}$ is the collection of vertices and $\mathcal{E}$ is the collection of edges. For each vertex $x \in \mathcal{V}$, we define a Hilbert space $\mathcal{H}^{x}_{1}\otimes ... \otimes \mathcal{H}^{x}_{i_{n_{x}}}$ of constant bond dimension $|\mathcal{H}^{x}_{i}| = D$ and consider a twirled AME state on each vertex
\be 
|TP_{x}\rangle := \left(\prod_{i=1}^{n_{x}} U^{x}_{i}\right)\ket{P_{x}}\,\,
\ee 
where the $U^{x}_{i}$ are independent Haar random unitaries in $U(D)$.

To define the TPTN on the graph, we contract these vertex states with maximally entangled states $\ket{\phi_e}$ defined in \eqref{eq:epr} at each edge of the graph $e\in \mathcal{E}$, which defines a (generally unnormalized) state $\ket{V_{\text{TPTN}}}$ on the vertices at the boundary of the graph. For the purposes of this paper, we may want to consider the linear map $V_{\text{TPTN}}$ defined by dualizing part of the boundary vertices
\be
\ket{V_{\rm TPTN}}
=\left(\prod_{e\in\mathcal E}\bra{\phi_e}\right)
\left(\prod_{x\in\mathcal V}\ket{TP_{x}}\right)
\in \mathcal H(\partial \mathcal{G})\,.
\ee

Because the Haar measure is invariant under left and right multiplication,
the twirling unitaries acting on internal legs can be combined pairwise.
Each internal edge initially carries two independent unitaries, one from
each adjacent vertex, but these combine into a single random unitary per edge. The result is shown pictorially as:
\begin{center}
\includegraphics[width=0.45\linewidth]{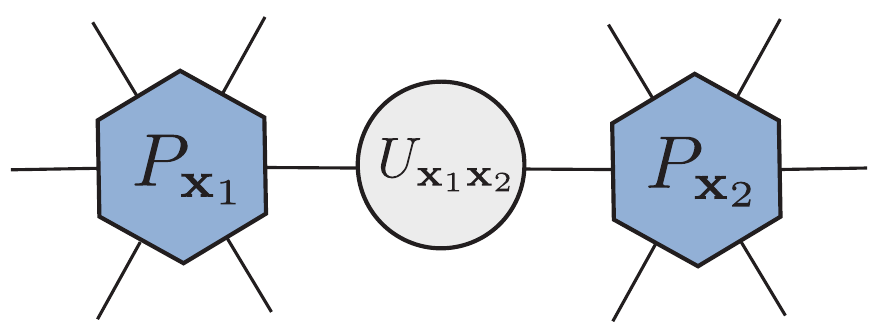}
\end{center}
Therefore, an equivalent definition of the TPTN
is to consider fixed AME states at the vertices
and contract them with maximally entangled random states on the
internal edges,\footnote{The same conclusion holds for more structured
ensembles, such as approximate unitary $k$-designs, provided one restricts
attention to sufficiently low moments.}
\be
|V_{\rm TPTN}\rangle
=
\left(\prod_{a\in \mathcal \partial\mathcal{G}}U_a\right)\left(\prod_{e\in\mathcal E}\bra{U_e}\right)
\left(\prod_{x\in\mathcal V}\ket{P_{x}}\right),
\qquad
\ket{U_e}=(U_e\otimes \mathbf 1_D)\ket{\phi_e},
\ee
where the $U_e, U_a$ are independent Haar-random unitaries associated with
the ``internal'' edges $e\in\mathcal E$, and open ``external'' tensor legs $a \in \partial \mathcal{G}$, respectively.

\begin{figure}[h]
 		\centering
 		\includegraphics[width = .5\textwidth]{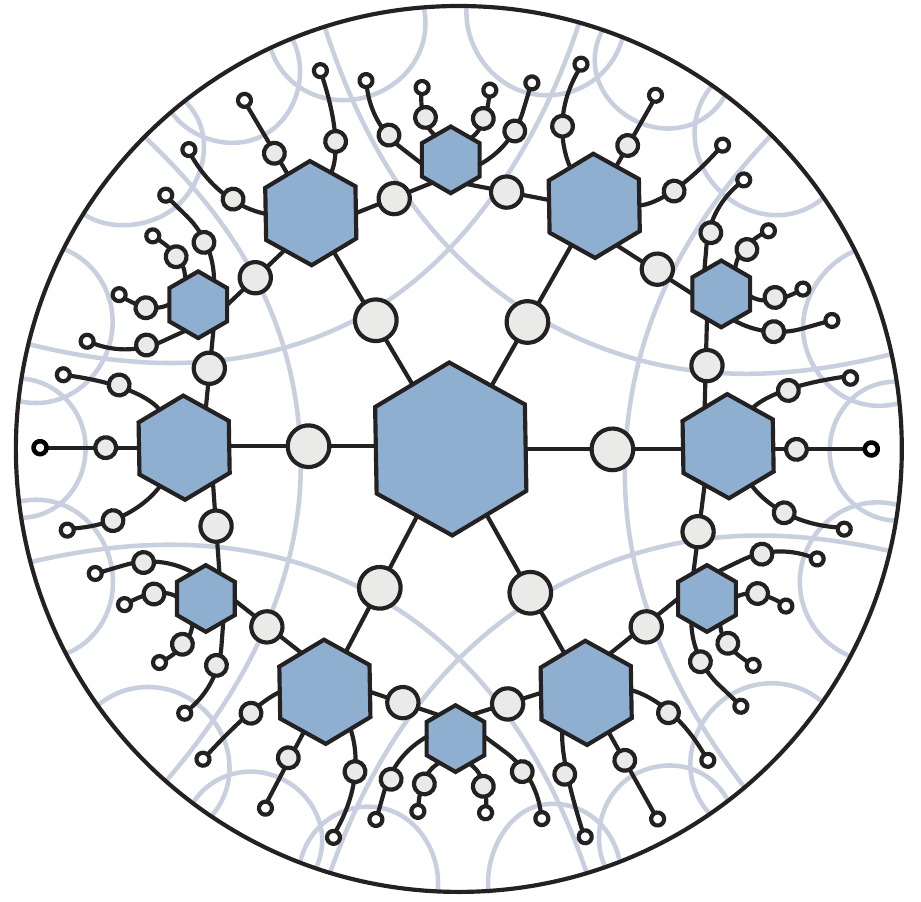}
 		\caption{Example of TPTN consisting of the twirled HaPPY stabilizer qubit state. The TN is computationally covariant and generic. In particular, it is not a stabilizer state.}
 		\label{fig:hexagontwirl}
 \end{figure}

\subsection{Lattice RT formula from spin model}\label{sec:RTTPTN}

We now show that holographic TPTNs saturate lattice versions of the RT formula for arbitrary boundary subregions in the large bond dimension limit. This contrasts with holographic stabilizer perfect TNs \cite{Pastawski:2015qua}, where the ``greedy algorithm'' fails to converge to the RT (or QES) cut for general boundary subregions.

The derivation for TPTNs closely parallels that of RTNs. In particular, subsystem R\'{e}nyi entropies of the state $|V_{\rm TPTN}\rangle$ can be mapped to an effective classical spin model on the network. There are several important differences from the usual RTN derivation, which we now explain.

To begin, we consider the TPTN state $|V_{\rm TPTN}\rangle$ and divide the boundary Hilbert space $\mathcal H(\partial \mathcal{G})\simeq \mathcal H_\A\otimes \mathcal H_{\A^c}$. The reduced density matrix associated with the subregion $\A$ is
\be\label{eq:densitymatrixA}
{\rho}_{\A} = \dfrac{1}{\mathcal{N}}\,\text{Tr}_{\A^c}\left(|V_{\rm TPTN}\rangle\langle V_{\rm TPTN}|\right)\,,
\ee
where $\mathcal{N} =\bra{V_{\rm TPTN}}\ket{V_{\rm TPTN}}$ is the normalization of the state defined by the TPTN.

The entanglement entropy is $S(\rho_{\A})=-\mathrm{Tr}(\rho_{\A}\log\rho_{\A})$. In the large bond-dimension limit, averaging over the ensemble of twirling unitaries leads to
\be
\overline{S(\rho_{\A})} \;\underset{D \to \infty}{\sim}\; |X_{\A}| \log D ,
\ee
where $|X_{\A}|$ denotes the minimal cut of the graph homologous to $\A$.

To see this, as it is standard, one starts from the $n$-th R\'enyi entropies of the density matrix $S_n(\rho_{\A}) = -\tfrac{1}{n-1}\log \text{Tr}(\rho_{\A}^n)$. We will compute the ``annealed'' average R\'enyi entropy
\be 
S_n^{\rm ann}(\rho_{\A})  = \dfrac{1}{1-n}\log\overline{\text{Tr}(\rho_{\A}^n)}\,,
\ee
and then proceed by analytic continuation to $n\to 1$, since $S_n(\rho_{\A})\to S(\rho_{\A})$ in this limit.\\

{\bf Spin model for second R\'{e}nyi entropy.} Let $\widehat{\mathcal E} = \mathcal E \cup \partial \mathcal G$ denote the
extended edge set (internal edges together with boundary legs). For $n=2$, we will now show that there is an effective classical $\mathbb{Z}_2$ spin model defined on $\widehat{\mathcal E}$, with a spin variable $s_e \in \{\pm 1\}$ for each edge $e \in \widehat{\mathcal E}$, which reproduces the purity in the large bond dimension limit. The inverse temperature of the spin model is related to the bond dimension via $\beta =\log D$. More precisely,
\be\label{eq:effstatmech} 
\overline{{\rm Tr}(\rho_{\A}^2)} 
\;\underset{\beta \to \infty}{\sim}\;
\sum_{\{s_e = \pm 1\}^*} 
e^{-\beta E[\{s_e\}]}\,,
\ee
for an effective Hamiltonian $E[\{s_e\}]$ defined below, in \eqref{eq:effHTPTNz2}.

We now provide a step by step derivation of this effective Hamiltonian, but we assume familiarity with \cite{Hayden:2016cfa}. First, one starts from the purity and introduces four replicas of the TN to write
\be 
\text{Tr}(\rho_{\A}^2) = \text{Tr}_{\A \A}\left(\mathcal{F}\left(\rho_{\A} \otimes \rho_{\A}\right)\right)\,,
\ee 
where $\mathcal{F}$ is the swap operator. The second step is to take the ``annealed'' average of the numerator and denominator associated with the density matrix \eqref{eq:densitymatrixA} in the last expression separately.

This involves doing the following Haar integral at each edge
\be\label{eq:n2Haar} 
\int \text{d}U \left( U \otimes U^* \otimes U \otimes U^*\right) = \dfrac{D^2}{D^2-1}\left(\ket{\mathcal{I}}\bra{\mathcal{I}} + \ket{\mathcal{F}}\bra{\mathcal{F}}   - \frac{1}{D} \left(\ket{\mathcal{I}}\bra{\mathcal{F}} + \ket{\mathcal{F}}\bra{\mathcal{I}}\right)\right)\,,
\ee 
for the invariant states in four replicas
\be
\ket{\mathcal{I}} = \ket{\phi}_{1\bar{1}}\ket{\phi}_{2\bar{2}}\,,\qquad \ket{\mathcal{F}} =  \ket{\phi}_{1\bar{2}}\ket{\phi}_{2\bar{1}}\,.
\ee
These states generate the invariant subspace under left/right unitary multiplication by $U_0\otimes U_0^* \otimes U_0\otimes U_0^*$ for fixed $U_0\in U(D)$.\footnote{By the Schur-Weyl duality, the invariant states also generate a representation of the permutation group $\text{Sym}(2)$.} In particular, the right hand side of \eqref{eq:n2Haar} is the orthogonal projector to this subspace. The coefficients in the expression are known as the Weingarten functions, and they are computed as the inverse of the Gram matrix of overlaps between the invariant states. In \eqref{eq:n2Haar} the first two terms are the ``Gaussian'' contractions, while the last two terms with a negative sign are subleading Weingartens, suppressed by $1/D$.

Due to these four terms, for each edge of the TPTN, we must assign {\it two} spin variables, one for each index of the unitary. In particular, if we make the assignment $\ket{\mathcal{I}}\to -1$ and $\ket{\mathcal{F}} \to +1$ each average \eqref{eq:n2Haar} inside the second R\'enyi receives four contributions of the form 
\be 
\dfrac{D^2}{D^2-1}\sum_{\{s_1,s_2 =\pm 1\}} \text{sign}(s_1s_2) e^{-\beta  \frac{1-s_1s_2}{2}}\left(\cdot\right)\,,
\ee
where recall that $\beta = \log D$. Because of the sign term, an additional distinction with RTNs is that the average purity calculation for a TPTN takes the form of $\mathbb{Z}_2$-graded partition function, with Boltzmann weights that carry a parity sign due to exact unitarity. 

We show these two spins diagrammatically as:
\begin{center}
\includegraphics[width=.5\linewidth]{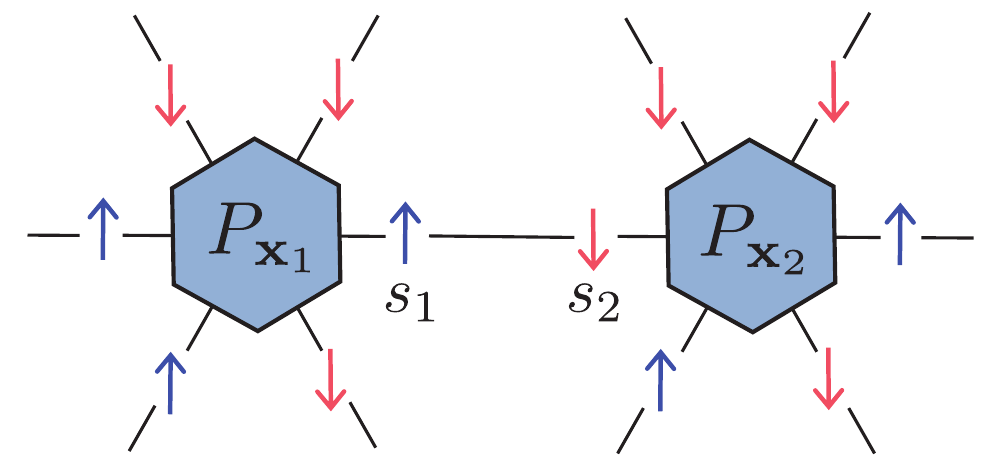}    
\end{center}

Let us now turn to how the edge spins interact at the vertices. We retain a single Ising variable on each extended edge, $s_e \in \{\pm 1\}$ for all $e \in \widehat{\mathcal E}$, namely the spin adjacent to the vertex (among the two originally associated with each edge). For each vertex $x \in \mathcal V$, define the incidence set $\delta(x) = \{ e \in \widehat{\mathcal E} : e \text{ is incident on } x \}$ and let the number of incident edges be $n_x = |\delta(x)|$. The resulting effective vertex interaction Hamiltonian is
\be\label{eq:bulkintTPTN2}
\beta E_{\rm int}[\{s_e\}]
=
\sum_{x \in \mathcal V} S_2\!\left(
P_x \,\middle|\,
A_x(\{s\})
\right)\,,
\ee
where the subset of legs at vertex $x$ selected by the spin configuration is
\begin{equation}
A_x(\{s_e\})
:=
\{\, e \in \delta(x) \;:\; s_e = +1 \,\}.
\end{equation}
Here $S_2(P_x \mid A_x)$ denotes the second Rényi entropy of the perfect tensor
state $\ket{P_x}$ reduced to the subsystem consisting of the legs in $A_x$.\footnote{So far, the construction is fairly general: twirling an arbitrary TN with local random unitaries yields an effective spin model that saturates the RT formula under general conditions. We nevertheless restrict attention to TPTNs because of their isometric covariance properties.} Recall that this form of the interaction follows from how loops of the four replicas of the TPTN are contracted when doing the Haar integrals.

For a perfect tensor at each vertex, it follows that the interaction energy is given in terms of the total magnetization at the vertex
\be 
S_2\!\left(
P_x \,\middle|\,
A_x(\{s_e\})\right) = \dfrac{1}{2}\log D\, (n_x -|\mathcal{S}_x(\{s\})|)\,,\qquad \mathcal{S}_x(\{s_e\}) := \sum_{e \in \delta(x)} \!s_e\,.
\ee 
The interaction \eqref{eq:bulkintTPTN2} is thus ``ferromagnetic'' and favors an ordered configuration of the edge spins at low temperatures.

A useful simplification occurs in the large bond-dimension limit $D \to \infty$. Fix all spins except the two living on a given edge and examine their contribution to the graded partition function. After canceling terms with opposite signs, the leading local contribution can be written using a single spin per edge, representing aligned configurations
\be\label{eq:localeffz}
\sum_{\{s_1,s_2 = \pm 1\}}\sign(s_1s_2) e^{-\beta E(s_1,s_2)} \;\underset{\beta \to \infty}{\sim}\; \sum_{ \{s=\pm 1\}^\ast}\,e^{-\beta E(s,s)}\,,
\ee
where
\be 
E(s_1,s_2) = -\frac{1}{2}\left[|\mathcal{S}_1(s_1)| + |\mathcal{S}_2(s_2)| + (s_1s_2-1)\right]\,.
\ee
is the effective Hamiltonian obtained by restricting to the two spins on the edge, with all other spins held fixed.

To see \eqref{eq:localeffz}, suppose without loss of generality that a configuration with $s_1 = -s_2 = +1$ minimizes the energy. Flipping one of the spins changes the magnetization of the adjacent vertices in multiples of $2$, so the configurations with $s_1 = s_2 = \pm 1$ necessarily have the same energy. In this situation, there is therefore a threefold degeneracy of the minimal energy level. However, in the graded partition function \eqref{eq:localeffz}, the anti-aligned configuration $s_1 = -s_2$ carries an opposite sign  because it comes from a negative Weingarten term and cancels against one of the aligned contributions. The net result is a single surviving lowest-energy contribution with $s_1 = s_2$. On the other hand, if only aligned configurations minimize the energy, the anti-aligned configurations are suppressed at zero temperature, but both aligned configurations may contribute.  The star in \eqref{eq:localeffz} or in \eqref{eq:effstatmech} represents that configurations related by single spin flips, which are degenerate, should only be considered once.

Hence, in the $D \to \infty$ limit, the model reduces to a classical $\mathbb{Z}_2$ spin system with one spin per edge. The total spin Hamiltonian then takes the form
\be\label{eq:effHTPTNz2}
E[\{s_e\}] = -\frac{1}{2}\left[ 
\sum_{x \in \mathcal V}
\left(
|\mathcal{S}_x(\{s_e\})| - n_x
\right)
+
\sum_{e\in \partial \mathcal{G}} (h_e s_e - 1)
\right]\,,
\ee
where we have also included the usual boundary pinning fields that implement the appropriate replica boundary conditions,
\be\label{eq:pinningm2} 
h_e = \begin{cases}
    +1, & e \in {\A}, \\[0.2cm]
    -1, & e \in {\A^c}.
\end{cases}
\ee 
We show the corresponding spin model in Fig. \ref{fig:spintptn}.

\begin{figure}[h]
 		\centering
 		\includegraphics[width = .5\textwidth]{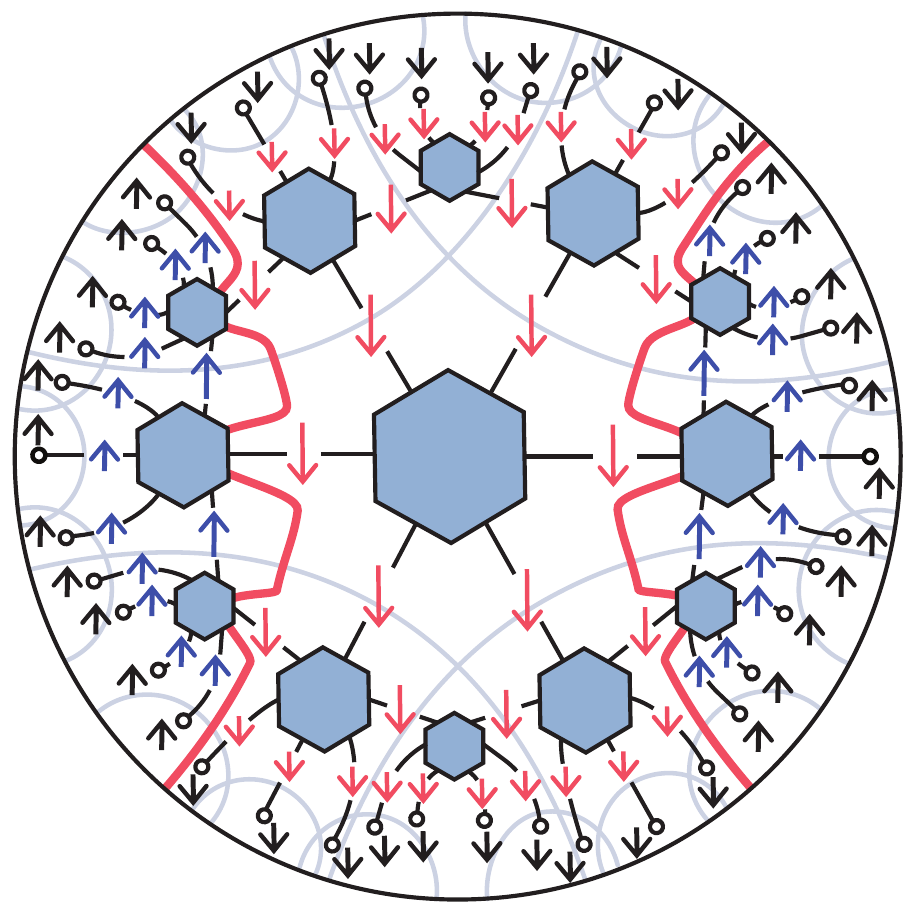}
 		\caption{Effective stat mech model for the purity of $\rho_{\A}$ for $\A$ the union of two disjoint intervals in the connected phase. Spins are associated with edges.  Each vertex contributes to the energy with minus the absolute total magnetization. The boundary pinning field is defined by the choice of $\A$. The ground state configuration (red domain wall) passes through the vertices and reproduces the lattice RT formula. Some local fluctuations of the domain wall cancel through subleading Weingarten contributions; uncanceled degeneracies remain.}
 		\label{fig:spintptn}
 \end{figure}

The leading cancellations described above must be included before dropping $1/D$ corrections. Although the negative Weingarten coefficients carry an extra factor of $1/D$, their tensor contractions can compensate for it, allowing them to cancel terms at the leading order. However, if one wishes to study subleading $1/D$ corrections (such as those relevant in the previous section when testing the exact unitarity of the TN) then the full Weingarten expression must be retained. In that case, the graded nature of the partition function becomes essential: it is precisely this graded structure that ensures the cancelation of unwanted terms and ultimately restores the fact that TPTNs define exact isometries along expanding foliations.\\
 
{\bf Average purity from the ground state.} Since the Hamiltonian is local in the graph, if the planar TN model is local, then the ground state configuration will contain a single domain wall separating two domains: one of $+1$ spins and one of $-1$ spins. In this case, the domain wall passes through the vertices of the network
We find, as expected, in the large bond dimension limit, the ground state energy of the TPTNs is 
\be
    E_{\rm GS} = |X_{\A}|\,.
\ee 
The cancellation in \eqref{eq:localeffz} removes one of the two aligned contributions when an anti-aligned configuration has the same minimum energy. Other degenerate aligned configurations can still contribute. If the total ground-state degeneracy after these cancellations is $g>0$, the purity scales as $gD^{-|X_{\A}|}$ and the second R\'enyi entropy acquires a term $-\log g$. In most situations this is an $O(1)$ correction and does not change the leading $\log D$ entropy.

From the ground state energy, allowing a finite surviving ground-state degeneracy, we obtain the annealed second R\'{e}nyi entropy
\be\label{eq:S_2TPTN}
    S_2^{\rm ann}(\rho_{\A})  \;\underset{D \to \infty}{\sim}\; |X_{\A}|\log D \,.
\ee

{\bf Entanglement entropy.} The result above fixes the leading entanglement entropy, since the von Neumann entropy satisfies 
\be\label{eq:ineq}
    |X_{\A}| \log D \,\geq S(\rho_{\A}) \geq S_2(\rho_{\A})\,.
\ee
The upper bound is the classic bound on TNs, while the lower bound follows from the monotonicity of the R\'enyi entropies. We now take the average and assume that annealed and quenched averages lead to the same answer (which is expected to be true very generally as $D\to \infty$). Then, the result for the R\'enyi \eqref{eq:S_2TPTN} implies the saturation of the inequalities \eqref{eq:ineq}, which yields the lattice RT formula
\be 
\overline{S(\rho_{\A})} \;\underset{D \to \infty}{\sim}\; |X_{\A}| \log D\,.
\ee 
The leading RT entropy does not imply a flat entanglement spectrum  at subleading orders.\\

{\bf Spin model for higher R\'{e}nyi entropies.} Higher R\'enyis for TPTNs also admit effective spin model interpretations, which can be derived starting from
\be 
\text{Tr}(\rho_{\A}^n) = \text{Tr}_{\A... \A}\left(C_{\eta}\left(\rho_{\A} \otimes ...\otimes \rho_{\A}\right)\right)\,,
\ee 
where $\eta = (12...n)\in \text{Sym}(n)$ is the cyclic permutation, and $C_\eta$ is the corresponding permutation operator of the $n$ replicas of $\A$.

Taking the annealed averages and performing the corresponding Haar integral, one finds that the corresponding stat mech model contains two $\text{Sym}(n)$-valued spins per edge corresponding to each permutation, with a ferromagnetic interaction of the form
\be 
E_{12} = n -\,\# \text{cycles}(\sigma_1^{-1}\sigma_2)\,,\qquad \sigma_1,\sigma_2 \in \text{Sym}(n)\,.
\ee

The vertex interactions are no longer R\'enyi entropies, but they are generalizations (sometimes dubbed ``multi-entropies'' \cite{Gadde:2022cqi}) that depend on more complicated contractions of the AME state in $n$ replicas. The partition function is also graded due to the subleading Weingartens, but we expect that this leads to a similar effect where, at $D\to \infty$, one can consider a single $\text{Sym}(n)$-valued spin per edge.

 A more extensive analysis of higher R\'enyi entropies is beyond the scope of this paper. A situation where the entanglement spectrum is exactly flat is if the TPTN is everywhere expanding from the minimal cut $X_\A$ to both boundary regions $\A$ and $\A^c$. In this case, the reduced density matrix $\rho_{\A}$ is an isometric embedding of the maximally mixed state on $\mathcal{H}(X_\A)$, and thus has a flat entanglement spectrum.\footnote{For these subregions, the ``greedy algorithm'' of \cite{Pastawski:2015qua} for the TPTN converges to the minimal cut.}

\subsection{Other properties} 

Here we outline other nice properties of TPTNs.

{\bf Covariant complexity exponent.} For TPTN models of the python's lunch we identify $\mathcal H(\partial \mathcal{G})\simeq \mathcal H(X_\A)\otimes \mathcal H(X_\A^c)$, and define the dual map $V_{\rm TPTN}:\mathcal H(X_\A)\to \mathcal H(X_\A^c)$ from the condition
\be\label{eq:VstateTPTN_action}
V_{\rm TPTN}\ket{\psi}
\,\propto\,
\left(\bra{\psi^*}_{X_\A}\otimes \mathbf 1_{X_\A^c}\right)
\ket{V_{\rm TPTN}}\,,
\ee
for every $\ket{\psi} \in \mathcal{H}(X_\A) $. The global proportionality constant will be fixed so that 
\be 
\overline{V_{\rm TPTN}^\dagger V_{\rm TPTN}} = \mathbf{1}_{D^{|X_{\A}|}}\,.
\ee 

By definition, TPTNs are computationally covariant. According to the considerations of Sec. \ref{sec:requirements} the map $V_{\rm TN}$ therefore has a complexity exponent bounded by the ``non-contracting bulge'' \eqref{eq:PLClike2} in general situations. In general, this cut need not coincide with the minimax cut, since its definition relies on the existence of everywhere non-contracting foliations.

{\bf Non-stabilizer holographic quantum error correcting codes.} Perfect tensors define quantum error correcting codes. For example, the six qubit AME stabilizer state defines a $[[5,1,3]]$ qubit code when interpreted as a map from one leg to the other five. The concatenation of such maps in a hyperbolic tiling defines the HaPPY stabilizer code \cite{Pastawski:2015qua}. This code provides a concrete, simple model of holographic quantum error correction, exhibiting complementary recovery for boundary subregions. 

If we twirl the corresponding HaPPY code, we obtain a more general, non-stabilizer TN, as shown in Fig. \ref{fig:pentagontwirl}. Globally, the twirling unitaries can be pushed through the network to the boundary, so the resulting code map is unitarily equivalent to the original one. Since the Knill–Laflamme conditions for quantum error correction are invariant under unitary conjugation, such a TPTN therefore defines a quantum error-correcting code.

\begin{figure}[h]
 		\centering
 		\includegraphics[width = .5\textwidth]{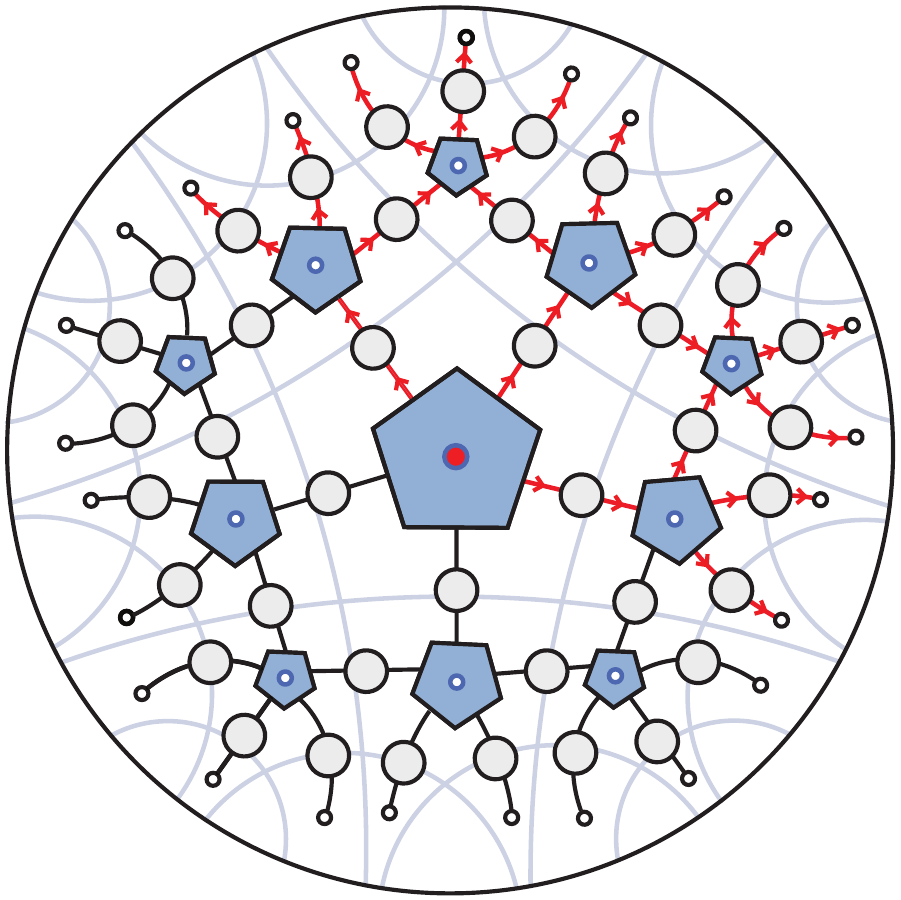}
 		\caption{Operator pushing in the twirled HaPPY quantum error correcting code.}
 		\label{fig:pentagontwirl}
 \end{figure}

However, because the physical Hilbert space is transformed from the original one by a unitary transformation, it is not immediately clear that the twirled code continues to exhibit complementary recovery for boundary subregions. For boundary regions for which the greedy algorithm of \cite{Pastawski:2015qua} yields the entanglement wedge, the twirling unitaries in the entanglement wedge can be pushed to the boundary and absorbed locally into those subregions. In such cases, the correctability of the corresponding erasures (and thus the associated entanglement wedge reconstruction properties) remains intact. For more complicated boundary regions or other network geometries, however, this pushing argument need not preserve the boundary tensor factorization, and the geometric interpretation of complementary recovery may be altered. Thus, while twirling preserves most of the underlying quantum error-correcting structure, it generically removes special algebraic features such as stabilizer structure. We expect that including the bulk contribution to the entanglement entropy of TPTNs will imply reconstruction beyond the greedy surface, at least for bulk code subspaces of fixed $O(1)$ entropy. We have not estimated the errors in this extension, which may become important for large code subspaces or large entanglement wedges. Figure~\ref{fig:pentagontwirl} illustrates operator pushing within the greedy wedge.

\section{Complexity of local postselection}
\label{sec:localshrink}

\begin{figure}[h]
    \centering
    \includegraphics[width=.85\linewidth]{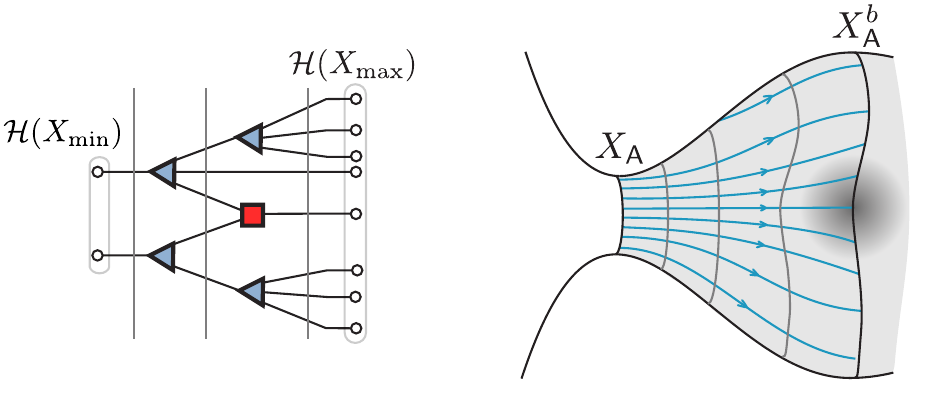}
    \caption{Left: A computationally covariant TN may admit slicings that are globally expanding but not everywhere locally expanding. In the example shown, pushing the cut across the red tensor increases the total cut size, so the associated map may still be an approximate isometry, but the move is locally shrinking and therefore incurs an additional postselection cost. Right: In the continuum geometric picture, the PLC surface may similarly contain a local bump where the areas locally decrease along a particular slicing. However, between the RT (or constriction) and the bulge one can always choose an everywhere-area-increasing foliation, generated by mean curvature flow. This local shrinking obstruction should therefore be absent in gravity.}
    \label{fig:everywheregrowslice}
\end{figure}

Computational covariance is not, by itself, sufficient to guarantee the PLC exponent in a TN model. One must still contend with a genuinely discrete obstruction: local shrinking along the foliation of the network may be unavoidable. In particular, the minimax prescription used to identify the bulge covariantly may require a slicing that is not everywhere area-increasing towards the bulge. An example is shown in Fig. \ref{fig:everywheregrowslice}. Whenever this happens, a unitary implementation of the layers of the foliation requires postselection on the locally removed degrees of freedom. As a result, the nominally ``simple'' steps of the PLC algorithm themselves become exponentially costly in $\log D$, and this potentially changes the complexity exponent.

A simple way to quantify this effect is to count the total amount of local postselection along a given foliation $\gamma$. For a computationally covariant TN, which we may take for concreteness to be a TPTN, we define
\be
N_{\rm post}(\gamma)\;:=\;
\sum_{t=0}^{T-1}\ \sum_{v\in\mathcal V_t}
\bigl(n_t(v)-n_{t+1}(v)\bigr)\Theta\bigl(n_t(v)-n_{t+1}(v)\bigr)\,,
\ee
where $n_t(v)$ denotes the number of edges in the cut $X_t\subset \gamma$ incident on the vertex $v$, and $\Theta$ is the Heaviside function. This quantity measures, step by step, how many local legs are lost as the cut evolves through the network. 

If one then implements the map by explicitly measuring all degrees of freedom that must be postselected, and assumes sufficient genericity of the local postselection operations, one is led to the estimate
\be\label{eq:PLClike3}
\log_D \Co(V_{\rm TPTN})\lesssim \min_{\gamma}\,\{
N_{\rm post}(\gamma)\}\,.
\ee
This should be viewed as a natural discrete refinement of the PLC exponent, modified to account for the unavoidable local shrinkage of the foliations in computationally covariant TNs. 

The inequality reflects the possibility that a more efficient implementation may exist. Amplitude amplification is an example. In particular, local postselections occurring ``in series'' along a growing slicing may sometimes be handled with only an additive cost per layer, as in the multi-lunch version of the PLC with robust oblivious amplitude amplification applied at each step \cite{Engelhardt:2021qjs}. We will remain agnostic about such improvements here and regard the right-hand side of \eqref{eq:PLClike3} as the natural generic answer in the absence of additional structure.

More generally, one might view \eqref{eq:PLClike3} as defining a notion of exponential complexity in any computationally covariant TN model of gravity, and not only in TN models of the python's lunch. For instance, even a TN model of a cylindrical wormhole with sufficiently irregular local structure should be expected to have exponentially large unitary complexity, with an exponent controlled by the cumulative amount of local postselection along an optimal foliation. Such TN geometries may provide good models of ``caterpillars'' \cite{Magan:2024aet,Magan:2025hce}.

\subsection{No local penalty in gravity}

This is one important way in which TN models of the python's lunch differ from the continuum geometric picture. In a smooth bulk geometry, one can always use spatial diffeomorphisms to find an everywhere locally area-increasing foliation connecting the constriction (or RT) to the bulge, so there is no clear analog of the complexity penalty associated with local postselection.

In Appendix~\ref{app:contraction} we prove Theorem \ref{thm:foliation}, which makes this intuition precise. The basic idea is simple. Given two extremal hypersurfaces in the same homology class (for instance, the bulge and the constriction, or the bulge and the RT surface), one can construct an interpolating foliation whose mean curvature has a definite sign everywhere. Physically, this is the continuum statement that the slices move monotonically toward the bulge, with their area locally increasing everywhere in that direction. An example of such flow is illustrated in Fig.\ \ref{fig:everywheregrowslice}.

The construction is based on mean curvature flow (MCF), initiated along the unstable mode of the bulge surface, so that the foliation is generated directly by the flow. This theorem supports the idea that the extra penalty captured by \eqref{eq:PLClike3} in TN models of the python's lunch is a genuinely discrete effect and is absent in gravity.

\section{Discussion}
\label{sec:discussion}

In this work, we have revisited holographic TN models of the black hole interior with the goal of sharpening the Python's Lunch Conjecture (PLC). While the PLC was originally motivated by TN models, we have shown that its quantitative validity in local TN models is far from automatic. In particular, although the presence of a bulge generically signals exponential complexity, the precise PLC exponent requires additional structural assumptions about the underlying TN.

Our main conceptual point is that the PLC implicitly assumes a notion of \emph{computational covariance}, which requires simple isometricity/unitarity under everywhere-non-contracting slicings of the network. When this condition is satisfied, one has the freedom to foliate the TN of the python's lunch in different ways to find a foliation that minimizes the amount of postselection, with the hope that the optimal foliation includes the minimax cut (bulge) of the PLC proposal.

We have proposed a new class of TNs, dubbed \emph{twirled perfect tensor networks} (TPTNs), which are computationally covariant by construction while remaining generic and universal for quantum computation. We have shown that TPTNs combine the desirable holographic features of perfect TNs and RTNs. In particular, they saturate lattice versions of the RT formula for arbitrary boundary subregions while being computationally covariant. At the same time, our analysis in Sec. \ref{sec:RTN} makes clear that the computational covariance property is a nontrivial and restrictive condition, rather than a generic feature of local TNs.

Finally, we have emphasized that even within the best-case scenario of computationally covariant TNs, there is an intrinsic additional source of exponential complexity in TNs due to local shrinking (postselection), which cannot be avoided by a clever choice of tensor. This reflects the fact that generally the ``non-contracting bulge'' cut is not the minimax cut of the network. Nevertheless, we have argued that this additional source of complexity is absent in the continuum geometric description in gravity.

Several open questions remain. 
\begin{itemize}
    \item A basic issue concerns the role of computational covariance in holography itself. If this property is necessary for the interior, one may ask whether it should also hold in models of the exterior geometry. The canonical exterior model, MERA, does not satisfy computational covariance, highlighting a tension between preparing well-defined quantum states of a CFT on a lattice and reproducing holographic considerations beyond that. Understanding whether TPTNs, or suitable modifications thereof, can be extended to approximate well-defined states of the boundary theories would substantially strengthen their holographic interpretation. Even identifying obstructions to this goal would shed light on the limitations of TN approaches to gravity.
    \item Additionally, it would be interesting to investigate TPTNs in more detail, not necessarily in a holographic context. One question is to rigorously show whether TPTNs exhibit flat entanglement spectra. Another possibly interesting question is whether TPTNs efficiently generate ensembles of quantum states/unitaries that approximate random states, such as quantum state $k$-designs or unitary $k$-designs. More generally, while we focused on perfect tensor vertices, the twirling procedure can be applied to broader classes of tensors, so studying the properties of more general twirled TNs may be of broad interest.
        
    \item Finally, it would be important to understand whether TPTN states can be implemented experimentally in near-term quantum devices. Perfect tensors can be chosen to be simple stabilizer tensors, and the twirling operation is local, making the construction relatively natural from an implementation point of view. This connects to the broader program of probing quantum-gravity-inspired phenomena in the lab \cite{Brown:2019hmk}; see, for example, the recent work of \cite{Jeffrey:2026oyb}. In this sense, TPTNs may provide a simple platform where several aspects of holography can be studied together experimentally.
\end{itemize}

\section*{Acknowledgments}

We thank Chris Akers, Jos\'{e} Barb\'{o}n, Charles Cao, Jeevan Chandra, Gong Cheng, Tom Hartman, Patrick Hayden, Kenneth Higginbotham, Cathy Li, Geoff Penington, Pratik Rath, Arvin Shahbazi-Moghaddam, Chris Waddell, and Beni Yoshida for discussions. We also thank Chris Akers for helpful comments on the manuscript. We acknowledge support from the U.S. Department of Energy through DE-SC0009986, QuantISED DE-SC0020360, and GeoFlow DE-SC0019380. We acknowledge the ``Informing Gravity Theory Through Quantum Simulation Experiments'' Q-FARM workshop at Stanford, where part of this work was presented. This work is supported in part by the Leinweber Institute for Theoretical Physics (MS). We thank Anthropic’s Claude models and OpenAI’s models for proofreading, suggesting references, and helping improve the presentation of the manuscript.

\appendix

\section{When the PLC overestimates complexity}\label{app:overestimate}

It is possible for the PLC to predict a complexity that exceeds the maximal allowed complexity, at least within the qubit model of the problem. The potential tension arises because a system of $n$ qubits can reach any state in its Hilbert space using  $O(2^{n})$ two-qubit gates. Likewise, an arbitrary unitary acting on $n$ qubits can be implemented with $O(4^{n})$ such gates. More generally, an isometry from $n$ qubits to $m$ qubits can be realized using $O(2^{n+m})$ gates. These bounds therefore constrain the allowed complexity of the python’s lunch map, represented by the approximately isometric map $V: \mathcal{H}(X_{\rm A}) \to \mathcal{H}(X^c_{\rm A})$, once the effective number of qubits is identified with the generalized entropies of the minimal QES and the constriction, namely $n \log 2 = S_{\rm gen}(X_\A)$ and $m \log 2 = S_{\rm gen}(X^c_\A)$. In particular, in geometries with a sufficiently large bulge, the PLC may predict a complexity that exceeds the upper bound $O(2^{n+m})$.

The upper bounds follow from constructive synthesis; circuit counting instead gives lower bounds. An isometry $W$ from $n$ qubits to $m$ qubits is a $2^{m} \times  2^{n}$ matrix satisfying $W^\dagger W = \mathbf{1}_{2^{n}}$. The real dimension of the manifold of isometries is $2^{n+m+1}-2^{2n}$. Let $N_\epsilon$ be the number of $\epsilon$-balls required to cover this manifold. At fixed sufficiently small $\epsilon$, $\log N_\epsilon$ grows at least proportionally to its dimension.

We can explore this space starting from a fixed reference isometry and acting with two-qubit gates. If we again assume a single non-symmetric two-qubit gate type, then at each step there are of order $m(m-1)$ possible placements of the gate on the output qubits. After $\Co$ gates, the number of distinct isometries reachable is bounded by
\be
\text{number of isometries with complexity } \Co
\;\le\;
\left(m(m-1)\right)^{\Co} \sim m^{2\Co}.
\ee

To cover all $\epsilon$-distinguishable isometries, it is necessary that\footnote{If the gate set includes continuous parameters, the regulator dependence of the maximal complexity may disappear.}
\be
m^{2\Co_{\max}}\gtrsim N_\epsilon
\quad\Rightarrow\quad
\Co_{\max}\gtrsim\frac{\log N_\epsilon}{2\log m}\,.
\ee

In the common regime $m\gg n$, this gives the lower bound
\be
\Co_{\max}\gtrsim\frac{2^{n+m}}{\log m}\,.
\ee
The independent upper bound $\Co_{\max}=O(2^{n+m})$ follows from constructive decompositions of isometries into elementary gates \cite{Iten:2016uzx}. 

We can apply this bound to the PLC and observe that the PLC prediction can, in many cases, exceed the upper bound. We rewrite the PLC prediction here
\be\label{eq:PLC2}
\Co_{\text{PLC}} \sim  \exp\left(\dfrac{S_{\text{gen}}(X_\A^b)-S_{\text{gen}}(X^c_\A)}{2}\right)\,.
\ee

We saw previously that $n \log 2 = S_{\rm gen}(X_\A)$ and $m \log 2 = S_{\rm gen}(X^c_\A)$ is the number of qubits. Therefore, if $\Co_{\text{PLC}}>\Co_{\rm max}$, we know that the PLC prediction cannot be the complexity of the isometry. The condition for this to happen is
\begin{gather}     
    \exp\left(\frac{S_{\text{gen}}(X_\A^b)-S_{\text{gen}}(X^c_\A)}{2}\right) > \exp\left(S_{\text{gen}}(X^c_\A) + S_{\text{gen}}(X_\A)\right) \,,\nonumber \\
    \Rightarrow S_{\text{gen}}(X_\A^b) > 3S_{\text{gen}} (X^c_\A)+ 2 S_{\text{gen}}(X_\A)\,.\label{eq:conditionmaxC}
\end{gather}

In case \eqref{eq:conditionmaxC} is satisfied, it is not clear what the more efficient circuit to approximate the python's lunch map $V$ is, but one must exist. Therefore, any geometry that has a relatively large bulge, with an area (or generalized entropy) satisfying \eqref{eq:conditionmaxC}, will not obey the PLC. In contrast to exponentially long wormholes \cite{Susskind:2020wwe}, this regime is far less exotic in gravity. The optimal circuit representation of $V$ might still correspond to an exponentially long wormhole.

\section{Details on planar python's lunches}
\label{app:detailsplanar}

In this appendix, we provide the details of the derivation of some formulas used in Sec. \ref{sec:markov}.

\subsection{PLC exponent}

 We start with the derivation of the PLC exponent \eqref{eq:areadifference}, where $\A$ is the full boundary. As we have seen in Sec. \ref{sec:markov}, the bulge $X^b_\A$ breaks the isometries of the IR regulator. In two dimensions, one can parametrize $\rho(\sigma), x(\sigma)$ and write the area functional
\be\label{eq:area1}
\text{Area} = \int \text{d}\sigma \sqrt{\dot{\rho}^2 + r^2 \dot{x}^2}\,,
\ee 
 where the dot represents $\text{d}/\text{d}\sigma$. Taking $\sigma$ to be the proper length of the bulge, with $\sqrt{\dot{\rho}^2 + r^2 \dot{x}^2}=1$ , $\rho(\sigma)$ is determined by the equation of motion of a non-relativistic particle moving in one dimension with zero total energy
 \be\label{eq:eomgeodesics}
 \dot{\rho}^2 + V_{\text{eff}}(\rho) = 0\,,
 \ee 
 subject to the effective potential
 \be\label{eq:eomgeodesicsveff}
 V_{\text{eff}}(\rho) = \dfrac{r_m^2}{r^2(\rho)}-1\,.
 \ee 

 The parameter $r_m = r^2 {\dot x}$ is a constant of motion along the trajectory. The endpoints of the trajectory will lie at the turning points of $V_{\text{eff}}(\rho)$, which lie at the same value of the radial coordinate $r(x_\pm) = r_m <r_0$. As $x$ goes from $-L_{\text{IR}}/2$ to $L_{\text{IR}}/2$, $\rho$ will oscillate between the turning points, traversing that region one time. Note that the endpoints will lie on different sides of this surface, as shown in Fig.\ \ref{fig:planarlunch}.
 
 The endpoint radius $r_m$ is determined by the total range $x$ coordinate, which is
 \be\label{eq:ellapsedx}
 L_{\text{IR}}= \int_{\rho^-_m}^{\rho^+_m} \dfrac{r_m\text{d}\rho}{r^2\sqrt{-V_{\text{eff}}(\rho)}}\,,
 \ee 
 where $\rho^{\pm}_m$ are the two solutions to $r(\rho) = r_m$ closest to $\rho_b$.  From \eqref{eq:area1} and \eqref{eq:ellapsedx}, it is easy to see that the area of the bulge will satisfy
 \be\label{eq:relation2dbulge} 
 \text{Area}(X^b_\A) - r_m L_{\text{IR}} =  \int_{\rho^-_m}^{\rho^+_m} \text{d}\rho \sqrt{-V_{\text{eff}}(\rho)}\,.
 \ee 
 
 In the thermodynamic limit $L_{\text{IR}}\rightarrow \infty$, the endpoint asymptotes to the value of the radius at the constriction, $r_m \rightarrow r_c= r(\rho_c)$. Then, the left-hand side of \eqref{eq:relation2dbulge} becomes $\text{Area}(X^b_\A) - \text{Area}(X^c_\A)$, which is the difference in areas entering into the exponent of the PLC \eqref{eq:PLC}. Under this identification, \eqref{eq:relation2dbulge} reduces to \eqref{eq:logcomplexityblackbrane}.

 \subsection{Markov length}

 We now consider $\A$ to be a single interval $x\in [0,\ell_{\A}]$ and consider the two boundary-anchored geodesics $X^1_\A$ and $X^2_\A$ homologous to $\A$. Essentially, the same equation \eqref{eq:eomgeodesics} with the effective potential \eqref{eq:eomgeodesicsveff} governs the geodesics, with $r^1_m \geq r_c$ for $X^1_{\A}$ and $r_c> r^2_m\geq r_h$ for $X^2_{\A}$. The allowed values of $r_m$ are implicitly determined as a function of $\ell_{\A}$ by the equation
 \be
 \dfrac{\ell_{\A}}{2}= \int_{\rho_m}^{\rho_\infty} \dfrac{r_m\text{d}\rho}{r^2\sqrt{-V_{\text{eff}}(\rho)}}\,,
 \ee
 where $\rho_m$ is the solution to $r(\rho_m) = r_m$. The area of the surface is then
 \be\label{eq:areaminlengt} 
 \text{Area}(X^i_\A) -r_m^i \ell_{\A} =2  \int_{\rho^i_m}^{\rho_\infty} \text{d}\rho \sqrt{-V_{\text{eff}}(\rho)}\,,
 \ee 
 where recall that the effective potential $V_{\text{eff}}(\rho)$ in \eqref{eq:eomgeodesicsveff} depends on the value of $r^i_m$. The upper bound of the integral in \eqref{eq:areaminlengt} must be suitably regularized, introducing the cutoff $\rho_\infty$, in order to render the area finite.

If the interval is large enough, we have $r_m^1\approx r_c$ and $r_m^2\approx r_h$, and the area difference becomes
 \be
 \text{Area}(X^2_\A) -\text{Area}(X^1_\A) = (r_h-r_c) \ell_{\A} + 2 \ell_{12}\,,
 \ee 
 for the length
 \be 
 \ell_{12} = \int_{\rho_h}^{\rho_\infty} \text{d}\rho \sqrt{1-\frac{r_h^2}{r^2(\rho)}} - \int_{\rho_c}^{\rho_\infty} \text{d}\rho \sqrt{1-\frac{r_c^2}{r^2(\rho)}} \,.
 \ee 
 The difference in areas becomes negative for intervals larger than the Markov length
 \be 
 \ell_{\rm M } = \dfrac{2\ell_{12}}{r_c-r_h}\,.
 \ee 

 For the vacuum solution with a planar thin stress-energy \eqref{eq:geomex}, this yields
 \be 
\dfrac{\ell_{12}}{\ell_{\text{AdS}}} = \log \dfrac{2r_b r_c}{r_h}  +\log \dfrac{\sqrt{r_b^2-r_c^2} + \sqrt{r_b^2-r_h^2}}{r_c^2-r_h^2} - \dfrac{r_h}{r_c} \left(\text{arctanh}\left(\dfrac{r_h}{r_c} \dfrac{\sqrt{r_b^2-r_c^2}}{\sqrt{r_b^2-r_h^2}}\right) + \text{arctanh}\left(\dfrac{r_h}{r_c}\right)\right) \,,
\ee
where we have removed the regulator, sending $\rho_\infty \rightarrow \infty$. For $r_c \gg r_h$, this length scale approaches the length of the wormhole $\ell_{12}\approx \ell_0$.

In Fig. \ref{fig:comp-grid} we numerically verify the hierarchy between the different lengthscales for this example
\be 
\ell \leq \ell_{12} \leq \ell_0\,.
\ee 

\begin{figure}[ht]
  \centering
  \begin{tabular}{cc}
    \includegraphics[width=0.48\textwidth]{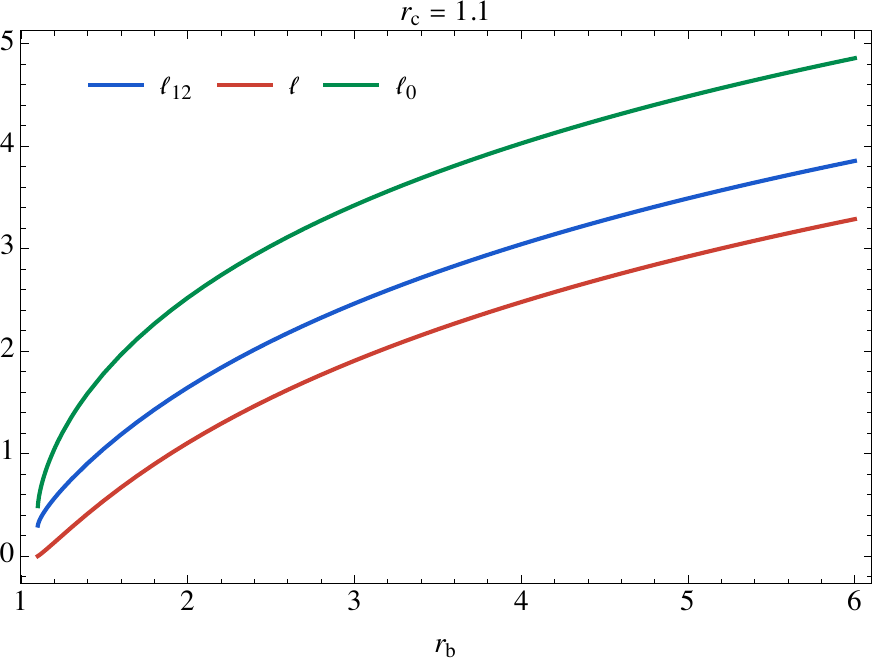} &
    \includegraphics[width=0.48\textwidth]{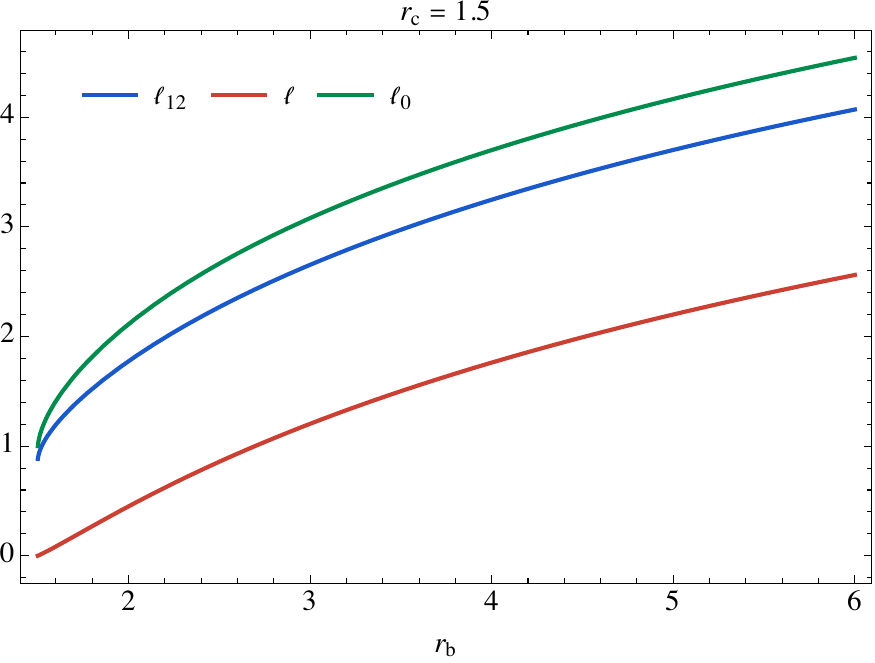} \\
    \includegraphics[width=0.48\textwidth]{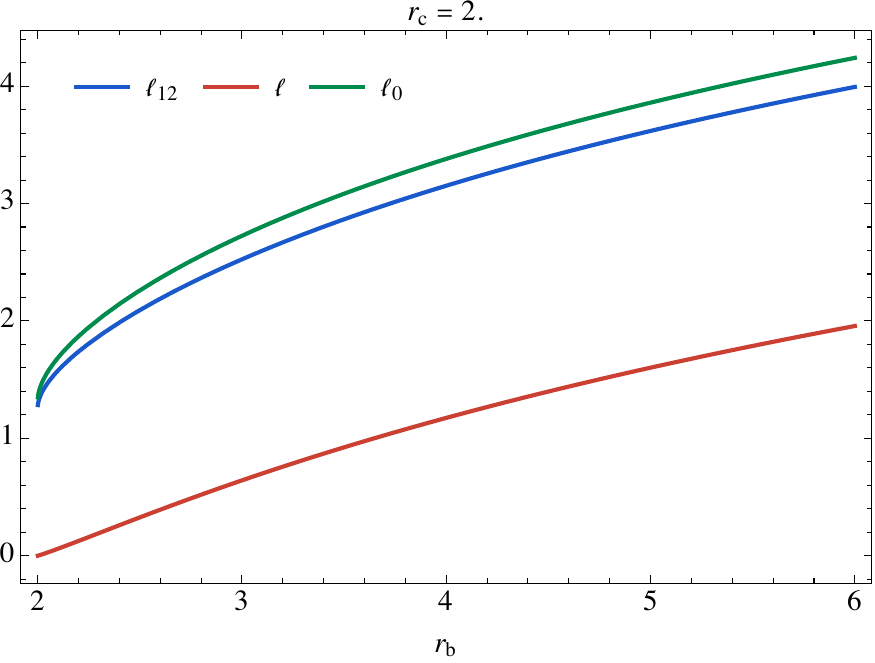} &
    \includegraphics[width=0.48\textwidth]{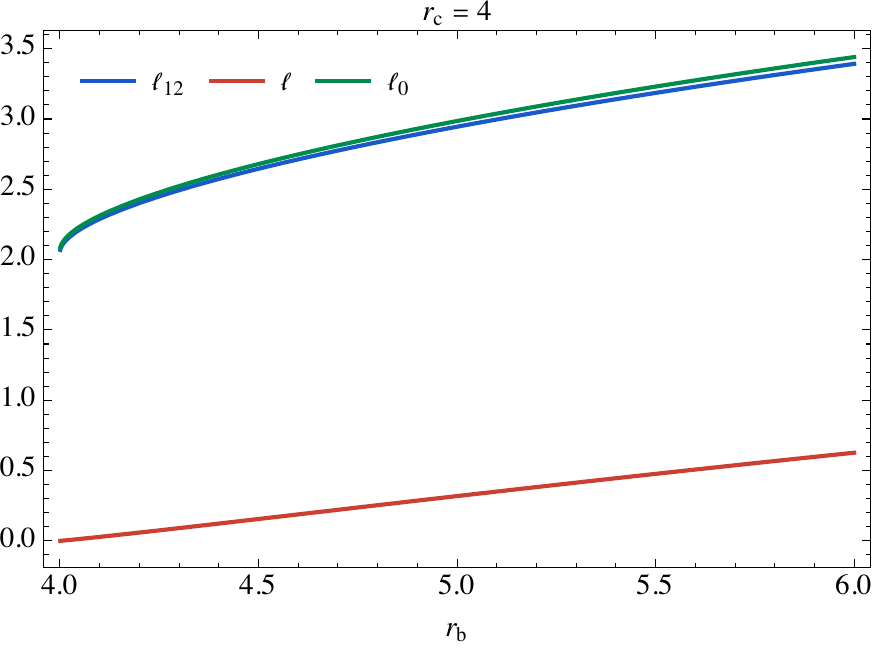}
  \end{tabular}
  \caption{Comparison of different lengthscales for different values of $r_c = \lbrace 1.1, 1.5, 2.0, 4.0\rbrace $ with $r_h =1$.}
  \label{fig:comp-grid}
\end{figure}

\section{Some implications of isoperimetric measure concentration}
\label{app:measureconcentration}

In this appendix we collect sharper finite-$D$ concentration statements for individual Gaussian random maps $R$, in the context of Sec.~\ref{sec:gaussianRT}. We first recall the standard operator-norm concentration result for rectangular Gaussian matrices. We then turn to the normalized Schatten-$1$ distance, which is most naturally treated as a bounded Lipschitz spectral statistic of the Wishart matrix $R^\dagger R$. Finally, we explain the corresponding implication for the circuit complexity of the polar isometry.

\begin{theorem}
Let $R:\mathbb C^{D^n}\to\mathbb C^{D^m}$ be a complex Gaussian random map with i.i.d.\ entries normalized so that $\overline{R^\dagger R}=\mathbf 1_{D^n}$. Write $r=D^{\,n-m}\le 1$. Then for every $t\in(0,1)$,
\be
\Pr\!\left[
\big(\max\{0,1-\sqrt r-t\}\big)^2\le \lambda_{\min}(R^\dagger R)\le
\lambda_{\max}(R^\dagger R)\le (1+\sqrt r+t)^2
\right]
\ge 1-2e^{-cD^m t^2},
\ee
for a universal constant $c>0$. In particular, with at least the same probability,
\be
\big\|R^\dagger R-\mathbf 1_{D^n}\big\|_\infty
\le 2(\sqrt r+t)+(\sqrt r+t)^2.
\ee
\end{theorem}

This is the standard extreme-singular-value bound for rectangular Gaussian matrices \cite{rudelson2009smallest,vershynin2018high}, rewritten in terms of the eigenvalues of the Wishart matrix $R^\dagger R$. The theorem determines the spectral edge of the Wishart matrix $R^\dagger R$, whose natural fluctuation scale is $O(\sqrt r)$.

To study the $1$-norm distance from isometry, it is convenient to first consider normalized spectral averages of bounded Lipschitz test functions.

\begin{theorem}
Let $R:\mathbb C^{D^n}\to\mathbb C^{D^m}$ be as above, and let $f:[0,\infty)\to\mathbb R$ be a bounded Lipschitz function. Then there exist constants $c_f,C_f>0$, depending only on $f$ and independent of $D$, such that for every $\delta>0$,
\be
\Pr\!\left[
\left|
\frac{1}{D^n}\Tr f(R^\dagger R)
-
\mathbb E\!\left(\frac{1}{D^n}\Tr f(R^\dagger R)\right)
\right|
\ge \delta
\right]
\le
C_f e^{-c_f D^{2n}\delta^2}.
\ee
\end{theorem}

This theorem follows from standard concentration results for spectral measures of Wishart matrices \cite{guionnet2000concentration}.

To apply this theorem to the Schatten-$1$ norm, fix $L>2$ and choose a bounded Lipschitz function $f_L$ such that
\be
f_L(x)=|x-1|
\qquad\text{for }x\in[0,L],
\ee
and $f_L$ is constant for sufficiently large $x$. Then $\frac{1}{D^n}\Tr f_L(R^\dagger R)$ is exponentially concentrated around its mean. Moreover, if $L>(1+\sqrt r+t)^2$, the previous theorem implies that with probability at least $1-2e^{-cD^m t^2}$ all eigenvalues of $R^\dagger R$ lie in $[0,L]$, and therefore
\be
\frac{1}{D^n}\Tr f_L(R^\dagger R)
=
\frac{1}{D^n}\big\|R^\dagger R-\mathbf 1_{D^n}\big\|_1.
\ee
It follows that the normalized Schatten-$1$ distance
\be
\frac{1}{D^n}\big\|R^\dagger R-\mathbf 1_{D^n}\big\|_1
\ee
inherits the same concentration, up to an exponentially small bad event.

Its mean is determined by the Marchenko--Pastur law. For fixed $m>n\ge1$, in the large-$D$ limit one has
\be
\mathbb E\!\left[\frac{1}{D^n}\big\|R^\dagger R-\mathbf 1_{D^n}\big\|_1\right]
=
\mu_1(r)+o(\sqrt r),
\ee
with
\be
\mu_1(r)=\int_{\lambda_-}^{\lambda_+}|\lambda-1|\,p(\lambda)\,{\rm d}\lambda.
\ee
The change of variables $\lambda=1+r+\sqrt r\,x$ gives
\be
\mu_1(r)=\frac{\sqrt r}{2\pi}\int_{-2}^{2}
\frac{|x+\sqrt r|\sqrt{4-x^2}}{1+r+\sqrt r\,x}\,dx.
\ee
For small $r$,
\be
\mu_1(r)=\frac{8}{3\pi}\sqrt r+o(\sqrt r),
\ee
and hence
\be
\mathbb E\big\|R^\dagger R-\mathbf 1_{D^n}\big\|_1
=
D^n\mu_1(r)+o(D^n\sqrt r)
\sim \frac{8}{3\pi}D^{(3n-m)/2}.
\ee
Thus the trace-norm deviation vanishes for $m>3n$, in agreement with Theorem~\ref{thrm:gaussianrandomtensor}.

We now turn to the circuit complexity of the polar isometry. Let
\be
R=T_R(R^\dagger R)^{1/2}
\ee
be the polar decomposition of $R$, so that $T_R:\mathcal H_A\hookrightarrow\mathcal H_B$ is an isometry with $\dim\mathcal H_A=D^n$ and $\dim\mathcal H_B=D^m$. Since the polar part of a Gaussian random matrix is Haar distributed on the complex Stiefel manifold of isometries from $\mathcal H_A$ to $\mathcal H_B$, standard $\varepsilon$-net counting arguments imply that a typical such isometry has maximal circuit complexity up to constants.

\begin{theorem}
Let $R:\mathcal H_A\to\mathcal H_B$ be a Gaussian random map with
$\dim\mathcal H_A=D^n$, $\dim\mathcal H_B=D^m$, and $m>2n$, and let $T_R$ be its polar isometry. Fix a universal finite gate set and an accuracy $\varepsilon\in(0,\varepsilon_0]$, where $\varepsilon_0$ depends only on the gate set. Then, with probability tending to one as $D\to\infty$, any circuit built from $K$ two-local gates on qudits of local dimension $D$ that prepares an isometry $T$ satisfying
\be
\|T-T_R\|_2\le \varepsilon
\ee
must obey
\be
K\ge c_\varepsilon D^{\,n+m},
\ee
for some constant $c_\varepsilon>0$ depending only on $\varepsilon$ and the gate set. Conversely, standard results for general isometries imply the upper bound
\be
K\le C_\varepsilon D^{\,n+m},
\ee
for some constant $C_\varepsilon>0$. Hence
\be
\mathcal C(T_R)=\Theta(D^{\,n+m}).
\ee
\end{theorem}

The lower bound follows from the fact that the manifold of isometries $\mathcal H_A\to\mathcal H_B$ has real dimension $2D^{n+m}-D^{2n}$, whereas the number of distinct circuits of length $K$ over a fixed finite gate set grows only exponentially in $K$ (see Appendix \ref{app:overestimate}). For fixed $\varepsilon$, the union of all $\varepsilon$-balls around circuits of size $K\ll D^{n+m}$ therefore covers only an exponentially small fraction of the Stiefel manifold. Since $T_R$ is Haar random on this manifold, it follows that approximating $T_R$ typically requires $K=\Omega(D^{n+m})$.

\section{Mean curvature flow and everywhere locally expanding foliation}
\label{app:contraction}

In this appendix we prove the following theorem:

\begin{theorem}\label{thm:foliation}
Let $\Sigma$ be a compact, connected, orientable Riemannian manifold with boundary $\partial\Sigma=X^b\cup X^c$, where $X^b$, $X^c$ are mutually disjoint, closed, and extremal, and where the interior of $\Sigma$ does not contain any closed extremal hypersurfaces. Suppose also that $X^b$ is connected and has non-zero Morse index. Let $\mathcal{C}$ be the (common) homology class of $X^b,X^c$. Then there exists a foliation of $\Sigma$ by hypersurfaces $X_t\in\mathcal{C}$, interpolating between $X^b$ and $X^c$, with extrinsic curvature obeying $K\le0$ everywhere, where $K$ is defined in the direction of increasing $t$ (in other words, the foliation is everywhere locally contracting).
\end{theorem}

In the context of the PLC, the theorem establishes the existence of an everywhere locally contracting (or, in the other direction, everywhere locally expanding\footnote{We state the theorem in terms of a contracting foliation, since mean curvature flow goes in that direction. On the other hand, in our physics application, we naturally think about going in the expanding direction, from the constriction to the bulge.}) foliation of the region between the bulge and constriction. The proof of the theorem will be at a physicists' level of rigor (in other words, it will not be rigorous at all). Before getting to the proof, we recall a few general definitions and facts about hypersurfaces in Riemannian manifolds and mean curvature flow, which will be the main tool in the proof. (See \cite{mantegazza2011lecture} for an overview of mean curvature flow.)
 
We fix a compact, connected, orientable $d+1$-dimensional Riemannian manifold $\Sigma$ (possibly with boundary) with metric $g_{\mu\nu}$. We will use $y$ to denote a point in $\Sigma$, and $y^\mu$ its coordinates in a given coordinate system. By \emph{surface} we mean a compact, oriented (not necessarily connected) $d$-dimensional submanifold (without boundary) of $\Sigma$. The orientation can be specified by choosing a unit normal vector field $n^\mu$, which determines the sign of the extrinsic curvature $K_{\mu\nu}$. We will use $x$ to denote a point on a surface, $x^i$ its coordinates in a given coordinate system, and $\gamma_{ij}$ the induced metric.

Given a function $\phi$ on $X$, deforming $X$ at each point $x$ by a distance $\phi(x)$ in the direction $n^\mu(x)$ changes its area to first order as follows:
\be
\delta|X| = \int_X\gamma^{1/2}\,K\phi\,.
\ee
The first-order change in $K$ is governed by the Jacobi operator $J_X$,
\be\label{deltaK}
\delta K = J_X\phi\,,
\ee
defined as follows:
\be
J_X:=-\nabla^2-R_{\mu\nu}n^\mu n^\nu-K_{\mu\nu}K^{\mu\nu}\,,
\ee
where $\nabla^2$ is the scalar Laplacian with respect to $\gamma_{ij}$ and $R_{\mu\nu}$ is the Ricci tensor of $\Sigma$. $J_X$ is a self-adjoint operator on the space of functions on $X$. Furthermore, since the functions appearing in the operator are real, its eigenfunctions can be chosen to be real. If $K=0$ everywhere on $X$, then $X$ is extremal (or minimal, in mathematicians' parlance). The Morse index of $X$ is then the number of negative eigenvalues of $J_X$. Note that the extrinsic curvature, Jacobi operator, and Morse index are defined even if $X$ is on the boundary of $\Sigma$.

Given a continuous 1-parameter family of surfaces $\{X_t\}$ defined by $y=y(x,t)$, we define the normal velocity by
\be\label{vdef}
v:=n_\mu\partial_ty^\mu\,.
\ee
The \emph{mean curvature flow} (MCF) equation is
\be\label{MCFdef}
v=-K\,.
\ee
MCF is the gradient flow of the area functional with respect to the natural $L^2$ norm $\int_X\gamma^{1/2}\phi^2$ on functions on $X$. A MCF is a homotopy and therefore, in particular, stays within a homology class. The fixed points of MCF are the extremal surfaces. Note that MCF is insensitive to the choice of orientation; under $n^\mu\to -n^\mu$, $K\to-K$, so $v\to-v$, so $\partial_ty^\mu$ is invariant.

MCF has a parabolic or diffusive character, which we can see by linearizing about a solution. Specifically, perturbing a solution $\{X_t\}$ by a normal deformation $\phi(x,t)$, i.e.\ $y^\mu(x,t)\to y^\mu(x,t)+\phi(x,t)n^\mu(x,t)$, the linearized MCF equation is:
\be
\partial_t\phi = -J_{X_t}\phi\,,
\ee
which is a parabolic PDE. The diffusive nature of MCF guarantees the uniqueness and short-time existence of a solution: given $t_0\in\mathbf{R}$ and a hypersurface $X'$ in the interior of $\Sigma$, there is a $t_1>t_0$ and a solution $\{X_t\}$ on the interval $[t_0,t_1]$ with $X_{t_0}=X'$, and this solution is unique.

A priori, there are a few obstacles that could prevent the existence of a solution for arbitrarily long times: a singularity could develop, where the surface would cease to be smooth; it could hit the boundary of $\Sigma$; or it could self-intersect and thereby cease to be embedded. We will address each of these possibilities in turn. Again, the arguments will be somewhat informal, and we will assume a generic situation.

First, while locally MCF tends to smooth out wrinkles and bumps, it can happen that part of the surface can shrink to zero size in finite time. Generically, this will only happen with spheres. The sphere could be $d$-dimensional (the same as $X_t$), or lower-dimensional, for example, a neck. To extend the flow, it can be stopped just before the singularity forms, then continued after a local topology-changing surgery. For example, a collapsing $S^d$ can simply be removed from the hypersurface. A neck that pinches off along an $S^{d-1}$ can be excised and replaced by two balls. More generally, the surgery replaces $B^p\times S^{d-p}$ by $S^{p-1}\times B^{d-p+1}$ (where $0\le p\le d$, and $B^p$ is the $p$-dimensional ball), which is possible since both have boundary $S^{p-1}\times S^{d-p}$. The hypersurfaces before and after the surgery together bound a small part of the ambient space $\Sigma$, so while they are not homotopic, they are homologous. Furthermore, the surgery does not change the surface's area. The general idea is that, within the homology class, different homotopy classes meet along walls where the surface becomes singular, and we simply continue the gradient descent through the wall.

Next, consider the possibility that the surface hits the boundary $\partial\Sigma$. As long as the boundary is mean-curvature convex (meaning that, with $n^\mu$ chosen outward-directed, $K_{\partial\Sigma}\ge0$), MCF will not drive the surface across the boundary. Proof: Suppose at some time $t$, $X_t$ is tangent to $\partial\Sigma$ at some point $x$. Choose the unit normals to be outward-directed for both $X_t$ and $\partial\Sigma$. At $x$, $K_{X_t}\ge K_{\partial\Sigma}$ (otherwise $X_t$ would not be contained in $\Sigma$); hence $K_{X_t}\ge0$ and $v\le0$, i.e., the surface will not be driven across $\partial\Sigma$, and the flow will continue to be well-defined.

Similarly, the surface cannot intersect itself under MCF. Proof: Assume that, at some time $t$, two points on $X_t$ are very close. For $X_t$ not to already intersect itself, the neighborhoods of the points must be parallel, and the surfaces must curve away from each other. This implies that the values of $K$ at the two points are such that, under MCF, they move away from each other (or at least do not get closer).

All in all, we may assume that the MCF can be extended infinitely in the positive time direction. Since $\Sigma$ is compact, as $t\to\infty$, $X_t$ must reach a limiting surface $X_\infty$, which must be extremal. (The other possibility, oscillating forever, is excluded by the diffusive nature of MCF.) If the initial surface is null-homologous, then the surface may disappear entirely; i.e., the limiting extremal surface may be the empty set.

Another important feature of MCF is that the surface cannot ``back up'': if at some time $t=t_0$, $v\ge0$ (equivalently, $K\le0$) everywhere on $X_t$, then for all $t>t_0$, $X_t$ has the same property. Proof: Look at the time derivative of $K$, $\partial_tK=-J_{X_t}K$ (obtained from combining \eqref{deltaK} and \eqref{MCFdef}). Suppose that, for some $t$, $K\le0$ everywhere on $X_t$, and at some point $x\in X_t$, $K=0$. Then, at $x$, $-J_{X_t}K=\nabla^2K$, which is non-positive since $K\le0$ in a neighborhood of $x$. Hence, at $x$, $\partial_t K\le0$, so $K$ cannot become positive.

The behavior of MCF in the vicinity of an extremal surface is controlled by its Jacobi operator. Let $X'$ be an extremal surface, and $y^\mu=(x^i,z)$ be  Gauss normal coordinates in a neighborhood of $X'$, with $z$ the coordinate in the direction $n^\mu$. For a family $X_t$ of surfaces described by $z=\phi(x,t)$, the MCF equation, to linear order in $\phi$, is $\partial_t\phi=-J_{X'}\phi$. Therefore, a generic MCF close to $X'$ is controlled by the smallest eigenvalue $\lambda_{\rm min}$ of $J_{X'}$. If $\lambda_{\rm min}<0$, then for any real $\alpha$, there is an MCF moving away from $X'$, which is approximately given by $\phi(x,t)=\alpha\psi(x)e^{-\lambda_{\rm min} t}$ for $t\ll0$, where $\psi(x)$ is the eigenfunction corresponding to $\lambda_{\rm min}$. On the other hand, if $\lambda_{\rm min}>0$, then a generic MCF in the vicinity of $X'$ moves toward it, and for $t\gg0$ is approximately given by $\phi(x,t)=\alpha\psi(x)e^{-\lambda_{\rm min} t}$ for some real $\alpha$. This assumes that $X'$ is in the interior of $\Sigma$, so it can be perturbed in either direction; otherwise, one or both signs of $\alpha$ may be disallowed.

We are finally ready to prove the theorem.

\begin{proof}
Choose $n^\mu$ on $X^b$ to point into $\Sigma$. Since $X^b$ has non-zero Morse index, $\lambda_{\rm min}<0$. Since $J_{X^b}$ is a Schr\"odinger-type operator and $X^b$ is connected, the corresponding eigenfunction $\psi$ can be chosen to be everywhere positive. Let $\{X_t\}$ be the MCF seeded by this eigenfunction, i.e., which for $t\ll0$ is approximately given by $e^{-\lambda_{\rm min}t}\psi(x)$. Extend the MCF to $t=+\infty$ by surgeries if necessary. For all $t$, $X_t\in\mathcal{C}$. For $t\ll0$, $v>0$ and $K<0$ everywhere on $X_t$; therefore, for all $t$, $v\ge0$ and $K\le0$ everywhere on $X_t$. Since $\{X_t\}$ limits as $t\to\infty$ to an extremal surface, and does not back up, it cannot limit to $X^b$; it must limit to $X^c$, as there are no other extremal surfaces in $\Sigma$. Therefore, as $t$ goes from $-\infty$ to $\infty$, $X_t$ moves monotonically from $X^b$ to $X^c$, and therefore $\{X_t\}$ foliates $\Sigma$. (If the last statement is not obvious, consider the following argument: Let $\tilde\Sigma_t:=\bigcup_{t'\le t}X_t$ be the part of $\Sigma$ that is foliated by all the surfaces $X_{t'}$ up to time $t$. We will build up $\tilde\Sigma_t$ step by step and show that, as $t\to\infty$, it covers all of $\Sigma$. For $t\ll0$, $\tilde\Sigma_t$ is a thin layer about $X^b$, with boundary $\partial\tilde\Sigma_t=X^b\cup X_t$. For any $t$ and any small $\delta t$, the difference between $\tilde\Sigma_{t-\delta t}$ and $\tilde\Sigma_t$ is a thin layer about $\tilde\Sigma_{t-\delta t}$, which moves the boundary from $X_{t-\delta t}$ to $X_t$. Hence, for any $t$, $\partial\tilde\Sigma_t=X^b\cup X_t$. Taking the limit $t\to\infty$, we have $\partial\tilde\Sigma_\infty=X^b\cup X^c$. Since $\Sigma$ is connected, the only subset of $\Sigma$ with boundary $X^b\cup X^c$ is itself; hence $\tilde\Sigma_\infty=\Sigma$.)
\end{proof}

The theorem does not assume anything about the Morse index of $X^c$. Since the MCF in the proof limits to $X^c$, we would expect generically that it must be stable and therefore have a vanishing Morse index. In fact, we can show that this is true as a corollary of the theorem without appealing to genericity.

\begin{corollary}
Under the assumptions of Theorem \ref{thm:foliation}, $X^c$ has vanishing Morse index.
\end{corollary}
\begin{proof}
Note that the theorem implies $|X^b|>|X^c|$. Assume first that $X^c$ is connected. If $X^c$ has non-zero Morse index, we can apply the theorem with $X^b$ and $X^c$ reversed, implying $|X^c|>|X^b|$, a contradiction. Now suppose $X^c$ is not connected. If its Morse index is non-zero, then at least one connected component has a non-zero Morse index. Call that component $X^{b\prime}$ and the rest of the components $X^{cr}$, $X^c=X^{b\prime}\cup X^{cr}$, and define $X^{c\prime}:=X^b\cup X^{cr}$. Now we can run the theorem with $X^{b\prime},X^{c\prime}$ in the roles of $X^b,X^c$ respectively, implying $|X^{b\prime}|>|X^{c\prime}|=|X^b|+|X^{cr}|$, again a contradiction with $|X^c|<|X^b|$.
\end{proof}

\bibliographystyle{ourbst}
\bibliography{bibliography}

\end{document}